\documentclass[acmsmall,screen]{acmart}
\usepackage{mathpartir}
\usepackage{tikz-cd}
\usepackage{enumitem}
\usepackage{wrapfig}
\usepackage{fancyvrb}

\newcommand{\bnfalt}{\mathrel{\bf \,\mid\,}}
\newcommand{\alt}{\bnfalt}

\newcommand{\kw}[1]{\texttt{#1}\,\,}
\newcommand{\binkw}[1]{\,\texttt{#1}\,\,}
\newcommand{\boolty}{\texttt{bool}}
\newcommand{\dyn}{{?}}
\newcommand{\raiseOpwithM}[2]{\kw{raise}#1(#2)}

\newcommand{\effarr}[3]{#1 : #2 \leadsto #3}
\newcommand{\effdecl}[3]{#1 : #2 \leadsto #3}

\newcommand{\letXbeboundtoYinZ}[3]{\kw{let}#2 = #1\binkw{in} #3}

\newcommand{\defXtobeY}[2]{\kw{define}{#1} = {#2}}
\newcommand{\neweffEofReqtoResp}[3]{\kw{effect}\effdecl{#1}{#2}{#3}}

\newcommand{\importvalXasYatZ}[3]{\kw{import-val} #1\kw{ as} #2 \mathop{@} #3}
\newcommand{\importeffXatYtoZ}[3]{\kw{import-eff} #1\mathop{@} #2 \leadsto #3}

\newcommand{\pipe}{\,\,|\,\,}
\newcommand{\hole}{\bullet}

\newcommand{\ltdyn}{\sqsubseteq}

\newcommand{\gtdyn}{\sqsupseteq}
\newcommand{\equidyn}{\mathrel{\gtdyn\ltdyn}}
\newcommand{\subty}{\leq}

\newcommand{\gsubty}{\lesssim}

\newcommand{\vsubty}{\rotatebox[origin=c]{90}{$\subty$}}

\newcommand{\tru}{\texttt{true}}
\newcommand{\fls}{\texttt{false}}

\newcommand{\ifXthenYelseZ}[3]{\kw{if} #1 \{ #2 \}\{ #3 \}}

\newcommand{\gjoin}{\mathop{\stackrel{\sim}{{\vee}}}}
\newcommand{\gmeet}{\mathop{\stackrel{\sim}{{\wedge}}}}

\newcommand{\sem}[1]{\llbracket#1\rrbracket}
\newcommand{\sig}{\Sigma}
\newcommand{\sg}{\sig\pipe\Gamma}

\newcommand{\uarrowl}{\mathrel{\rotatebox[origin=c]{-30}{$\leftarrowtail$}}}

\newcommand{\darrowl}{\mathrel{\rotatebox[origin=c]{30}{$\twoheadleftarrow$}}}

\newcommand{\upcast}[2]{\langle{#2}\uarrowl{#1}\rangle}
\newcommand{\dncast}[2]{\langle{#1}\darrowl{#2}\rangle}
\newcommand{\obliqueCast}[2]{\langle{#1}\Leftarrow{#2}\rangle}
\newcommand{\apart}{\#}

\newcommand{\Erel}[1]{\mathcal{E}\sem{#1}}

\newcommand{\Vrel}[1]{\mathcal{V}\sem{#1}}
\newcommand{\Rrel}[1]{\mathcal{R}\sem{#1}}
\newcommand{\Krel}[1]{\mathcal{K}\sem{#1}}

\newcommand{\pivrel}[3]{\mathcal{V}^{#1}_{#3}\sem{#2}}
\newcommand{\pierel}[3]{\mathcal{E}^{#1}_{#3}\sem{#2}}
\newcommand{\pigrel}[3]{\mathcal{G}^{#1}_{#3}\sem{#2}}
\newcommand{\pirrel}[3]{\mathcal{R}^{#1}_{#3}\sem{#2}}
\newcommand{\pikrel}[3]{\mathcal{K}^{#1}_{#3}\sem{#2}}

\newcommand{\precltdyn}{\mathrel{\preceq}}
\newcommand{\precgtdyn}{\mathrel{\succeq}}

\newcommand{\ltivrel}{\pivrel{\precltdyn}}
\newcommand{\gtivrel}{\pivrel{\precgtdyn}}
\newcommand{\simivrel}{\pivrel{\sim}}

\newcommand{\ltierel}{\pierel{\precltdyn}}
\newcommand{\gtierel}{\pierel{\precgtdyn}}
\newcommand{\simierel}{\pierel{\sim}}

\newcommand{\simigrel}{\pigrel{\sim}}

\newcommand{\ltirrel}{\pirrel{\precltdyn}}
\newcommand{\gtirrel}{\pirrel{\precgtdyn}}
\newcommand{\simirrel}{\pirrel{\sim}}

\newcommand{\gtikrel}{\pikrel{\precgtdyn}}
\newcommand{\simikrel}{\pikrel{\sim}}

\newcommand{\later}{{\blacktriangleright}}
\newcommand{\stepstar}{\mathrel{\mapsto^*}}
\newcommand{\stepsin}[1]{\mathrel{\mapsto^{#1}}}
\newcommand{\err}{\mho}

\newcommand{\sgd}{\sig\pipe\Gamma\pipe\Delta}

\newcommand{\TmhastySGRhoMT}[5]{{#1}\pipe{#2} \vdash_{#3} {#4} : {#5}}

\newcommand{\TmhastyRhoMT}[3]{\sg \vdash_{#1} {#2} : {#3}}
\newcommand{\Tmhasty}[2]{\TmhastyRhoMT{\sigma}{#1}{#2}}

\newcommand{\hastyDRhoMT}[4]{\sg\pipe{#1} \vdash_{#2} {#3} : {#4}}
\newcommand{\hastyGDRhoMT}[5]{\sig \pipe {#1} \pipe{#2} \vdash_{#3} {#4} : {#5}}
\newcommand{\hastyRhoMT}[3]{\sgd \vdash_{#1} {#2} : {#3}}
\newcommand{\hasty}[2]{\hastyRhoMT{\sigma}{#1}{#2}}

\newcommand{\RbngT}[2]{({#1}\, !\, {#2})}
\newcommand{\holeRhoT}[2]{\bullet : \RbngT {#1} {#2}}

\newcommand{\hastySGRhoMT}[5]{{#1}\pipe{#2} \vdash_{#3} {#4} : {#5}}

\newcommand{\valatomlhs}[1]{\text{VAtom}\, {#1}}
\newcommand{\termatomlhs}[3]{\text{TAtom}\, {#1}\, {#2}\, {#3}}
\newcommand{\ecatomlhs}[3]{\text{ECtxAtom}\, {#1}\, {#2}\, {#3}}

\newcommand{\inj}[2]{\texttt{Inj}({#1}, {#2})}

\newcommand{\expandTermPrecisionDef}[5]
  {({#1}, {#2}) \in \simierel {#3} {#5} {\simivrel {#4} {}}}

\newcommand{\elabarr}{\Rightarrow}

\newcommand{\elabty}[3]{#1\vdash #2 \elabarr #3}
\newcommand{\elabtm}[5]{#1\vdash #2 \elabarr #3 : \compty {#4}{#5}}
\newcommand{\elabmod}[7]{{#1}\pipe{#2}\pipe {#3} \vdash {#4} \elabarr {#5};{#6};{#7}}
\newcommand{\elabprog}[7]{#1\pipe #2 \vdash #3 \elabarr {#4}\vdash_{#6} #5 : #7}
\newcommand{\elabHandleType}[5]{#1 \vdash \textrm{handleTy}(#2, #3, #4) = {#5}}
\newcommand{\elabdecl}[7]{#1\pipe#2\pipe#3\vdash #4 \elabarr #5 ; #6 ; {#7}}

\newcommand{\effname}{\varepsilon}

\newcommand{\compty}[2]{{#1}\, !\, {#2}}
\newcommand{\hndl}[4]{{\kw{handle}#1\,\{ \kw{ret}#2. #3 \pipe #4\}}}
\newcommand{\surfHndl}[5]{{\texttt{handle}_{#1}\;{#2}\,\{ \kw{ret}#3. #4 \pipe #5\}}}
\newcommand{\dom}{\textrm{dom}}
\newcommand{\effto}[1]{\mathrel{\to_{#1}}}

\newcommand{\parfin}{\rightharpoonup_{\textrm{fin}}}

\setcopyright{acmcopyright}
\copyrightyear{2018}
\acmYear{2018}
\acmDOI{XXXXXXX.XXXXXXX}

\acmConference[Woodstock '18]{Woodstock '18: ACM Symposium on Neural
  Gaze Detection}{June 03--05, 2018}{Woodstock, NY}
\acmBooktitle{Woodstock '18: ACM Symposium on Neural Gaze Detection,
  June 03--05, 2018, Woodstock, NY}
\acmPrice{15.00}
\acmISBN{978-1-4503-XXXX-X/18/06}

\begin{document}


\title{Gradual Typing for Effect Handlers}
\author{Max S. New}
\affiliation{
  \department{Computer Science and Engineering}              
  \institution{University of Michigan}            
  \country{USA}                    
}
\email{maxsnew@umich.edu}          
\author{Eric Giovannini}
\affiliation{
  \department{Computer Science and Engineering}              
  \institution{University of Michigan}            
  \country{USA}                    
}
\email{ericgio@umich.edu}
\author{Daniel R. Licata}
\affiliation{
  \department{Mathematics and Computer Science}
  \institution{Wesleyan University}
  \country{USA}
}
\email{dlicata@wesleyan.edu}

%

\begin{abstract}
  We present a gradually typed language, GrEff, with effects and
  handlers that supports migration from unchecked to checked effect
  typing. This serves as a simple model of the integration of an
  effect typing discipline with an existing effectful typed language
  that does not track fine-grained effect information. Our language
  supports a simple module system to model the programming model of
  gradual migration from unchecked to checked effect typing in the
  style of Typed Racket.

  The surface language GrEff is given semantics by elaboration to a
  core language Core GrEff. We equip Core GrEff with an inequational
  theory for reasoning about the semantic error ordering and desired
  program equivalences for programming with effects and handlers. We
  derive an operational semantics for the language from the equations
  provable in the theory. We then show that the theory is sound by
  constructing an operational logical relations model to prove the
  graduality theorem. This extends prior work on embedding-projection
  pair models of gradual typing to handle effect typing and subtyping.
\end{abstract}

\maketitle

\hypertarget{introduction}{%
\section{Introduction}\label{introduction}}

Gradually typed programming languages are designed to support smooth
migration from a lax to a strict static type
discipline \cite{tobin-hochstadt08, siek-taha06}. Most commonly, gradually
typed languages add a static type system to an existing dynamically
typed language and allow for (1) safe interoperability between the
languages and (2) semantic guarantees that adding types to existing
programs only results in stricter type enforcement, and no other
behavioral change. More generally, gradual typing has been
applied to provide a spectrum of precision in other kinds of typing
disciplines such as refinement typing or effect
typing \cite{lehmann17,gradeffects2014}, where the ``dynamic''
side is a statically typed language itself.

One particular presentation of effects and effect typing that is
gaining popularity is \textbf{effect
handlers} \cite{DBLP:conf/esop/PlotkinP09}. Operationally, effect
handlers are \textbf{resumable exceptions}, code can "raise" an effect
operation, which will then be handled by the closest enclosing
handler, which in addition to the exception data will also receive the
continuation for the raising code that can be invoked to resume at the
original point where the effect was raised. Effect handlers provide an
intuitive typed interface to delimited continuations, and can
similarly be used to conveniently implement backtracking search,
non-determinism, mutable state, and as a convenient interface to
external system calls.
Effect handlers have been implemented in a number of libraries and
experimental languages, and more recently have been incorporated as a
built-in feature into OCaml 5, and have been proposed as an extension
to WASM
\cite{DBLP:journals/corr/Leijen14,DBLP:conf/haskell/KiselyovSS13,DBLP:conf/popl/LindleyMM17,DBLP:conf/fmco/CooperLWY06,DBLP:journals/jfp/BrachthauserSO20,DBLP:conf/pldi/Sivaramakrishnan21,wasmfx}.

Designers of languages supporting effect handlers, much like designers
of languages with exceptions, are left with a choice of whether the
type system should merely validate that the input and output types of
effect operations are respected, or if an \emph{effect typing} system
should be employed to determine that a particular effect can only be
raised when the context is known to implement a handler for it.  On
the one hand, checked effects allow programmers to easily reason about
which effects can be raised by subprocedures and ensure they are
handled appropriately, rather than being caught by the runtime system
and causing the program to crash.
On the other hand, strict checking may necessitate large code changes
when code is extended to raise new operations, and even in languages
such as Java that support both checked and unchecked exceptions,
unchecked exceptions are preferred in many scenarios.
Furthermore, when adding effect typing to a language that does not
already support it, even correct existing libraries may not typically
pass the necessarily conservative static type checker. It may be
infeasible to rewrite large amounts of existing library code to
precisely track effect usage.
Gradual typing provides a linguistic framework for designing languages
where a programmer is not entirely locked in to one system or another:
they might use unchecked exceptions in one module and checked
exceptions in another, while supporting well-defined interoperability
with useful error messages at runtime if there is an effect raised in
a context where it is not expected. Further, a gradually typed
language provides a path for gradually \emph{migrating} code from
less precise to more precise static type checking.
This potential for gradual typing to be used in this way to
incorporate effect typing disciplines into existing languages has been
eloquently pdiscussed in prior work by Phil Wadler \cite{wadlerGATE}.

In this work we present the design and semantics of GrEff, a gradual
language with effect handlers that supports gradual migration from
unchecked effects to precise effect typing.
The untracked sublanguage of GrEff is designed to be similar to SML
and Java's treatment of exceptions: new effect operations are declared
with specified input and output types, and these can be imported and
used to raise and handle those operations in other modules, but which
effects are raised by a function is not tracked by the type system.
In addition, GrEff supports \emph{tracked} function types $A \effto \sigma
B$ where the input values must be of type $A$, output values will be
of type $B$, and the function may raise any of and only the effects in the
set $\sigma$.
The untracked function type is modeled then as a type $A \to_{\dyn} B$
which has a ``dynamic'' effect type, in the sense that it may raise
any effect, possibly including unknown effect operations declared in
some independent module of the program.
Since our main focus in this work is on providing a foundation for
extending existing statically typed languages such as OCaml 5 with
effect types, we have chosen not to support full dynamic typing in the
design of GrEff. However, the design should easily accomodate
supporting fully dynamic value typing in addition to the dynamic
effect typing using standard gradual typing techniques.
We note that 

In GrEff, new effect operations can be declared in each module, just
as new exceptions can be declared in Java and ML-style languages.
When an effect is declared in a module, it is given an
associated \emph{request} and \emph{response} type.
For instance, an effect for reading a boolean state would
be \verb+get : Unit ~> Bool+, the user provides a trivial value
as the request and receives a boolean value as the
response, while an effect for writing to boolean state would
be \verb+set : Bool ~> Unit+.
Similar to ML and Java, GrEff takes a \emph{nominal} approach to
effect operations: each effect operation has an associated request and
response type that are used to determine when an effect is properly
raised or handled.
However, having a single, global assignment from effect names to
request/response types is problematic from the perspective
of \emph{gradual} migration from untracked to tracked effects.
In a completely nominal form of effect typing, if an effect operation
is used in many different modules with imprecise typing, and one
module is migrated to use a more precise version of the effect's
request/response type, then we would need to migrate all modules to
use the more precise type.
Instead, gradual migration should allow for this to be done a single
module at a time.
To achieve this, in GrEff, we take a \emph{locally nominal} but
\emph{globally structural} approach to the typing of effect operations.
That is, \emph{locally}, within each module, the request and response
type for an effect are fixed, and all \texttt{raise}
and \texttt{handle} constructs are checked with the same typing.
On the other hand, \emph{globally}, different modules across the
program can associate different types to the same effect operation.
At module boundaries, i.e., imports and exports, modules are
statically allowed to interoperate if they agree on the precisely
typed portion of the effects they share. If one module is more precise
than the other, then dynamic runtime monitoring is inserted in the
implementation to ensure that the runtime behavior agrees with the
static typing, raising an error if the dynamically typed code violates
the imposed runtime type discipline.

There are two aspects in designing a \emph{sound} gradually typed
language: designing the syntax and gradual type checking of the
surface language and designing the corresponding core language and
semantics. The syntax should support a simple process for migrating
from an imprecise to a precise style, satisfying the
\emph{static gradual guarantee} \cite{refined}.
We designed the surface language with the goal of modeling program
migration from static to dynamic typing. For this reason we include a
simple module system in the style of Typed
Racket \cite{tobin-hochstadt08} so that we can express that different
portions of the program have different views on how the effect
operations are typed. Once the base language is designed, the gradual
type checking is based on prior work on defining gradual type systems
that satisfy the static gradual guarantee\cite{siek-taha06,AGT}.

Next, the core language provides a definition for
the runtime semantics. The semantics should admit useful type-based
reasoning principles for precisely typed code, even in the presence of
interaction with imprecisely typed components. Further, the
aforementioned migration process should have a predictable impact on
program semantics: migrating from to more precise checking may result
in new errors being identified (statically or dynamically), but
otherwise should not impact program behavior, a property known as
the \emph{dynamic gradual guarantee}
or \emph{graduality} \cite{refined,newahmed18}.
To design the core language and runtime semantics, we follow the prior
work (\cite{newahmed18,newlicataahmed19}) which established a recipe
for designing a new gradual core language to satisfy the graduality
theorem and validate strong type-based equational reasoning
principles.
Their approach is to \emph{axiomatize} the type-based reasoning
principles as equations and the graduality theorem as inequalities,
where casts are defined not by specifying their operational
behavior \emph{a priori} but instead by assuming they are given by
least upper bounds/greatest lower bounds.
Then the operational behavior of the casts can be \emph{derived} from
the inequational theory.
An operational or denotational model must then be constructed to prove
the theory is consistent, which implies the graduality theorem.
But since the operational semantics is \emph{derived} from the
inequational theory, this also establishes a stronger theorem that the
observable behavior of the casts is \emph{uniquely determined} by the
desired type-based reasoning and graduality, showing that any
observably different cast semantics must violate one or more of the
axioms.

For designing our core language, called Core GrEff, we extend this
recipe, which previously has only been demonstrated on simple types, to
apply also to \emph{effect casts} and \emph{subtyping} of value and effect
types.
We then show that every rule of an operational semantics is derivable
from the least upper bound/greatest lower bound specifications of
casts as well as congruence rules and an \emph{effect forwarding}
principle for handlers.
The effect forwarding principle states that a handler clause that
simply re-raises the effect it handles with the same continuation can
be removed without changing the observable behavior of the system, an
intuitive principle as well as a highly desirable compiler
optimization.

In this work, we extend prior step-indexed logical relations models
for proving graduality to handle effects and subtyping, by showing
that the runtime casts satisfy the properties of
being \emph{embedding-projection pairs} \cite{newahmed18}. In doing
so, we show how to combine effect and value embedding-projection pairs
within the same system, and how they interact. Additionally, we
identify new semantic principles for the interaction between subtyping
and runtime casts.

The contributions of the paper are as follows:
\begin{enumerate}
\item
  We define a gradually typed language GrEff supporting migration from
  unchecked to checked effects and handlers.

  
\item
  We prove this language satisfies the static gradual guarantee and the
  dynamic gradual guarantee (graduality).

\item
  We give the language a semantics by elaboration into a core
  language, core GrEff.

\item
  We axiomatize the desired graduality and program equivalence
  properties of the core language by giving an inequational theory. We
  then derive from this an operational semantics by orienting certain
  equations in the theory, showing that the operational behavior is
  derivable from the graduality and extensionality principles.

\item
  We prove type soundness and graduality by constructing a logical
  relations model, extending prior work on embedding-projection pair
  semantics to effects and subtyping.
\end{enumerate}

\section{Overview of GrEff}
\label{sec:overview}

Before discussing the syntax and semantics of GrEff, we provide an
informal introduction to its features and how it supports a gradual
migration from unchecked to checked effect handlers.
As an example, consider the implementation of a simple threading
library using effect handlers. We start with a system using unchecked
effect types in an ASCII syntax in Figure~\ref{fig:threads-imprecise}.
We split this program across three modules: first, a module
\texttt{Operations} defines the effects we will be using in our other
modules. These are the effects that the threads use: \texttt{print} for
displaying output so that we can observe the interleaving of threads,
\texttt{yield}, which yields back control to the scheduler, and most
importantly, \texttt{fork}, which allows for a thread to spawn new
threads. Each effect declaration \verb+effect e : Req ~> Resp+ is
annotated with two types: the type of \emph{requests} to the ambient
handler, and the type of expected \emph{responses} from the ambient
handler. For instance, the request type for print is a \texttt{str}ing
to be printed, and the response is unit. In a more realistic setting,
the response type might be a boolean to say if the printing succeeded,
or an unsigned integer to say how many bytes were succesfully
printed. For yield, the request and response are both unit. For fork,
the response type is again unit and the request type is a thunk
\texttt{1 -[?]> 1} where the $\dyn$ is the type of \emph{effects} the
function may raise when called. In this case, $\dyn$ indicates the
thunk might raise any effect.

\begin{figure}
\begin{verbatim}
module Operations where
  effect print : str ~> 1
  effect yield : 1 ~> 1
  effect fork  : (1 -[?]> 1) ~> 1
module Scheduler where
  import Operations.print : str ~> 1
  import Operations.yield : 1 ~> 1
  import Operations.fork  : (1 -[?]> 1) ~> 1
  define sch-loop : Queue (1 -[?]> 1) -[?]> str -[?]> str = lambda q.
    match q with
      empty           -> ()
      dequeue(thunk, q') -> shallow-handle thunk() with
        ret _ -> sch-loop q'
        fork(new,k) -> sch-loop (enqueue (enqueue q new) k)
        yield(_, k) -> sch-loop (enqueue q k)
        print(s, k) -> lambda s'. k(s' ++ s)
  define scheduler : (1 -[?]> 1) -[?]> str = lambda thunk.
    sch-loop (enqueue empty thunk) ""
module Main where
  import Operations.print : str ~> 1
  import Operations.yield : 1 ~> 1
  import Operations.fork  : (1 -[?]> 1) ~> 1
  import Scheduler.scheduler : (1 -[?]> 1) -[?]> str
  define letters : 1 -[?]> 1 =
    print("a"); yield(); print("b"); ()
  define numbers : 1 -[?]> 1 =
    print("1"); fork(letters); print("2"); ()
  define main: 1 -[?]> str =
    scheduler(numbers)
\end{verbatim}
  \caption{GrEff Threading Program with Imprecise Types}
  \label{fig:threads-imprecise}
\end{figure}

\begin{figure}
\begin{verbatim}
module Operations where
  effect print : str ~> 1
  effect yield : 1 ~> 1
  effect fork  : (1 -[fork,print,yield]> 1) ~> 1
module Scheduler where
  import Operations.print : str ~> 1
  import Operations.yield : 1 ~> 1
  import Operations.fork  : (1 -[fork,print,yield]> 1) ~> 1
  define sch-loop : Queue (1 -[fork,print,yield]> 1) -[]> str -[]> str = ...    
  define scheduler : (1 -[fork,print,yield]> 1) -[]> str = ...
module Main where
  import Operations.print : str ~> 1
  import Operations.yield : 1 ~> 1
  import Operations.fork  : (1 -[fork,print,yield]> 1) ~> 1
  import Scheduler.scheduler : (1 -[fork,print,yield]> 1) -[]> str
  define letters : 1 -[print,yield]> 1 =
    print("a"); yield(); print("b"); ()
  define numbers : 1 -[fork,print]> 1 =
    print("1"); fork(letters); print("2"); ()
  define main: str =
    scheduler(numbers)
\end{verbatim}
  \caption{GrEff Threading Program with Precise Typing}
  \label{fig:threads-precise}
\end{figure}

Next, module \texttt{Scheduler} defines a round-robin scheduler as a
handler for the provided effects. For simplicity the implementation
relies on some built-in queue implementation, and shallow handlers, a
simple extension to our formalism which uses the more complex deep
handlers.
Finally, we have the \texttt{Main} module, which uses the scheduler
defined in the \texttt{Scheduler} module with a thunk that uses the
effects defined in the \texttt{Operations} to implement a program that
prints a simple message using threads whose output will depend on the
scheduler's behavior.

The imprecision of the effect typing in this program means that
programmers have to rely on documentation or understanding of the code
to understand what effects might be raised when they import a function
from another module. With effect typing, this information can be
expressed precisely using effect annotations on the functions
themselves.
For instance, in the declaration of the \texttt{fork} operation, the
request is a thunk that when launched as a thread itself may raise
further effects such as manipulating shared state, \texttt{yield}ing
to other threads, or \texttt{fork}ing additional threads. However with
imprecise effect tracking, the \texttt{scheduler} procedure has the
uninformative type \texttt{(1 -[?]> 1) -[?]> 1} so we cannot specify
in the type which operations the scheduler will handle and which it
will propagate forward.

GrEff allows as well for the introduction of \emph{precise} effect
types to express these choices in the type structure.
In figure~\ref{fig:threads-precise}, we show a fully precisely typed
version of the same threading program (with implementations, which are
unchanged, now elided).
This allows us to specify in the \texttt{Scheduler} module that the
scheduler expects threads that can (1) \texttt{print} a string, (2)
\texttt{yield} to the other threads and (3) \texttt{fork} further
threads with the same effects. To express this, the scheduler module
changes the type to \texttt{1 -[fork,print,yield]> 1 -[]> str}
expressing that the scheduler will be passed a thunk that may \texttt{fork},
\texttt{print} or \texttt{yield}, but will itself return a string without raising any
effects. Additionally, we can express that \emph{fork}ed threads
should only raise these three effects as well. This is expressed by
annotating the \emph{import} statement, which defines fork as a
recursive\footnote{though recursive effect types are natural here, we
do not support them in our core language and leave this extension to
future work} effect type whose response type is trivial and whose
request type is that of thunks that can raise the three provided
effects.
This typing will then be used by all occurrence of the fork effect, in
raise or handlers, within this module.
The types are also changed in the main module, where the
\texttt{letters} thunk can be given a type expressing it only prints
and yields, whereas \texttt{numbers} thunk only forks and prints.
These are compatible with the types in scheduler using an effect
subtyping that allows functions that use fewer effects to be used in a
context that can handle more.

Since GrEff is a \emph{gradual} effect language, a programmer who
started with the imprecise program does not need to fully type the
entire program before running it. Instead, the programmer can
\emph{gradually} migrate from the imprecise style to the more precise
style, for example one module at a time. In fact, any of the $2^3=8$
combinations of the imprecise versions and precise versions of the
three modules presented here will pass the GrEff gradual
type-and-effect checker.
For instance, we might start with adding precise effect typing to the
\texttt{Operations} module to specify the effects that a forked thread
can have.
Whereas in a non-gradual type system, this would require changing the
consumer modules to use the more precise typing, in GrEff, the import
statements allow for the uses within the module to continue to use the
imprecise typing, and at the module boundary it is checked that the
precise components of the declared type for the \texttt{fork} effect
match the precise components of the declaration in the defining
module.
On the other hand, we can keep the \texttt{Operations} module
imprecisely typed, and instead add typing to the \texttt{Scheduler}
module first. This is again unusual compared to a conventional typed
language, we have declared a nominal data type in one module, but use
it at a different type in a client module. The import statements allow
for the gradual migration of the client code without changing the
original library.

The module system plays a crucial role in allowing for the programmer
to independently choose between migrating the declaration site of the
nominal datatype and its uses. If we were in a purely
expression-oriented language, then any change to the module
declaration, even in a gradual language, would change the typing of
the uses of the operation. Here we use the module boundaries in the
style of Typed Racket as a way to formally specify different
expectations of what the type of the nominal effect operations should
be in different portions of the codebase.

%
%
%
%
%



\section{Surface and Core Greff}

In this section, we introduce the syntax and typing of GrEff along
with its elaboration into a core language, Core GrEff.
GrEff includes a module system and nominal effect operations, as well
as a gradual type checking algorithm that allows for a mix of dynamic
and static effect tracking.
Core GrEff, on the other hand, is a simpler expression language with a
declarative type system where all gradual type casts (but not
subtyping) are explicit in the term.
The high-level features of GrEff are elaborated away into core GrEff.
Because Core GrEff is simpler, we describe its syntax and typing
first, and then describe GrEff and its type-checking/elaboration
algorithm.

\subsection{Syntax and Typing of Core GrEff}
\label{sec:corelang}

We give an overview of the Core GrEff syntax in
Figure~\ref{fig:core-syntax}.
Core GrEff expression syntax include typical lambda calculus syntax
for variables, let-bindings, functions and booleans.
Additionally, there is a term $\err$ that represents a runtime error produced by a failed cast
Next, it includes forms for raising an effect operation $\raiseOpwithM
\effname M$ and handling effect operations $\hndl M x N \phi$. The
handler includes a clause $\kw{ret} x. N$ to handle a return value for
$M$ as well as clauses for handling effects $\phi$. Abstracting from
syntactic details, $\phi$ is modeled as a finitely supported partial
function (written $\parfin$) from effect names to terms, which all
have two free variables $x$ and $k$ for the payload of the effect
raised and its continuation. That is, if syntactically a handler has a
clause $\effname(x, k) \mapsto N_\effname$, we model this by having
$\phi(\effname) = N_\effname$.
Next, Core GrEff includes four explicit gradual type cast forms:
downcasts ($\dncast A B M$) and upcasts ($\upcast A B M$) for value
types, as well as analogous casts for effect types ($\dncast \sigma
\tau M$ and $\upcast \sigma \tau M$).
Finally, we include a term $\err$ that represents a runtime cast
error.

The value types $A, B, C$ classify runtime values: in this simple
calculus, just booleans and functions, where functions are typed with
respect to a domain, codomain as well as an effect type $\sigma$ which
classifies what effects the function may raise when it is called.
The effect types are either $\dyn$ to indicate dynamically tracked
effects, or a concrete effect type. A concrete effect type says which
effect names $\effname$ can be raised, and when they are raised, what
is the type of the request $A$ the raising party provides and what is
the type of responses $B$ with which the handling party can
resume. Abstracting from syntactic details, this is defined to be a
finitely supported partial mapping from names to pairs of value types
(i.e., an element of the cartesian product $\text{ValueType}^2 =
\text{ValueType}\times\text{ValueType}$). To model that an effect
$\effname$ can be raised with request type $A$ and response type $B$
we would define $\sigma_c(\effname) = (A, B)$, which we will notate
more suggestively as $\effarr \effname A B \in \sigma_c$.
As shown in Section~\ref{sec:overview}, programs declare which effect names
can be used, and with which associated request and response types.
To track this information in typing core GrEff expressions, we type
check all GrEff expressions against a \emph{Signature} $\Sigma$ which
associates a pair of \emph{non-tracking} types to each name. By a
\emph{non-tracking} type, we mean a value types that only use $\dyn$
effect types.
Additionally, expressions are type-checked with respect to an ordinary
typing context $\Gamma$.
Finally, we define typical notions of value and evaluation context to
encode a call-by-value, left-to-right evaluation order.
Most notably, all casts are evaluation contexts, and function casts
are values, i.e. ``proxies'' that delay type enforcement until an
application is performed.

The use of non-tracking types in the signature is a design decision in
the semantics of GrEff: it means that when an effect is declared in a
module, it fully specifies only the non-effect typing portions of the
request and response types.
When a module imports an effect, it is only checked that the new
request and response type are \emph{consistent} with the exporting
module. Since effect types can be re-exported and the consistency
relation is not transitive, this means that in general the types used
in one module will not be consistent with those of the module where it
was originally declared.
However, transitive closure of consistency \emph{does} ensure that the
types have the same non-tracking portion, and so it is sensible to
define the valid instances of the effect type to be any that agree on
this non-tracking portion of the type.
An alternative would be for the signature to have a fully specified
type and limit all uses of the effect to be at least as precise as the
original declaration. However we argue that this is not in the spirit
of gradual typing: for instance it might be the case that module $P$
\textbf{p}rovides an effect declaration, module $I$ is an
\textbf{i}ntermediate that re-exports the effect and module $C$ is a
\textbf{c}lient of $I$ that uses the effect but does not directly
interact with $P$.
Say $P,I,C$ all initially use untracked effects, but then $C$ becomes
typed and so specifies precise effect typing for the effect.
The program functions properly and eventually $P$ is additionally made
more precise but in such a way that the effect implementation is
incompatible with the usage in $C$. In GrEff this does not lead to a
static error, because $C$ and $P$ are not directly communicating along
a precisely typed interface, but rather through an intermediary $I$
that uses imprecise typing. Indeed, it may be the case that $I$ uses
the effect differently between $C$ and $P$ and there is no runtime
type error. However, if $I$ becomes precisely typed, it must specify
its interpretation of the effect and will result in a static error
with either $C$ or $P$.

\begin{figure}
\begin{mathpar}
\begin{array}{rcl}
  \text{Terms } M,N &::=& x \alt \lambda x. M \alt M\, M' \alt \tru \alt \fls \alt \ifXthenYelseZ M N N \\
  && \alt\letXbeboundtoYinZ M x N \alt \raiseOpwithM {\effname} M \alt \hndl M x N \phi \\
  && \alt\upcast A B M \alt \dncast A B M \alt \upcast \sigma {\tau} M \alt \dncast {\sigma} {\tau} M \alt \err\\
  \text{Handler clause } \phi & \in & \text{Name} \rightharpoonup_{\textrm{fin}} \text{Term}\\
  \text{Value Types } A,B,C &::=& A \effto \sigma B \alt \boolty\\
  \text{Effect Types } \sigma,\tau &::=& \dyn \alt \sigma_c \\
  \text{Concrete Effect Types } \sigma_c &\in& \text{Name}\rightharpoonup_{\textrm{fin}} \text{ValueType}^2\\
  \text{Signature } \sig &\in& \text{Name}\rightharpoonup_{\textrm{fin}} \text{NonTrackingType}^2\\
  \text{Non-tracking Types } A_\dyn &::=& A_\dyn \to_{\dyn} A_\dyn \alt \boolty\\
  \text{Typing Contexts } \Gamma & ::= & \cdot \alt \Gamma, x : A\\
  \text{Values } V &::=& x \alt \lambda x:A. M \alt \tru\alt\fls \\
  &&\alt \upcast{A' \effto {\sigma'} B'}{A \effto {\sigma} B}V \alt \dncast{A' \effto {\sigma'} B'}{A \effto \sigma B}V\\
  \text{Evaluation Context } E & ::= & \bullet \alt \upcast A B E \alt \dncast A B E \alt \upcast \sigma {\tau} E \alt \dncast {\sigma} {\tau} E \\
  &&\alt \raiseOpwithM {\effname} E \alt \hndl E x N \phi \alt E\,M \alt V\,E\\
  &&\alt \ifXthenYelseZ E {N_t} {N_f} \alt \letXbeboundtoYinZ E x N
\end{array}
\end{mathpar}
\caption{Core GrEff Syntax}
\label{fig:core-syntax}
\end{figure}

Next, we present \emph{declarative} term typing rules in
Figure~\ref{fig:typing}.
The main judgment $\Tmhasty M A$ says that under the assumptions
$\Gamma$, $M$ can raise effects drawn from $\sigma$, and produce a
final value of type $A$.
We follow the convention that whenever we form the judgment $\Tmhasty
M A$ we must already have established that the types in $\Gamma, A,
\sigma$ are well-formed under the signature $\sig$.
First, we include a subsumption rule for value and effect
subtyping, which we will soon define.
The rules for value forms (variable, booleans, and lambdas) all have
an arbitrary effect type $\sigma$ because they do not raise any effects
themselves.
The runtime cast error $\err$ can be given any value or effect type.
The let, application and if rules simply require that all the
sub-terms use the same effect type, though subsumption can be used to
combine effects.
The raise rule says that the effect being raised needs to be in the
current effect type and the payload of the request must also have the
same effect type.

Next, the rule for typing a handler works as follows. First, the
output value type is $B$ and output effect type is $\tau$, while for
the scrutinee $M$ the corresponding types are $A$ and $\sigma$. First,
we check that the return clause $N$ has the same output types as the
handler overall, when its input $x$ has the type of the output of
$M$. Next, for each effect operation $\effarr \effname
{A_\effname}{B_\effname}$ raised by $M$, either the effect is not
handled by $\phi$, in which case it must be included in the final
effect type, or it is handled by $\phi$. If it is handled by $\phi$,
then the clause $\phi(\effname)$ must be well typed with a request
value $x: A_\effname$ and a continuation that takes responses and has
output effect and value types that match the term overall $k :
B_\effname \effto \tau B$.
Lastly, we include the rules for type and effect upcasts and
downcasts. Whenever a \emph{type precision} relationship $A \ltdyn B$
holds (to be defined), we get an \emph{up}cast from the more precise
type $A$ to the more imprecise type $B$ and a corresponding downcast
from $B$ to $A$.

\begin{figure}
  \begin{mathpar}
    \inferrule
    {\TmhastyRhoMT{\sigma}{M}{A} \and \sg \vdash A \subty B \and \sg \vdash \sigma\subty\tau}
    {\TmhastyRhoMT{\tau}{M}{B}}

    \inferrule
    {\Gamma(x) = A}
    {\Tmhasty x A}

    \inferrule
    {}
    {\Tmhasty \err A}

    \inferrule
    {}
    {\Tmhasty {\tru,\fls} \boolty}

    \inferrule
    {\TmhastySGRhoMT{\sig}{\Gamma, x : A}{\tau}{M}{B}}
    {\TmhastySGRhoMT{\sig}{\Gamma}{\sigma}{\lambda x . M}{A \effto \tau B}}

    \inferrule
    {\Tmhasty M A \\\\
      \TmhastySGRhoMT{\sig}{\Gamma, x : A}{\sigma}{N}{B}}
    {\Tmhasty {\letXbeboundtoYinZ M x N} {B}}

    \inferrule
    {\TmhastyRhoMT{\sigma}{M}{A \to_\sigma B} \\\\ \TmhastyRhoMT{\sigma}{N}{A}}
    {\TmhastyRhoMT{\sigma}{M\, N}{B}}
    \quad
    \inferrule
    {\Tmhasty M \boolty \\\\
      \TmhastyRhoMT{\sigma}{N_t}{B} \and
      \TmhastyRhoMT{\sigma}{N_f}{B}}
    {\Tmhasty {\ifXthenYelseZ M {N_t}{N_f}} B}
    \quad
    \inferrule
    {\TmhastyRhoMT{\sigma}{M}{A} \and
      \epsilon @ A \leadsto B \in \sigma
    }
    {\TmhastyRhoMT {\sigma}{\raiseOpwithM {(\epsilon @ A \leadsto B)} M} {B}
    }

    \inferrule
    {\TmhastyRhoMT \sigma M A\\\\
     \TmhastySGRhoMT \sig {\Gamma, x : A} \tau   N B\\\\
     (\forall (\effarr \effname {A_\effname}{B_{\effname}}) \in \sigma.~
     (\effname \not\in \dom(\phi) \wedge(\effarr \effname {A_\effname}{B_{\effname}}) \in \tau)\\\\
     \quad\vee(\TmhastySGRhoMT \sig {\Gamma, x:A_\effname, k:B_\effname \effto\tau B} \tau {\phi(\effname)} B))
    }
    {\TmhastyRhoMT {\tau}{\hndl M x N \phi} B}

    \inferrule
    {\Tmhasty M A \and A \ltdyn B}
    {\Tmhasty {\upcast A B M} {B}}

    \inferrule
    {\Tmhasty M B \and A \ltdyn B}
    {\Tmhasty {\dncast A B M} {A}}
    \quad
    \inferrule
    {\TmhastyRhoMT{\sigma} M A \and  \sigma \ltdyn \sigma'}
    {\TmhastyRhoMT{\sigma'} {\upcast \sigma {\sigma'} M} A}
    \quad
    \inferrule
    {\TmhastyRhoMT{\sigma'} M A \and \sigma \ltdyn \sigma'}
    {\TmhastyRhoMT{\sigma} {\dncast \sigma {\sigma'} M} A}
  \end{mathpar}
  \caption{Core Greff Typing}
  \label{fig:typing}
\end{figure}

Finally, finishing out the syntax, in Figure~\ref{fig:type_effect_precision_subtyping},
we define three judgments on types: well-formedness, subtyping and type precision.
Well-formedness $\sig \vdash A$ and $\sig \vdash \sigma$ checks that
the types used in effect operations erase to the types associated in
the signature. Here we use the notation $\lceil A\rceil$ to mean the
erasure of effect typing information in that we replace any effect
type subterms $\sigma$ with dynamic $\dyn$.
Subtyping works as usual for booleans and functions, contravariant in
domain of the function type, but covariant in the codomain and effect.
Subtyping for effect types includes both a \emph{width} subtyping
aspect: a smaller type can raise fewer operations, as well as a
\emph{depth} aspect that is \emph{covariant} in the request type and
\emph{contravariant} in the response type. This variance makes sense
from the perspective of the party \emph{producing} the request, to
match the function type subtyping.
Finally, type precision $A \ltdyn B$ tracks instead how ``dynamic'' or
``imprecise'' a type is.
For functions it is covariant in every argument, and for effect types,
the dynamic effect is the most imprecise and for two concrete effect
sets, it has a depth rule that that is covariant in request and
response positions. In a more standard gradual language with full
dynamic typing, in addition to the dynamic effect type we would have a
dynamic value type $\dyn_v$ that is similarly maximally imprecise
among value types.

\begin{figure}
  \begin{mathpar}
    \inferrule{}{\sig \vdash \boolty}\and
    \inferrule
    {\sig \vdash A \and\sig\vdash \sigma \and \sig \vdash B}
    {\sig \vdash A \effto \sigma B}\and
   \inferrule{}{\sig \vdash \dyn}\and
   \inferrule
   {\forall \effarr\effname A B \in \sigma_c.\\\\
    (\effarr\effname {|A|}{|B|} \in \sig). \wedge \sig \vdash A\wedge\sig \vdash B)}
   {\sig \vdash \sigma_c}

    \inferrule{}{\boolty \subty \boolty}\quad
    \inferrule
    {A' \subty A \and \sigma \subty \sigma' \and  B \subty B'}
    { A \to_{\sigma} B \subty A' \to_{\sigma'} B'}\quad
    \inferrule
    {}
    {\dyn \subty \dyn}\quad
    \inferrule
    {\forall \effarr\effname {A_{\sigma}} {B_\sigma} \in \sigma_c. \exists \effarr\effname {A_\tau}{B_\tau} \in \tau_c. \\\\\quad
      A_\sigma \subty A_\tau \wedge B_\tau \subty A_\tau}
    {\sigma_c \subty \tau_c}

    \inferrule
    {}
    {\boolty \ltdyn \boolty} \and
    \inferrule
    {A \ltdyn A' \and \sigma \ltdyn \sigma' \and B \ltdyn B'}
    {A \effto\sigma B \ltdyn {A'} \effto{\sigma'} {B'}}
    \and
    \inferrule
    {}
    {\sigma \ltdyn \dyn} \and
    \inferrule
    {\dom(\sigma_c) = \dom(\sigma_c') \\\\
     \forall \effarr\effname {A} {B} \in \sigma_c. \exists \effarr\effname {A'}{B'} \in \sigma_c'. \\\\\quad
      A \ltdyn A' \wedge B \ltdyn B'          
    }
    {\sigma_c \ltdyn \sigma_c'}
    \end{mathpar}
  \caption{Well formed types and effects, Type and Effect Precision}
  \label{fig:type_effect_precision_subtyping}
\end{figure}

\subsection{Syntax and Elaboration of GrEff}
\label{sec:syntax}

\begin{figure}
  \begin{mathpar}
    \begin{array}{rcl}
      \text{Programs } P & ::= & L; \cdots \, L_{main}\\ 
      \text{Value Types } A,B,C &::=& A \effto \sigma B \alt \boolty\\
      \text{Effect Types } \sigma,\tau &::=& \dyn \alt \sigma_s \\
      \text{Operation Set } \sigma_s,\tau_s &\in& \mathcal{P}_{\textrm{fin}}(\textrm{Name}) \\
      \text{Values } V &::=& x \alt \lambda x:A. M \alt \tru\alt\fls \\
      \text{Terms } M,N &::=& x \alt \raiseOpwithM {\varepsilon} M \alt \surfHndl {\compty C \sigma} M x N \phi \\
      && \alt \lambda x. M \alt M\, M' \alt \tru \alt \fls \alt \ifXthenYelseZ M N N \\
      && \alt M :: A \alt M :: \sigma \\
      \text{Handler clauses } \phi & \in & \text{Name} \rightharpoonup_{\text{fin}}\text{Term}\\
      \text{Modules } L & ::= & \kw{module} m \; \{ b \}\\
      \text{Module Body } b & ::= & \cdot \alt D; b \\
      \text{Main Module } L_{main} & ::= & \kw{main} \{ b; M \}\\
      \text{Module reference } r & ::= & m.x \alt m.\effname \\
      \text{Declaration } D & ::= &
      \importeffXatYtoZ r A B \alt \neweffEofReqtoResp \effname A B\\
      &&\alt \defXtobeY x V \alt \importvalXasYatZ r x A\\
      \text{Program Typing Contexts } \Delta & ::= & \cdot \alt \Delta, x \mapsto \Gamma_s\\
      \text{Module  Typing Contexts } \Gamma_s & ::= & \cdot \alt  \Gamma_s , \varepsilon : A \leadsto B \alt \Gamma_s, x : A
    \end{array}
  \end{mathpar}
  \caption{GrEff Syntax}
  \label{fig:greff-syn}
\end{figure}

We present the syntax for the surface language GrEff in
Figure~\ref{fig:greff-syn}.
A GrEff program $P$ consists of a sequence of modules ending in a
single ``main'' module.
Each module $m$ consists of two parts: first, the effect definitions
and then the value definitions, whose types annotations may use the
effects previously defined in that module.
An effect definition is either a declaration of a new effect operation
$\neweffEofReqtoResp \effname A B$ or an import of an existing effect
operation $\importeffXatYtoZ {m.\effname} A B$. In either case, the
declaration includes the request type $A$ and the response type $B$ of
the effect.
An effect import brings an effect defined in another module into the
current scope, but with a possibly different request and response
type. To support \emph{gradual} migration, these types are allowed to
have a different level of precision than the original, but where both
are precise they must match.
After the effect declarations are the value definitions which are also
either a definition of a new value $\defXtobeY x V$ or an import
of a value declared in a different module at a possibly different type
$\importvalXasYatZ r x A$.
For simplicity, all effects and values are public and can be imported
by later modules.
Finally a program ends with a main module, which consists of the
same kind of effect and value declarations, followed by a final main
expression.

\begin{figure}
  \begin{small}
  \begin{mathpar}
    \inferrule
    {\elabmod \sig \Delta \cdot b {\sig'} \gamma \Gamma\\\\
     \elabtm \Gamma {M_s} M \sigma A}
    {\elabprog \sig \Delta {\kw{main} b\; M_s} {\sig'} {\letXbeboundtoYinZ \gamma {\Gamma} M} \sigma A}
    \quad
    \inferrule
    {\elabmod \sig \Delta \cdot b {\sig'} \gamma \Gamma \\\\
     \elabprog {\sig,\sig'} {\Delta,m\mapsto\Gamma} P {\sig''} M \sigma A}
    {\elabprog \sig \Delta {\kw{module} m\;b;\; P} {\sig',\sig''} {\letXbeboundtoYinZ \gamma {\Gamma} M} \sigma A}\\

    \inferrule
    {}
    {\elabmod \sig \Delta \Gamma \cdot \cdot \cdot \cdot}

    \inferrule
    {\varepsilon \not\in \sig\and
     \elabty \Gamma {A_s} A\and
     \elabty \Gamma {B_s} B
    }
    {\elabdecl \sig \Delta \Gamma {\neweffEofReqtoResp \varepsilon {A_s} {B_s}} {(\varepsilon @ \lceil A\rceil \leadsto \lceil B\rceil)} {\cdot}{\varepsilon @ A \leadsto B}}

    \inferrule
    {\elabdecl \sig \Delta \Gamma D {\sig'} {\gamma'}{\Gamma'}\\\\
     \elabmod {\sig,\sig'}{\Delta} {\Gamma,\Gamma'} b {\sig''}{\gamma''}{\Gamma''}
    }
    {\elabmod \sig \Delta \Gamma {D; b} {\sig',\sig''} {\gamma',\gamma''} {\Gamma',\Gamma''}}
    \quad
    \inferrule
    {\Delta(m) \ni \varepsilon @ A' \leadsto B'\and
     \elabty\Gamma {A_s} A\and
     \elabty\Gamma {B_s} B\\\\
     A \sim A'\and
     B \sim B'
    }
    {\elabdecl \sig \Delta \Gamma {\importeffXatYtoZ {m.\varepsilon} {A_s} {B_s}} \cdot \cdot {\varepsilon @ A\leadsto B}}

    \inferrule
    {\elabtm \Gamma {V_s} V \emptyset A}
    {\elabdecl \sig \Delta \Gamma {\defXtobeY x {V_s}} \cdot {V/x} {x:A}}
    \quad
    \inferrule
    {\Delta(m) \ni x:A'\and
      \elabty \Gamma {A_s} A\and
      A' \gsubty A
    }
    {\elabdecl \sig \Delta \Gamma {\importvalXasYatZ {m.x} y {A_s}} \cdot {\obliqueCast {A}{A'}x/y} {y:A}}
  \end{mathpar}
  \end{small}
  \caption{GrEff Typing/Elaboration, Module Language}
  \label{fig:elab-mod}
\end{figure}

\begin{figure}
  \begin{mathpar}
    \obliqueCast A B M = \dncast{A}{\lceil A\rceil}\upcast{B}{\lceil B \rceil}M\and
    \obliqueCast \sigma \tau M = \dncast{\sigma}{\dyn}\upcast{\tau}{\dyn} M
    
    \inferrule
    {\Gamma \ni x:A}
    {\elabtm \Gamma x x \emptyset A}

    \inferrule
    {}
    {\elabtm \Gamma \tru \tru \emptyset \boolty}

    \inferrule
    {}
    {\elabtm \Gamma \fls \fls \emptyset \boolty}

    %
    \inferrule
    {\elabtm \Gamma {M_s} M \sigma {A'}\and
     \elabty \Gamma {A_s} A\and
     A' \gsubty A
    }
    {\elabtm \Gamma {M_s :: A_s} {\obliqueCast {A}{A'} M} \sigma A}
    \quad
    \inferrule
    {\elabtm \Gamma {M_s} M {\sigma'} {A}\and
     \elabty \Gamma {\sigma_s} \sigma\and
     \sigma' \gsubty \sigma
    }
    {\elabtm \Gamma {M_s :: \sigma_s} {\obliqueCast {\sigma}{\sigma'} M} \sigma A}

    \inferrule
    {\elabtm \Gamma {M_s} M {\sigma_m} \boolty\and
     \elabtm \Gamma {N_s} N {\sigma_n} B\and
     \elabtm \Gamma {N_s'}{N'}{\sigma_n'} {B'}\\\\
     C = B \gjoin B'\and
     \sigma = \sigma_m \gjoin \sigma_n\gjoin \sigma_n'}
    {\elabtm \Gamma {\ifXthenYelseZ {M_s} {N_s}{N_s'}} {\ifXthenYelseZ {\obliqueCast {\sigma}{\sigma_m}M} {\obliqueCast \sigma {\sigma_n}\obliqueCast C B N}{\obliqueCast \sigma {\sigma_n'}\obliqueCast C {B'} N'}} \sigma \boolty}

    %
    \inferrule
    {\elabty \Gamma {A_s} A\and
     \elabtm {\Gamma, x:A} {M_s} M \sigma B
    }
    {\elabtm \Gamma {\lambda x:A_s. M_s} {\lambda x. M} \emptyset {A \to_{\sigma} B}}
    \quad
    \inferrule
    {\elabtm \Gamma {M_s} M {\sigma_m} {A \to_{\sigma_o} B}\and
     \elabtm \Gamma {N_s} N {\sigma_n} {A'}\\\\
     A' \gsubty A\and
     \sigma = \sigma_m \gjoin \sigma_n \gjoin \sigma_o
    }
    {\elabtm \Gamma {M_s\,N_s} {({\obliqueCast \sigma {\sigma_m}M})(\obliqueCast {A}{A'}\obliqueCast {\sigma}{\sigma_n} N)} \sigma B}

    %
    \inferrule
    {\elabtm \Gamma {M_s} M {\sigma_m} {A'}\and
     \Gamma \ni \varepsilon @ A \leadsto B\and
     A' \gsubty A\and
     \sigma = \sigma_m \gjoin \{\varepsilon @ A \leadsto B \}
    }
    {\elabtm \Gamma {\raiseOpwithM \varepsilon {M_s}} {\letXbeboundtoYinZ {\obliqueCast \sigma {\sigma_m} M} x {\obliqueCast{\sigma}{\{\varepsilon @ A \leadsto B \}}\raiseOpwithM \varepsilon {\obliqueCast A {A'} x}}} \sigma B}

    \inferrule
    {\elabty \Gamma {C_s} C\and\elabty \Gamma {\sigma_s} \sigma\and
     \elabtm \Gamma {M_s} M {\sigma_m} {A_m}\\\\
     \elabtm {\Gamma,x:A} {N_s} N {\sigma_n} {C_n} \and \sigma_n \gsubty \sigma \and C_n \gsubty C \\\\
     \dom(\phi_{\Leftarrow}) = \dom(\phi_s) \and\elabHandleType \Gamma {\sigma_m}{\sigma}{\dom(\phi_s)} {\sigma_m'}\\\\
     (\forall \effname \in \dom(\phi_s).~ \exists (\effname @ A_\effname \leadsto B_\effname) \in \Gamma.\\\\
     \elabtm {\Gamma, x : A_\effname, k : B_\effname \effto{\sigma} C} {\phi_s(\effname)} {N_\effname} {\sigma_\effname}{C_\effname}\\\\
     \sigma_\effname \gsubty \sigma \and C_\effname \gsubty C \and \phi_{\Leftarrow}(\effname) = \obliqueCast{\sigma}{\sigma_\effname}\obliqueCast{C}{C_\effname}N_\effname)
    }
    {\elabtm \Gamma {\surfHndl {\compty {\sigma_s}{C_s}} {M_s} x {N_s} {\phi_s}}
      {\hndl {\obliqueCast {\sigma_m'} {\sigma_m} M} x {\obliqueCast {\sigma}{\sigma_n} \obliqueCast {C}{C_n}N} {\phi_{\Leftarrow}}}
      \sigma C
    }
  \end{mathpar}

  \begin{mathpar}
    \inferrule
    {\dom(\sigma_c) \subseteq \dom(\tau_c) \cup \sigma_s}
    {\elabHandleType \Gamma {\sigma_c}{\tau_c} {\sigma_s} {\tau_c \cup \Gamma(\sigma_s)}}
    

    \inferrule
    {}
    {\elabHandleType \Gamma \dyn {\tau_c} {\sigma_s} {\tau_c \cup \Gamma(\sigma_s)}}


    \inferrule
    {}
    {\elabHandleType \Gamma {\sigma_c} {\dyn} {\sigma_s} {\sigma_c|_{\sigma_s}} \uplus \lceil \Gamma(\dom(\sigma_c) - \sigma_s)\rceil}


    \inferrule
    {}
    {\elabHandleType \Gamma \dyn \dyn {\sigma_s} \dyn}
  \end{mathpar}
  \caption{GrEff Typing/Elaboration, Expression Language}
  \label{fig:elab-exp}
\end{figure}

Next, we present the elaborator from GrEff into core GrEff, which also
serves as the type checker.
We view GrEff programs as essentially a description of an effect
signature $\sig$ and a closed expression well-typed under that
signature.
The module system is a way to manage the declaration of new effect
operations in the signature and a way to manage the typing of effect
operations by giving nominal associations to request and response
types rather than solely the structural typing in core GrEff.
We describe the elaboration of the module language in
Figure~\ref{fig:elab-mod}.
The top-level judgment $\elabprog \sig \Delta P {\sig'} M \sigma A$
says that under the starting signature $\sig$ and previously defined
modules $\Delta$, we can elaborate $P$ to a term $M$ with effect type
$\sigma$ and value type $A$ that is well-typed under the extension of
the signature by $\sig'$. To elaborate a complete program, we
initialize this with empty signature and module typing ($\elabprog
\cdot \cdot P \sig M \sigma A$).
This expresses that not only does a program denote a core GrEff
program, but it also has a ``side effect'' of allocating new effect
names $\sig'$.
A module is elaborated with the judgment $\elabmod \sig \Delta \Gamma
b {\sig'} {\gamma'} {\Gamma'}$. The outputs of this judgment are the
newly allocated effects of the module $\sig'$, the names of effect
operations and types for values the module defines $\Gamma'$ and the
definitions of all the values the module defines, given as a
substitution $\gamma'$ from names in $\Gamma'$ to terms of their
associated types.
Then to elaborate a program consisting of several modules, first you
elaborate the modules and then elaborate the remainder of the program
and finally combine the two by let-binding all of the names declared
in the module, which we write as a shorthand $\letXbeboundtoYinZ
\gamma {\Gamma} M$. Note that though $\Gamma$ contains both variables
and effect declarations, the effect declarations are unused in this
part of the elaboration.
A module is elaborated by combining the results of elaborating each
declaration.
A new effect declaration checks that the name is not previously
declared, and then recursively elaborates the syntactic types declared
for request and response and then adds these to the allocated effects
as well as the local effect names declared in the module.
When adding to the signature, we take erasure of the types
because signatures use untracked types.
Next, to import an effect from a different module, the types given for
the effect are checked to be compatible with the types declared in the
other module. Note that for simplicity of presentation, all effects
must be used with the same name in all modules. More flexible renaming
mechanisms can easily be supported in a realistic implementation.
Here the compatibility judgment $A \sim A'$ is defined as the
conjunction of gradual subtyping in both directions, $A \gsubty A'$
and $A' \gsubty A$, to be defined soon.
This ensures that any imports from that module using this effect name
will succeed. We check gradual subtyping in both directions, as the
effect may be used in both postive and negative positions in a later
import.
This effect name is added to the local names only, and not the
signature, because it is using an already allocated effect name.
Next, defining a value simply elaborates the value and adds its type
to the output typing and associates the value to that name.
Importing a value is similar, except that we check that the declared
type is a gradual subtype, and so can be coerced by the cast
$\obliqueCast {A} {A'}$, whose definition will be described shortly.

Next, we define the elaboration of the expression language in
Figure~\ref{fig:elab-exp}.
The judgment $\elabtm \Gamma {M_s} M \sigma A$ says that under the
typing of names given by $\Gamma$, the GrEff expression $M_s$
elaborates to the core GrEff function $M$, which will be well-typed
with inferred effect type $\sigma$ and value type $A$.
All forms essentially elaborate to similar forms in core GrEff, but
with suitable casts inserted.
First, we define the translation of value type casts $\obliqueCast
{A}{B} M$ and effect type casts $\obliqueCast \sigma \tau M$ as an
upcast followed by a downcast.
For the effect cast, these casts go through the dynamic effect type,
but for two value types there is no single most dynamic effect type so
we again use the erasure operation.
Note that this will only be well-typed in case $\lceil B\rceil =
\lceil A\rceil$, which is ensured whenever $A \gsubty B$, which is a
precondition for inserting a cast.
This is not necessarily the most efficient implementation of the cast,
we discuss optimizations in Section~\ref{sec:semantics:casts}

Next, variables, boolean values and function values elaborate to
themselves with an empty effect type $\emptyset$.
The let-binding form shows how different effect types are combined:
the effect types of $M$ and $N$ are combined using a gradual join
$\gjoin$, and casts are inserted into $M$ and
$N$ to give them this effect type.
The gradual join acts as the join on precise parts of the type, but
extended such that $\sigma \gjoin \dyn = \sigma$.
The ascription forms simply check that the appropriate kind of type
satisfies a gradual subtyping judgment and inserts a cast.
This uses the elaboration of types $\elabty \Gamma {A_s} A$, defined
below.
The if rule checks that the condition has boolean type and gives the
output value type as the gradual join of the branches, and the output
effect type as the gradual join with the condition expression as well,
matching prior work \cite{AGT}.
The application rule is similar except that the argument is cast to
have the type of the domain of the function and the effect type of the
function is joined with the effect types of the terms.
Next, we have the raise form, which elaborates to a raise but first
let-binds the request term and casts the raise term to have an effect
type that is the join of the request term's effect type and the
operation's type.
Finally, we have the most complex case, the handle form.
The handle form elaborates to a handle form in the core language with
casts inserted in each case to make them agree with the ascribed value
type $C$ and effect type $\sigma$.
The request variables and input to the continuations are given by
looking up the effect in $\Gamma$, while the output is given by the
ascription.
The most complex part of this elaboration is the cast needed for the
scrutinee $M$. In the core language, we need that all of the effects
that $M$ raises but are \emph{not} caught by the handle are in the
output type $\sigma$.
But when $\sigma$ is dynamic and $M$ has concrete effect type or
vice-versa, this is not necessarily true, so in these cases a cast
must be inserted that effectively handles all of the ``other'' effects.
This definition is given below in a special elaboration of handle
scrutinees $(\elabHandleType \Gamma {\sigma} {\tau} {\sigma_s}
{\sigma_o})$. Here, the type $\sigma$ is the elaborated type of the
scrutinee, $\tau$ is the elaborated type of the result of the handle
expression, and $\sigma_s$ is the set of effects caught by the
handler, where we write $\Gamma(\sigma_s)$ for the map that looks up
the currently associated types for each operation in
$\sigma_s$. First, if $\sigma$ and $\tau$ are both precise collections
of effects, then we check that all of the effects it raises are either
caught or still occur in the output type, and we insert a subtyping
cast. Second, if $\sigma$, the type of the scrutinee is imprecise,
then we downcast it to include only the union of the output effects
and the caught effects, otherwise erroring. Third, if the scrutinee is
precise but the result $\tau = \dyn$ is dynamic, then then we need to
upcast all of the unhandled effect operations to their dynamic
versions. This is expressed by having the result type be the
combination $(\uplus)$ of the effects who are handled as is, written
$\sigma_c|_{\sigma_s}$ with the most dynamic version of any other
effects that are not handled $\lceil \Gamma(\dom(\sigma_c) -
\sigma_s)\rceil$. Here $\sigma_c|_{\sigma_s}$ means the restriction of
the partial function $\sigma_s$ to only be defined on the set
$\sigma_s$.
Finally, if the scrutinee and the goal are both imprecise then we put
a trivial identity cast to $\dyn$ on the scrutinee.

\begin{figure}
  \begin{mathpar}
  \inferrule
  {}
  {\elabty \Gamma \boolty \boolty}\and
  \inferrule
  {\elabty \Gamma {A_s} A\and
   \elabty \Gamma {B_s} B\\\\
   \elabty \Gamma {\sigma_s} \sigma
  }
  {\elabty \Gamma {A_s \effto {\sigma_s} B_s} {A \effto \sigma B}}\and
  \inferrule
  {}
  {\elabty \Gamma \dyn \dyn}\and
  \inferrule
  {\dom(\sigma_c) = \sigma_s\\\\
    \forall \effname \in \sigma_c.~ \sigma_c(\effname) = \Gamma(\effname)}
  {\elabty \Gamma {\sigma_s} {\sigma_c}}
  \\

  \inferrule
  {}
  {\boolty \gsubty \boolty}\quad
  \inferrule
  {A' \gsubty A' \and \sigma \gsubty \sigma' \and B \gsubty B'}
  {A \effto \sigma B \gsubty A' \effto {\sigma'} B'}\quad
  \inferrule
  {}
  {\dyn \gsubty \sigma}\quad
  \inferrule
  {}
  {\sigma \gsubty \dyn}\quad
    \inferrule
    {\forall \effarr\effname {A_{\sigma}} {B_\sigma} \in \sigma_c. \exists \effarr\effname {A_\tau}{B_\tau} \in \tau_c. \\\\\quad
      A_\sigma \gsubty A_\tau \wedge B_\tau \gsubty A_\tau}
    {\sigma_c \gsubty \tau_c}
  \end{mathpar}
  \caption{Gradual Subtyping and Type Elaboration}
  \label{fig:typing-elab-and-gsub}
\end{figure}

Finally, Figure~\ref{fig:typing-elab-and-gsub} describes the elaboration
of types and gradual subtyping.
Value and effect type elaboration $\elabty \Gamma A B$ is mostly
structural except that the rule for concrete effect sets resolves the
request and response types of the effect operation based on the
context $\Gamma$.
Next, we describe the mostly standard \emph{gradual} subtyping of
value types $A \gsubty B$ and effect types $\sigma \gsubty \tau$ to
determine when a dynamic cast $\obliqueCast {B} A$ or $\obliqueCast
\tau \sigma$ would reduce to subtyping on the precise portions of the
types.
Note that we define gradual subtyping of types in the core language
i.e., after elaboration, so that we can compare effect types across
module boundaries that use different typings for the effect names.
With this intuition, the definition is like that of subtyping, except
that the dynamic effect type is a gradual subtype and supertype than
all other effect types. We conclude by noting the following syntactic
properties of elabortation, which follow by structural induction.
\begin{lemma}[Elaboration is a function]
  If $\elabprog \cdot \cdot P \sig M \sigma A$ and $\elabprog \cdot \cdot P
  {\sig'}{M'}{\sigma'}{A'}$ then $\sig=\sig'$ and $M = M'$ and $\sigma
  = \sigma'$ and $A = A'$.
\end{lemma}
\begin{lemma}[Elaborated terms are Well-typed] 
  If $\elabprog \cdot \cdot P \sig M \sigma A$, then $\sig \pipe\cdot \vdash_\sigma M : A$.
\end{lemma}

\section{Axiomatics and Operational Semantics}

Next we turn to the semantic aspects of GrEff: how expressions are
evaluated, what simplifications/optimizations are correct to perform,
and that the graduality principle holds for the language.
We formalize these three aspects axiomatically in the form of an
\emph{inequational theory} for reasoning about Core GrEff programs.
That is, we define a notion of inequality $M \ltdyn N$ between
expressions called \emph{term precision}, which is a kind of extension
of the notion of type precision to expressions.
The semantic interpretation of this inequality is that $M$ has the
same behavior as $N$ with respect to output and termination, except in
that it may raise a dynamic type error when $N$ does not.
From this notion inequality we get an induced equivalence relation $M
\equiv N$ that specifies when $M$ and $N$ have the same behavior.
Term precision and the induced equivalence are used to model our
desired semantic ideas: an expression $M$ can be evaluated to a value
$V$ when the equivalence $M \equiv V$ holds, $M$ can be
simplified/optimized to $N$ when $M \equiv N$ holds, and the
graduality principle states that when $M$ is rewritten in the surface
language to some $M'$ that has more precise typing information, than a
corresponding relationship $M' \ltdyn M$ should hold: adding more
precise type information results in more precise dynamic type
checking.
With this in mind, we axiomatize the valid optimizations known from
effect handlers as well as desired inequalities from prior work on
graduality in our inequational theory.

Axioms are only useful if we can construct models in which they are
satisfied. For GrEff, we do this by constructing an \emph{operational}
semantics that specifies more precisely how to evaluate programs and
then define notions of observational equivalence and an error ordering
to model $\equiv$ and $\ltdyn$ and prove that all of the axioms are
valid in this operational model.
We will construct this operational semantics, based on the axiomatics:
we show in Section~\ref{sec:operational} that every reduction $M
\stepsin {} N$ is justified by a provable equivalence $M \equiv N$ in
the inequational theory. For many rules this is very straightforward,
e.g., $\beta$ reduction of functions is justified by a corresponding
$\beta$ equation. The most utility we get from the axioms in this case
is for the cast reductions: cast reductions for handlers are justified
not by a direct corresponding rule in the axioms, but instead by
extensionality ($\eta$) principles for handlers combined with a least
upper bound/greatest lower bound property of casts identified in prior
work as being key to the graduality property
\cite{newlicata18}. This shows that the operational behavior we
define has a canonical status: if certain optimizations for
handlers are to be valid, and the graduality property is desired, then
the cast reductions we define must be used.

\subsection{Axiomatics}

We present a selection of the rules of the inequational theory of term
precision in Figure~\ref{fig:inequational-theory}.
The full rules are provided in the appendix.
The form of the inequality judgment is $\Gamma^\ltdyn \vdash_{\sigma \ltdyn \tau} M \ltdyn
N : A \ltdyn B$, which says that $M$ is more precise, or, roughly,
``errors more'' than $N$.
This is a kind of \emph{heterogeneous} inequality relation in that $M$
and $N$ are not required to have the same type: $M$ must have value
type $A$ and effect type $\sigma$ and $N$ must have value type $B$ and
effect type $\tau$ under the context $\Gamma^\ltdyn$ and $A \ltdyn B$
and $\sigma \ltdyn \tau$ must hold.
We allow for $M$ and $N$ to be open terms, typed with respect to the
typing context $\Gamma^\ltdyn$.
The typing context $\Gamma^\ltdyn$ is like an ordinary typing context
$\Gamma$, except that variables are typed $x : A \ltdyn B$ where the
left type $A$ is the type $x$ has in the left term $M$ and $B$ is the
type for $N$. For the context to be well formed, each of the $A \ltdyn
B$ must be provable.

First, we add an axiom that $\err$ is the least term of any type, to
model the graduality property.  Next, we add an axiom that $\ltdyn$ is
transitive, where both the value and effect type are allowed to vary
simultaneously. The relation is reflexive as well, but this is
admissible from congruence rules.
Secondly, we give the congruence rules for functions and application,
and the full system includes such a congruence rule for all term
constructors.
Next we have computation ($\beta$) and reasoning ($\eta$) rules for
each type. For functions and if, these are standard call-by-value
$\beta\eta$ rules, so we instead show only the handle rules. There are
two $\beta$ rules for handle. If the term being handled is a value,
then the return clause is used. If the term being handled is a raise
of an effect $\effname$, it is equivalent to the handler clause
$\phi(\effname)$ where the continuation is the captured continuation
surrounding the original handler term. We require this to be a let,
but note that we have additional rules that imply that any evaluation
context that doesn't handle can be re-written as a let.
We then have two reasoning ($\eta$) rules for handle. First, if $M$ is
handled by a handler with no effect clauses, then the handler is
equivalent to a let-binding. This can be combined with standard rules
for let binding to show that any term is equivalent to a handler with
no clauses $M \equiv \hndl M x x \emptyset$. We call this the
\emph{non-handling} principle.
Second, we have a rule that says that any clause that simply re-raises
its operation with the same continuation it was passed can be dropped
from the handler, as this is the same behavior as not catching the
term at all. We call this the \emph{effect forwarding} principle, as
it says that forwarding an effect to the ambient context is equivalent
to not handling it explicitly at all. Combined with the non-handling
principle, any term $M$ with effect type $\sigma$ can be shown
equivalent to $\hndl M x x {\phi_\sigma}$ where $\phi_\sigma$ simply
forwards all the effects in $\sigma$.
We next show rules describing the interaction of subtyping with value
type casts, the full system includes analogous rules for effect types.
The first says that an upcast followed by a subtyping coercion is less
than a subtyping coercion followed by an upcast, and the downcast rule
is similar.
Finally, we have rules specifying the behavior of value and effect
casts. These rules characterize upcasts as least upper bounds and
downcasts as greatest lower bounds. The first rule shows that the
downcast is a lower bound and the second that it is the greatest. The
upcasts have similar rules, and we include analogous rules for effect
casts as well.
These lub/glb properties are adapted from prior work on axiomatics for
gradual typing \cite{newlicataahmed19}, but now incorporate the
ordering on both effect and value typing.
We found that this general form of the rule, where the effect is
allowed to differ ($\sigma \ltdyn \sigma'$) while performing a value
cast, is essential for proving the commutativity of value and effect
casts, which is used in the derivation of the operational semantics
and also valid in our logical relations model.

\begin{figure}
  \begin{mathpar}
    \inferrule
    {\Gamma \vdash_{\sigma} M : A}
    {\Gamma \vdash_{\sigma\ltdyn \sigma} \err \ltdyn M : A \ltdyn A}

    \inferrule
    {\Gamma^\ltdyn \vdash_{\sigma_1 \ltdyn \sigma_2} M_1 \ltdyn M_2 : A_1 \ltdyn A_2\and
     \Gamma'^\ltdyn \vdash_{\sigma_2 \ltdyn \sigma_3} M_2 \ltdyn M_3 : A_2 \ltdyn A_3\\\\
     \text{rhs}(\Gamma^\ltdyn) = \text{lhs}(\Gamma'^\ltdyn)\quad
     \text{lhs}(\Gamma^\ltdyn) = \text{lhs}(\Gamma''^\ltdyn)\quad
     \text{rhs}(\Gamma'^\ltdyn) = \text{rhs}(\Gamma''^\ltdyn)
    }
    {\Gamma''^\ltdyn \vdash_{\sigma_1 \ltdyn \sigma_3} M_1 \ltdyn M_3 : A_1 \ltdyn A_3}

    \inferrule
    {\Gamma^\ltdyn, x:A\ltdyn A' \vdash_{\tau\ltdyn\tau'} M \ltdyn M' : B \ltdyn B'}
    {\Gamma^\ltdyn \vdash_{\sigma \ltdyn \sigma'} \lambda x. M \ltdyn \lambda x. M' : A\effto \tau B \ltdyn A' \effto {\tau'} B'}

    \inferrule
    {\Gamma^\ltdyn \vdash_{\sigma\ltdyn\sigma'} M \ltdyn M' : A \effto \sigma B \ltdyn A' \effto {\sigma'} B'\\\\
     \Gamma^\ltdyn \vdash_{\sigma\ltdyn\sigma'} N \ltdyn N' : A \ltdyn A'}
    {\Gamma^\ltdyn \vdash_{\sigma \ltdyn \sigma'} M\,N \ltdyn M'\,N' : B \ltdyn B'}

    \inferrule
    {}
    {\hndl {V} y M \phi \equiv M[y/v]}

    \inferrule
    {}
    {\hndl {({\letXbeboundtoYinZ {\raiseOpwithM \effname x} o {N_k}})} y M \phi \\\\\equiv
      \phi(\effname)[\lambda o. \hndl {N_k} y M \phi/k]}

     \inferrule
     {}
     {\hndl M x N {\emptyset} \equiv \letXbeboundtoYinZ M x N}
     \quad
    \inferrule
    {\forall\effname \in \dom(\phi).~ \psi(\effname) = \phi(\effname) \\\\
      \forall \effname \in \dom(\psi). \effname\not\in\dom(\phi) \Rightarrow 
        \psi(\effname) = k(\raiseOpwithM \effname x)
    }
    {\hndl M y N \phi \equiv \hndl M y N \psi : \compty \sigma B}
    
    \inferrule*[]
    {
      A \subty A' \and B \subty B' \and\\\\
      \Gamma^\ltdyn \vdash_{\sigma \ltdyn \sigma'} M \ltdyn N : A
    }
    {
      \Gamma^\ltdyn \vdash_{\sigma \ltdyn \sigma'}
      \upcast A B M \ltdyn \upcast {A'} {B'} N : B' \ltdyn B'
    }
    {\quad}
    \inferrule*[]
    {
      A \subty A' \and B \subty B' \and\\\\
      \Gamma^\ltdyn \vdash_{\sigma \ltdyn \sigma'} M \ltdyn N : B \ltdyn B
    }
    { \Gamma^\ltdyn \vdash_{\sigma \ltdyn \sigma'}
      \dncast {A'} {B'} M \ltdyn \dncast A B N : A' \ltdyn A'}

    \inferrule
    {\Gamma^\ltdyn \vdash_{\sigma \ltdyn \sigma'} M \ltdyn N : B \ltdyn B}
    {\Gamma^\ltdyn \vdash_{\sigma \ltdyn \sigma'} \dncast A B M \ltdyn N : A \ltdyn B}

    \inferrule
    { \Gamma^\ltdyn \vdash_{\sigma \ltdyn \sigma'} M \ltdyn N : A \ltdyn B } 
    { \Gamma^\ltdyn \vdash_{\sigma \ltdyn \sigma'} M \ltdyn \dncast {A} {B} N : A \ltdyn A}
  \end{mathpar}
  \caption{Inequational Theory}
  \label{fig:inequational-theory}
\end{figure}

\subsection{Operational Semantics}
\label{sec:operational}

Next, we show a selection of the rules of the operational semantics $M
\stepsin {} M'$ in Figure~\ref{fig:operational-semantics}, eliding the
standard call-by-value rules for booleans, functions and let-bindings.
We capture the left-to-right, call-by-value evaluation order by using
evaluation contexts defined in Section~\ref{sec:corelang}.
First, we have the $\beta$ rules for handlers: when handling a value,
execute the return clause. Next, when a raise occurs, we search for
the closest enclosing handler that handles the raised effect and
capture the intermediate evaluation context in the continuation passed
to the appropriate handler..
We capture this with the relation $E' \apart \effname$ which says that
the evaluation context does not handle the given
operation.

The next rules concern the behavior of effect casts.
First, all effect casts are the identity on values. Next, when
upcasting a raise, we re-raise the effect, but \emph{up}cast the
request and \emph{down}cast the response according to the types in the
output effect type.
An effect downcast works dually if the effect occurs in the result
effect type.
However, if the effect does not occur in the output effect type (which
can only occur if the input effect type is $\dyn$), then an error is
raised.
Finally, we have the function downcast. Recall that a function cast
applied to a value itself is a value, and only reduces when applied to
a value.
When this occurs in a downcast, as shown, the result reduces
to applying the original function to an \emph{up}casted version of
the input and \emph{down}cast of the output, where this time we cast
both value and effect types.
Note the order of the value and effect casts on the output is
arbitrarily chosen: because value casts only affect values and effect
casts only affect effect operations, the two possible orders are
equivalent.
The elided cast for function upcasts is precisely dual, and finally
there is a trivial cast rule for the identity cast on booleans.

\begin{figure}
  \begin{mathpar}
    \inferrule
    {}
    { E[\hndl {V} {x} {N} {\phi}] 
    \stepsin {} E[N[V/x]]
    }
    \quad
    \inferrule
    {\effname \in \dom(\phi) \and E' \apart \effname}
    {E[\hndl 
      {E'[\raiseOpwithM {\effname} V]} x N {\phi}] \\
    \stepsin {}
    E[\phi(\effname)[V/x][(\lambda y.
      \hndl {(E'[y])} x N {\phi}
    )/k]]
    }

    \inferrule
    {}
    { E[(\lambda x. M) V] \stepsin {} E[M[V/x]] }

    \inferrule
    {}
    { E[\letXbeboundtoYinZ{V}{x}{M}] \stepsin {} E[M[V/x]] }


    \inferrule
    {}
    { E[ \upcast \sigma {\sigma'} {V} ]
    \stepsin {}
    E[ V ]
    }
    \and
    \inferrule
    { \effname \in \sigma' \and E' \apart \effname}
    { E[ \upcast \sigma {\sigma'} E'[
        \raiseOpwithM {\effname} {V}] ]
    \stepsin {} \\
    E[ \letXbeboundtoYinZ
      {\dncast {B} {B'} \raiseOpwithM{\effname}{\upcast {A} {A'} V}}
      {x}
      {\upcast {\sigma} {\sigma'} E'[x]} ]
    }

    \inferrule
    {}
    { E[ \dncast \sigma {\sigma'} {V} ]
    \stepsin {}
    E[ V ]
    }

    \inferrule
    { \effname \in \sigma \and E' \apart \effname}
    { E[ \dncast \sigma {\sigma'} E'[
        \raiseOpwithM {\effname} {V}] ]
    \stepsin {} \\
    E[ \letXbeboundtoYinZ
      {\upcast {B} {B'} \raiseOpwithM{\effname}{\dncast {A} {A'} V}}
      {x}
      {\dncast {\sigma} {\sigma'} E'[x]} ]
    }

    \inferrule
    { \effname \notin \sigma \and E' \apart \effname}
    { E[ \dncast \sigma {\dyn} E'[ 
        \raiseOpwithM {\effname} {V}] ]
    \stepsin {}
    \err
    }
    \and
    \inferrule
    {}
    { E[ (\dncast {(A \effto \sigma B)} {(A' \effto {\sigma'} B')} V_f)\, V ] 
    \stepsin {}\\\\
    E [\dncast B {B'} \dncast \sigma {\sigma'} (V_f \upcast {A} {A'} V) ]
    }
  \end{mathpar}
  \caption{Operational semantics of Core GrEff}
  \label{fig:operational-semantics}
\end{figure}

We conclude the operational semantics with the following theorem,
which establishes that the operational rules are all valid equational
reasoning principles in any system that models the inequational theory.
\begin{theorem}
  If $\cdot \vdash_{\emptyset} M, N : A$ and $M \stepsin{} N$ then $M
  \equiv N$ is provable in the axiomatic semantics.
\end{theorem}

The full proof is in the appendix, but we give an overview of how the
behavior of effect casts $\dncast {\sigma}{\sigma'} M$ is derived in
particular. The core of the argument is to show that the downcast is
equivalent to a particular handler, and then derive the operational
reductions from the $\beta$ reductions for handlers. The handler is \( \dncast \sigma {\sigma'} M \equiv \hndl M x x {\phi_{\dncast \sigma {\sigma'}}} \)
where the $\phi_{\dncast \sigma {\sigma'}}$ handles precisely the effects in $\sigma'$ and for each such $\effname @ A_\sigma' \leadsto B_\sigma' \in \sigma'$, the clause is defined as
\[
\phi_{\dncast\sigma\sigma'}(\effname) = 
\begin{cases}
  \err & \effname \not\in\dom(\sigma)\\
  k(\upcast{B_\sigma}{B_\sigma'}\raiseOpwithM \effname {\dncast {A_\sigma}{A_\sigma'} x}) & \effname @ A_\sigma \leadsto B_\sigma \in\sigma
\end{cases}
\]
That is, if the effect is not present in $\sigma$, the handler errors,
and otherwise it re-raises the effect to its context, but first
\emph{downcasting} the request, and \emph{upcasting} the received
response, before passing this back to the original continuation.
Then we show that $\dncast \sigma {\sigma'} M \equiv \hndl M x x
{\phi_{\dncast \sigma {\sigma'}}}$ by showing an ordering each
way. 
For the $\dncast \sigma {\sigma'} M \equiv \hndl M x x {\phi_{\dncast
    \sigma {\sigma'}}}$ case, we apply the effect forwarding principle
to transform the left-hand side to $\hndl {\dncast\sigma{\sigma'} M} x
x {\phi_{\sigma}}$. Then we apply congruence for handlers, with the
cases of the right-hand side that handle effects not in $\sigma$ being
irrelevant. Then the remaining clauses are all of the same syntactic
structure except for upcasts and downcasts, and so the proof follows
by congruence and the upcast/downcast rules.
To show $\hndl M x x {\phi_{\dncast \sigma {\sigma'}}} \ltdyn
\dncast{\sigma}{\sigma'} M$, we first apply the downcast right rule to
eliminate the cast on the right. Then to show $\hndl M x x
{\phi_{\dncast \sigma {\sigma'}}} \ltdyn M$ we again use the effect
forwarding principle to rewrite the right-hand side as $\hndl M x x
{\phi_{\sigma'}}$. We again apply handler congruence, with the cases
where $\effname\in \sigma$ analogous to the prior argument. In the
remaining remaining cases $\phi_{\dncast{\sigma}{\sigma'}}(\effname)
\ltdyn \phi_{\sigma'}(\effname)$ where $\effname \not\in\sigma$, we
have the left hand side is an error, and so the argument follows by
the fact that the error is the minimum in the ordering.

\subsection{Subtyping, Gradual Subtyping and Coercions}
\label{sec:semantics:casts}

The elaboration defined in Section~\ref{sec:syntax} inserts casts of the
form $\dncast{A}{\lceil A\rceil}\upcast{B}{\lceil B \rceil}M$ when a
gradual subtyping $A\gsubty B$ is used in the type-checker.
If we think of $\lceil A \rceil$ as the type of programs in the
untracked language, this says to cast a program from one type to
another, we should cast it to an untracked type and then to the other
effect-tracking type, similar to prior work on cast calculi based on
upcasts and downcasts \cite{newahmed18}. This is a reasonable cast if
we think of the untracked language as our ``operational ground
truth'', and so we should prove that any other translation is
extensionally equivalent to this one.
However, operationally, this can be quite a wasteful translation, as a
cast can result in proxying at runtime, while \emph{subtyping}
coercions have no runtime behavior, and so are zero cost.
For instance, if $A \gsubty B$ is true because in fact $A \subty B$,
then there need not be any runtime cast at all.
For this reason, we would prefer to optimize the cast based on the
subtyping information in the proof of $A\gsubty B$.
Since $A$ may be more imprecise than $B$ in some subterms and
vice-versa, the structure of the cast should still be an upcast
followed by a downcast, but with the possibility that we use implicit
subtyping coercions at some points.
There are three places we might insert the implicit subtyping
coercion: before the upcast, between the upcast and downcast and after
the downcast.
From the proof of $A \gsubty A'$, we can extract types and
subtyping/precision derivations as in
Figure~\ref{fig:gradual-subtyping-situation}.
\begin{wrapfigure}{r}{0.33\textwidth}
  \quad
    \begin{tikzcd}
	{A_h} & {D_h} & {B} \\
	A & {D_l} & {B_l}
	\arrow["\sqsubseteq"{description}, draw=none, from=2-1, to=2-2]
	\arrow["\sqsupseteq"{description}, draw=none, from=2-2, to=2-3]
	\arrow["\sqsubseteq"{description}, draw=none, from=1-1, to=1-2]
	\arrow["\sqsupseteq"{description}, draw=none, from=1-2, to=1-3]
	\arrow["\vsubty"{description}, draw=none, from=2-1, to=1-1]
	\arrow["\vsubty"{description}, draw=none, from=2-3, to=1-3]
	\arrow["\vsubty"{description}, draw=none, from=2-2, to=1-2]
\end{tikzcd}
  \caption{Situation derivable from $A \gsubty B$}
  \label{fig:gradual-subtyping-situation}
\end{wrapfigure}
On the left we have a ``pure subtyping'' component of the gradual
subtpying proof coming from $A$, and on the right we we have the pure
subtyping component coming from $B$. In the middle we have two
``dynamic'' types also related by subtyping.
There are then three paths from $A$ to $A'$ in this diagram, which
generate three different potential casts with implicit subtyping
coercions ensuring they are well-typed as taking $A$ to $A'$:
(1) Up and then right twice $\dncast{B}{D_h}\upcast{A_h}{D_h}$
(2) Right, up and then right: $\dncast{B}{D_h}\upcast{A}{D_l}$
(3) Right twice and then up:  $\dncast{B_l}{D_l}\upcast{A}{D_l}$
Fortunately we can choose whichever is operationally preferable: each
of these casts is equivalent as a function from $A$ to $B$ and they
are all equivalent to the ground truth cast $\dncast{A'}{\lceil A
  \rceil}\upcast{A}{\lceil A\rceil}$.
The above discussion applies equally well to effect casts, which are
even simpler in that the ``ground-truth'' always factors through the
single most imprecise effect type: the dynamic effect type.

\section{Soundness and Graduality}

In this section we establish that the axiomatic semantics of core
GrEff has a sound model in terms of its operational semantics. This
establishes two key properties: equivalent terms ($M \equiv N$) are
contextually equivalent in the operational semantics, and the
graduality property holds. First, we review the definition of the
graduality property, and then we give a logical relations model and
prove that any provable inequality $M \ltdyn N$ implies that the terms
are related in the logical relation.

\subsection{Static and Dynamic Gradual Guarantees}

GrEff is designed to support a smooth \emph{migration} from imprecise
to precise typing. The static gradual guarantee \cite{refined}
formalizes a syntactic element of this idea of a smooth migration. The
static gradual guarantee informally says that increasing the precision
of type annotations on a program can only make it \emph{harder} to
satisfy the static type checker, or viewed the other way around,
decreasing the precision of type annotations can only make it
\emph{easier} to satisfy the static type checker.
Then the dynamic gradual guarantee, also known as \emph{graduality},
establishes the semantic counterpart: increasing the precision of type
annotations on a program should only make it \emph{harder} to
terminate without a \emph{dynamic} type error, and furthermore except
where there are dynamic type errors, the behavior of the program
should match the original.
These properties can be formalized as a form of \emph{monotonicity} of
the elaboration of the syntactic programs of surface GrEff into the
semantically meaningful core GrEff programs as follows.
First, we define a syntactic term precision ordering
$\mathrel{\ltdyn^\textrm{syn}}$ on untyped GrEff programs as the
congruence closure of the type precision ordering. Then the static
gradual guarantee says that this is a monotone \emph{partial} function
from the syntactic term precision ordering to the axiomatic inequality
on core GrEff terms:
\begin{theorem}[Static Gradual Guarantee]
  If $P \mathrel{\ltdyn^{\textrm{syn}}} P'$, then if $\elabprog \cdot
  \cdot P \sig M \sigma A$, then there exist $M', \sigma', A'$ such
  that $\elabprog \cdot \cdot {P'} \sig {M'} {\sigma'} {A'}$ such that
  $\cdot \vdash_{\sigma \ltdyn \sigma'} M \ltdyn M' : A \ltdyn A'$.
\end{theorem}


Then the dynamic gradual guarantee says that this extends to
monotonicity in the following \emph{semantic} ordering on core GrEff
terms:
\begin{definition}[Error Ordering on Closed Programs]
  Given $\cdot \vdash_{\emptyset} M, M' : \boolty$, define
  $M \mathrel{\ltdyn^\textrm{sem}} M'$ to
  hold when one of the following is satisfied (1) $M \stepstar \mho$,
  (2) $M \Uparrow$ and $M' \Uparrow$, (3) $M \stepstar \tru$ and $M'
  \stepstar \tru$ (4) $M \stepstar \fls$ and $M' \stepstar \fls$.
\end{definition}
\begin{theorem}[Dynamic Gradual Guarantee]
  If $\sig\pipe\cdot \vdash_{\emptyset\ltdyn\emptyset} M \ltdyn M' : \boolty$,
  then $M \mathrel{\ltdyn^\textrm{sem}} M'$.
\end{theorem}

This theorem is stated in terms of closed terms of a fixed type, but
to prove it we need a stronger inductive hypothesis, i.e., the logical
relation for open terms.
The resulting theorem that any inequation provable in the theory
implies the semantic ordering is called \emph{graduality}, as it is
analogous in structure to the \emph{parametricity} theorem in
parametric polymorphism.
Then the dynamic gradual guarantee follows as a corollary.

\subsection{Logical Relation}

In Figure \ref{fig:lr}, we present the definition of the step-indexed logical
relation for graduality.
Following prior work on logical relations for graduality, the relation
is indexed not by types, but by \emph{derivations} of type precision
facts, i.e., proof terms for $c : A \ltdyn B$ or $d_\sigma : \sigma
\ltdyn \sigma'$.
We present the definition of these proof terms in the appendix.
For a type precision derivation $d$, define $d^l$ and $d^r$ to be the
types such that $d : d^l \ltdyn d^r$, and analogously for effect
types.

\begin{figure}[h!]
  \begin{mathpar}
    \begin{array}{rcl}
      (M_1, M_2) \in (\later R)_j & \iff &
        j = 0 \vee (j = k + 1 \wedge (M_1, M_2) \in R_k)\\

      (V_1,V_2) \in \simivrel{\boolty} j & \iff &
        (V_1, V_2) \in \valatomlhs{\boolty} \wedge (V_1 = V_2 = \tru) \vee (V_1 = V_2 = \fls)\\

      (V_1,V_2) \in \simivrel{d_i \to_{d_\sigma} d_o} j & \iff &
        (V_1, V_2) \in \valatomlhs {d_i \to_{d_\sigma} d_o} \wedge \\&&
        \forall k \leq j. 
        \forall (V_{i1}, V_{i2}) \in \simivrel{d_i} k.\\
        &&(V_1 \,V_{i1}, V_2\,V_{i2}) \in \simierel{d_\sigma} k (\simivrel{d_o}{})\\

      (M_1, M_2) \in \ltierel{d_\sigma} j (R, A^l, A^r) & \iff &
        (M_1, M_2) \in \termatomlhs {A^l} {A^r} {d_\sigma} \wedge (M_1 \stepsin {j+1} \\&&
        \vee (\exists k \leq j.~ (M_1 \stepsin{j-k} \err)\\&&\quad\vee
        (\exists (N_1,N_2) \in \ltirrel{d_\sigma} k R \wedge M_1 \stepsin{j-k} N_1 \wedge M_2 \stepstar N_2)))\\

      (M_1, M_2) \in \gtierel{d_\sigma} j (R, A^l, A^r) & \iff &
        (M_1, M_2) \in \termatomlhs {A^l} {A^r} {d_\sigma} \wedge (M_2 \stepsin {j+1} \\&&
\vee(\exists k \leq j. (M_2 \stepsin{j-k} \err \wedge M_1 \stepstar \err)\\&&\quad\vee
        (\exists N_2. M_2 \stepsin{j-k} N_2 \wedge M_1 \stepstar \err)\\&&\quad\vee
        (\exists (N_1,N_2) \in \gtirrel{d_\sigma} k R \wedge M_2 \stepsin{j-k} N_2 \wedge M_1 \stepstar N_1)))\\

      (M_1, M_2) \in \simirrel{d_\sigma} j (R, A^l, A^r) & \iff &
        (M_1, M_2) \in \termatomlhs {A^l} {A^r} {d_\sigma} \wedge \\&&
        ((\texttt{val}(M_1) \wedge \texttt{val}(M_2) \wedge (M_1, M_2) \in R j)
        \\&& \quad \vee
        (\exists \epsilon : c \leadsto d \in d_\sigma,
        E^l \apart \epsilon, E^r \apart \epsilon, V^l, V^r.\\&& \quad
        (V^l,V^r)  \in (\later \simivrel{c}{})_j      \wedge \\&& \quad
        (x^l.E^l[x^l], x^r.E^r[x^r]) \in (\later \simikrel {d }{})_j \\ && \quad\quad
           {(\simierel{d_\sigma} {} {(R, A^l, A^r)},
             \RbngT{d_\sigma^l}{A^l},
             \RbngT{d_\sigma^r}{A^r})}     \wedge \\&& \quad
        M_1 = E^l[\raiseOpwithM{\epsilon}{V^l}] \wedge 
        M_2 = E^r[\raiseOpwithM{\epsilon}{V^r}])) \\

      (x^l.M^l, x^r.M^r) \in \qquad & \iff &
        (M^l, M^r) \in \ecatomlhs {c} {\RbngT{\sigma^l}{A^l}} {\RbngT{\sigma^r}{A^r}} \wedge \\
        \simikrel{c} j (S, \RbngT{\sigma^l}{A^l}, \RbngT{\sigma^r}{A^r}) &&
        \forall k \leq j. (V^l,V^r) \in \simivrel{c} k.
        (M^l[V^l/x^l], M^r[V^r/x^r]) \in S_k \\
      (\gamma_1,\gamma_2) \in \simigrel{\Gamma^\ltdyn} j & \iff &
        \forall(x_1 \ltdyn x_2 : c) \in \Gamma^\ltdyn.
        (\gamma_1(x_1), \gamma_2(x_2)) \in \simivrel c j\\
      
    \end{array}
  \end{mathpar}
  \caption{Logical Relation}
  \label{fig:lr}
\end{figure}

Many of the details are similar to prior work, especially
\cite{DBLP:journals/pacmpl/NewJA20}, so we highlight the handling of
effect types, which is novel. 
in Section \ref{sec:LR}.  In addition to the usual expression and
value relations, we have a \emph{result} relation and a
\emph{continuation} relation.  In our language, a \textbf{result} is
either a value, or an evaluation context $E$ wrapping a raise of an
effect $\effname$, such that $E \apart \effname$.  The result relation
specifies the conditions for two such results to be related.  Finally,
the relations are parameterized by precision derivations. In the case
of the expression and result relations, this is an effect precision
derivation, while for values and continuations, it is a value type
precision derivation.  This is analogous to the usual approach whereby
logical relation is indexed by a type. But instead of using types, we
use precision derivations, i.e., the proof that the LHS term is more
precise than the type of the RHS term.

As in previous work on logical relations for graduality, the
expression logical relation $\simierel{\cdot}{}{}$ is split into two
relations $\ltierel{\cdot}{}{}$ and $\gtierel{\cdot}{}{}$. The former
counts the steps taken by the left-hand term, while the latter counts
steps taken by the right-hand term. This is captured by the
quantitative small-step reduction $M \stepsin j N$ which means $M$
takes exactly $j$ steps to reduce to $N$. Despite needing two
relations, we are for the most part able to abstract over their
differences: most of the lemmas we prove hold for both relations with
no adjustment needed. Notable exceptions are transitivity and the
anti- and forward reduction lemmas: these lemmas make crucial use of
step counting, so naturally the side whose steps we are counting makes
a difference.

Given a step-indexed relation $R$, we define an operator $\later R$ (pronounced ``later $R$'')
as follows: Terms $M_1$ and $M_2$ are related in $\later R$ at index $n$ if and only if either
$n$ is zero, or $n \ge 1$ and $M_1$ and $M_2$ are related in $R$ at index $n - 1$.

One novel aspect of our logical relation is the result relation $\simirrel{\cdot}{}{}$.
This relation relates terms $M_1$ and $M_2$ -- of type $A^l$ and $A^r$ respectively -- 
representing either two values or two ``evaluations'' of raised operations. The relation
is parameterized by a step-indexed relation $R$ between \textit{values} of type $A^l$ and
$A^r$ (the types of $M_1$ and $M_2$). $M_1$ and $M_2$ are related by $\simirrel{\cdot}{}{}$
when either (1) both terms are values and are related by $R$ at the appropriate step index,
or (2) there exists an effect $\epsilon$ in $d_\sigma$, values related later, and evaluation
contexts (i.e., continuations -- see below) related later, such that $M_1$ is equal to
raising the effect and then wrapping it in the continuation, and likewise for $M_2$.

The relation $\simikrel{\cdot}{}{}$ relates evaluation contexts $E_1$ and $E_2$ representing continuations that accept values, similar to prior work on logical relations for continuations \cite{DBLP:conf/sfp/Asai05}.
The evaluation contexts each have a hole $\bullet$ -- the type of the hole is the type of the input value to the continuation.
To enforce that the continuations accept \textit{values} only, and not arbitrary terms, the inputs to the continuation
relation are actually \textit{terms} $M^l$ and $M^r$ with free variables $x^l$ and $x^r$, respectively.
$E_1$ and $E_2$ also have ``output'' types ($A^l$ and $A^r$) and ``output'' effect sets ($\sigma^l$ and $\sigma^r$).
When values are plugged into $E_1$ and $E_2$, the result is two terms having types $A^l$ and $A^r$ and
effect sets $\sigma^l$ and $\sigma^r$, respectively.


\subsection{Proof of Graduality}

Our goal is to prove that the inequational theory is sound with
respect to the logical relation.
First we define the notion of two terms being related semantically:
\begin{small}
  \[
  \Gamma^\ltdyn \vDash_{d_{\sigma}} M_1 \ltdyn M_2 \in c := \forall \sim\, \in \{<,\, >\}. \forall j \in \mathbb{N}. 
  \forall (\gamma_1, \gamma_2) \in \simivrel{\Gamma} {j}. 
  (M_1[\gamma_1], M_2[\gamma_2]) \in \simierel{d_{\sigma}} {j} {\simivrel c {}}.
  \]
\end{small}


That is, $M_1$ and $M_2$ are related if for all $j$ and all
substitutions of values $\gamma_1$ and $\gamma_2$ related at $j$,
the resulting terms are related in $\simierel {d_\sigma} {j} {\simivrel c {}}$,
where this needs to hold both when $\sim$ is $<$ and when it is $>$.
Our goal is then to prove the following:

\begin{theorem}[Graduality]
  If $\Gamma^\ltdyn \vdash_{d_\sigma} M \ltdyn N : c$ then
  $\Gamma^\ltdyn \vDash_{d_\sigma} M \ltdyn N \in c$
\end{theorem}

We provide here a high-level overview of the proof; the complete
proofs are in the appendix.
We begin by establishing variants of standard anti- and
forward-reduction lemmas as well as monadic bind.
We also prove a L\"{o}b induction principle to structure the induction
over step-indices.
With these lemmas, we first prove soundness of each of the congruence
rules for term precision, by uses of the monadic bind lemma along with
the reduction lemmas. Next, we prove soundness of the rules of the
equational theory, e.g., the $\beta$ and $\eta$ laws, and
transitivity. Finally, we prove soundness of the rules for casts and
subtyping.

\section{Discussion}

\paragraph{Prior Work on Gradual Effects}

The most significant prior work on gradual effects is
\citet{gradeffects2014}, which defined a gradual effect system based
on the generic effect calculus of \citet{DBLP:conf/tldi/MarinoM09}
using an early version of the \emph{abstracting gradual typing} (AGT)
framework for gradual type systems\cite{AGT}.
While we based GrEff on effect handlers rather than the generic effect
calculus, there are significant similarities in the typing: function
types and typing judgments are indexed by a set of effect operations
in each system.
The most significant syntactic difference is that their framework is
parameterized by a fixed effect theory, whereas GrEff has explicit
support for declaration of new effects in the program. In particular,
this means that their system does not need to support modules
containing different views of the same nominal effect as we did.
They additionally support a form of partially tracked functions, in
GrEff syntax this would look like $A \to_{\effname,\dyn} B$, a
function type where the function is known specifically to possibly
raise the effect $\effname$ in addition to raising other effects.
In GrEff this partial tracking would ensure that any effects raised
with the name $\effname$ match the module's local view of the effect
typing of $\effname$.
Finally, on the semantic side, this prior work proves only a type
safety proof, whereas here we have proven graduality and the
correctness of type-based optimizations and handler optimizations.

Another related area of research is on gradual typing with delimited
continuations, which are mutually expressible with effect handlers
\cite{DBLP:journals/jfp/0002KLP19,DBLP:conf/rta/PirogPS19}.
Takikawa et al (\cite{DBLP:conf/esop/TakikawaST13}) propose a gradual type system and
semantics via contracts for a language with delimited continuations
using typed prompts.
They consider only value types and untracked function types that do
not say which prompts are expected to be present.
They show that a naive contract based implementation is unsound
because a dynamically typed program can interact with a typed prompt
and thereform the prompts themselves must be equipped with contracts,
even though it does not correspond to any value being imported.
In core GrEff, this unsoundness is ruled out by using intrinsic
typing: the problem corresponds to raising an effect operation with a
different type than the type expected by the closest handler, which is
precisely what the effect type system tracks. Wrapping the prompt in
contracts is behaviorally equivalent to what is achieved by our effect
type casts.
Sekiyama, Ueda and Igarashi present a blame calculus for a language
with shift and reset \cite{DBLP:conf/aplas/SekiyamaUI15}. The blame calculus
is analogous to our core GrEff language, and uses a type and effect
system for the answer types of shift/reset. They do not develop a
surface language that elaborates to this blame calculus like our
GrEff, and there is no analogue of effect operations in
shift/reset-based systems so there are no nominal aspects of their
language. Additionally, while they have an effect system to keep track
of answer types, they do not have effect casts.

\paragraph{Prior Approaches to Gradual Nominal Datatypes}

We are also not the first to consider the combination of gradual and
nominal typing.
The closest match to our design is in Typed Racket's support for typed
structs.
In Racket, a struct is a kind of record type that (by default) is
generative in that it creates a new type tag distinct from all others.
Typed Racket supports import of untyped Racket structs into Typed
Racket, where types are assigned to the fields, and values of the
struct type are then wrapped in contracts accordingly.
This is quite close to our treatment of nominal effect operations
which can be thought of as adding new cases to the dynamic effect
monad rather than dynamic type.
Our type system is more complex however, since in our system modules
can use dynamically typed effects whereas in Typed Racket, there is no
syntactic type for dynamially typed values, when imported into typed
code the system must give a completely precise type.
Malewski et al (\cite{DBLP:journals/pacmpl/MalewskiGT21}) present a
design for gradual typing with nominal algebraic datatypes. Their
focus is on the gradual migration from datatypes whose cases are
open-ended to datatypes with a fixed set of constructors. They do not
consider the use-case we have where different modules have different
typings for the same nominal constructor.

\paragraph{Prior Work on Subtyping}

Much prior work on incorporating subtyping with gradual types has
focused on the static typing aspects
\cite{wadler-findler09,siek-taha07,DBLP:conf/popl/GarciaC15,10.1145/3290329}. The
most significant prior semantic work on subtyping and gradual typing
is the Abstracting Gradual Typing work \cite{AGT} which proves the
dynamic gradual guarantee for a system with subtyping developed using
the AGT methodology. In this work we establish equivalence between
multiple different ways to combine gradual type casts and subtyping
coercions, summarized in Figure~\ref{fig:gradual-subtyping-situation},
which are derivable from our newly identified cast/coercion ordering
principle in our equational theory
(Figure~\ref{fig:inequational-theory}).

\paragraph{Towards a Practical Language Design}

GrEff is intended as a proof-of-concept language design to provide the
semantic foundation for extending a language such as OCaml 5 with
gradual effect typing. We discuss the current mismatches with OCaml's
design and how these might be rectified.
First, OCaml uses \emph{extensible variant} types for effects and
exceptions, whereas in GrEff effects are not first-class values. This
should not be difficult to support as the variant type can be treated
somewhat similarly to a dynamic type.
Next, OCaml supports \emph{recursive} effect types, meaning that the
request or response of an effect can refer to the effect being
defined. For instance, this allows for a variant of our coroutine
example where forked threads can fork further threads. This would
complicate the metatheory of GrEff but shoul work in principle.
A final syntactic difference is that OCaml is based on
Hindley-Milner-style polymorphic type schemes, whereas GrEff is based
on a simple type system. It may be possible to adapt previous work for
gradual typing in unification-based type
systems\cite{DBLP:conf/dls/SiekV08,DBLP:conf/popl/GarciaC15,10.1145/3290329}.

Implementing gradual effects brings its own challenges. Our derivation
of the operationsl semantics is based on proving that effect casts can
be implemented as handlers, and so can be implemented by a
source-to-source transformation. However, such an implementation may
suffer from similar performance issues as other naive wrapper
semantics, which can be solved by defunctionalizing the casts
\cite{spaceefficient}. Additionally, strong gradual typing between
fully dynamically typed and static code can result in high performance
penalties \cite{10.1145/2837614.2837630} even with space efficient
implementations. However since effect casts would not be as pervasive
in typical programs as value type casts, it is not obvious that the
same pathological behaviors would arise in gradually effect typed
OCaml programs. This is a clear empirical question to be addressed in
future work.

\bibliographystyle{ACM-Reference-Format}
\bibliography{gradual}
\clearpage
\appendix
\section{Syntax and Elaboration}

We give a term assignment for effect precision in
Figure~\ref{fig:precision-term-assignment}.
In it we use the notion of an effect operation being in a precision
derivation $\effarr \effname c d \in d_c$. For when $d_c$ itself is a
partial function this is just as with earlier usage, but when $d_c =
\dyn$ or $d_c = \texttt{inj}(d_c')$ we use the definition at the
bottom of the figure.

\begin{figure}
  \begin{mathpar}
      \inferrule
      {}
      {\sig\vdash \boolty : \boolty \ltdyn \boolty}
  
      \inferrule
      {\sig\vdash d_i : A \ltdyn A'\and \sig\vdash d_e : \sigma \ltdyn \sigma'\and \sig\vdash d_o : B \ltdyn B'}
      {\sig\vdash d_i \to_{d_e} d_o : A \to_\sigma B \ltdyn A' \to_{\sigma'} B'}
  
      \inferrule
      {}
      {\sig \vdash \dyn : \dyn \ltdyn \dyn}

      \inferrule
      {\textrm{supp}(d_{c}) = \textrm{supp}(\sigma_c) = \textrm{supp}(\sigma_c')\\\\
        (\forall \effarr \effname c d \in d_c, \effarr \effname A B \in \sigma_c, \effarr \effname {A'}{B'} \in \sigma_c'.\\\\ \quad
        \sig \vdash c : A \ltdyn A' \and \sig \vdash B \ltdyn B')
      }
      {\sig \vdash d_{c} : \sigma_c \ltdyn \sigma_c' }

      \inferrule
      {\sig \vdash d_c : \sigma_c \ltdyn \sig|_{\textrm{supp}(\sigma_c)}}
      {\sig \vdash \texttt{inj}(d_c) : \sigma_c \ltdyn \dyn}
    \end{mathpar}

  \begin{mathpar}
    \inferrule
    {\effarr \effname c d \in \in \sig}
    {\sig \vdash \effarr \effname c d \in \dyn}

    \inferrule
    {\sig \vdash \effarr \effname {c'} {d'} \in d_c \and c = \texttt{inj}(c') \and d = \texttt{inj}(d')}
    {\sig \vdash \effarr \effname c d \in \texttt{inj}(d_c)}
  \end{mathpar}
  \caption{Type and Effect Precision Derivations}
  \label{fig:precision-term-assignment}
\end{figure}

Thought the generating axioms are different from the simple
presentation in the body of the paper, we show that provability is not
affected:
\begin{lemma}[Correctness of Term Assignment]
  Assuming $\sig \vdash A$ and $\sig \vdash B$, the following are equivalent
  \begin{itemize}
  \item $A \ltdyn A'$ is provable in the system in Figure~\ref{fig:type_effect_precision_subtyping}
  \item There exists a derivation $\sig \vdash c : A \ltdyn A'$ in the  system in Figure~\ref{fig:precision-term-assignment}.
  \end{itemize}
  Similarly, assuming $\sig \vdash \sigma$ and $\sig \vdash \sigma'$, the following are equivalent
  \begin{itemize}
  \item $\sigma \ltdyn \sigma'$ is provable in the system in Figure~\ref{fig:type_effect_precision_subtyping}
  \item There exists a derivation $\sig \vdash c_e : \sigma \ltdyn \sigma'$ in the  system in Figure~\ref{fig:precision-term-assignment}.
  \end{itemize}
\end{lemma}

Next we define gradual join and meet of value and effect types in
Figure~\ref{fig:gjoin}. Note that the definition is quite simple for
concrete effect sets because this is only used on effects
\emph{within} the same module, so we never have to consider the case
where the two sides assign different effects to the same operation
name $\effname$.

\begin{figure}
    \begin{align*}
      \boolty \gjoin \boolty &= \boolty\\
      (A \effto \sigma B) \gjoin (A' \effto {\sigma'} B') &=
      (A \gmeet A') \effto {\sigma \gjoin \sigma'} (B \gjoin B')\\
      \dyn \gjoin \sigma &= \sigma\\
      \sigma \gjoin \dyn &= \sigma\\
      \sigma_c \gjoin \tau_c &=
      \{ \effname @ A \leadsto B \,|\, \effname @ A \leadsto B \in \sigma_c \wedge \effname \not\in\dom(\tau_c) \}\\
      &\quad\cup \{ \effname @ A' \leadsto B' \,|\,  \effname \not\in\dom(\sigma_c) \wedge \effname @ A' \leadsto B' \in \tau_c \} \\
      &\quad\cup
      \{ \effname @ A \gjoin A' \leadsto B \gmeet B' \,|\, \effname @ A \leadsto B \in \sigma_c \wedge \effname @ A' \leadsto B' \in \tau_c \}
    \end{align*}
    \begin{align*}
      \boolty \gjoin \boolty &= \boolty\\
      (A \effto \sigma B) \gmeet (A' \effto {\sigma'} B') &=
      (A \gjoin A') \effto {\sigma \gmeet \sigma'} (B \gmeet B')\\
      \dyn \gmeet \sigma &= \dyn\\
      \sigma \gmeet \dyn &= \dyn\\
      \sigma_c \gmeet \tau_c &=
      \{ \effname @ A \gmeet A' \leadsto B \gjoin B' \,|\, \effname @ A \leadsto B \in \sigma_c \wedge \effname @ A' \leadsto B' \in \tau_c \}
    \end{align*}
  \caption{Gradual Join and Meet}
  \label{fig:gjoin}
\end{figure}

Now we define a notion of subtyping of precision derivations, which will be needed
in the proofs involving the interaction between subtyping and casts.

\begin{figure}
  \begin{mathpar}
    \inferrule
    {}
    {\boolty \subty \boolty}

    \inferrule
    {d_i \subty c_i \and c_e \subty d_e \and c_o \subty d_o}
    {c_i \to_{c_e} c_o \subty d_i \to_{d_e} d_o}

    \inferrule
    {}
    {\dyn \subty \dyn}

    \inferrule
    {c \subty d}
    {\texttt{inj}(c) \subty \texttt{inj}(d)}

    \inferrule
    {\dom(d_c) \subseteq \dom(d'_c) \\\\
      \forall \effarr \effname c d \in d_c.
      \effarr \effname {c'} {d'} \in d'_c \wedge
      c \subty c' \wedge d' \subty d}
    {d_c \subty d'_c}

    \inferrule
    {c \subty \texttt{inj}(\sig)}
    {c \subty \dyn}

    \inferrule
    {c \subty d}
    {c \subty \texttt{inj}(d)}
  \end{mathpar}
  \caption{Subtyping of Precision Derivations}
\end{figure}

\section{(In)Equational Theory}

In this section we describe the full inequational theory and then
prove several derivable theorems in the theory.

Note that for brevity, we use some shorthands: rather than writing out
the full $\sg^\ltdyn \vdash_{\sigma\ltdyn \tau} M \ltdyn N : A \ltdyn
B$, (1) we elide $\sg^\ltdyn$, and all rules should be interpreted as
holding under an arbitrary such contexts (2) rather than write
$\sigma\ltdyn\tau$ and $A \ltdyn B$, we use instead precision
derivations $d_\sigma$, $c$ and (3) whenever it is clear, we elide the
types as well, especially for equational rules.

First we need general call-by-value reasoning principles.
\begin{mathpar}
  \inferrule*[right=ValSubst]
  {M[x:A] \equiv N[x:A] \and V \equiv V' : A}
  {M[V/x] \equiv N[V'/x]}
  
  \inferrule*[right=MonadUnitL]{}{\letXbeboundtoYinZ y x N \equiv N[y/x]}

  \inferrule*[right=MonardUnitR]{}{\letXbeboundtoYinZ M x x \equiv M}

  \inferrule*[right=MonadAssoc]{}{\letXbeboundtoYinZ {(\letXbeboundtoYinZ M x N)} y P \equiv \letXbeboundtoYinZ M x \letXbeboundtoYinZ N y P}

  \inferrule*[right=BoolEta]{}{M[x : \boolty] \equiv \ifXthenYelseZ x {M[\tru/x]}{M[{\fls/x}]}}

  \inferrule*[right=BoolBetaTru]{}{\ifXthenYelseZ \tru {N_t}{N_f} \equiv N_t}

  \inferrule*[right=BoolBetaFalse]{}{\ifXthenYelseZ \fls {N_t}{N_f} \equiv N_f}

  \inferrule*[right=IfEval]{}{\ifXthenYelseZ {M} {N_t}{N_f} \equiv \letXbeboundtoYinZ M x \ifXthenYelseZ x {N_t} {N_f}}

  \inferrule*[right=FunBeta]{}{(\lambda x. M) V \equiv M[V/x]}

  \inferrule*[right=FunEta]{}{(V : A \to B) \equiv \lambda x. V x}

  \inferrule*[right=AppEval]{}{M\,N \equiv \letXbeboundtoYinZ M x \letXbeboundtoYinZ N y x\,y}
\end{mathpar}

Next, the rules specifically for raise and handlers:
\begin{mathpar}
  \inferrule*[right=HandleBetaRet]{}{\hndl x y M \phi \equiv M[x/y]}

  \inferrule*[right=HandleBetaRaise]{}
    {\hndl {({\letXbeboundtoYinZ {\raiseOpwithM \effname x} o {N_k}})} y M \phi \equiv \\\\
    \phi(\effname)[\lambda o. \hndl {N_k} y M \phi/k]}

  \inferrule*[right=RaiseEval]{}{\raiseOpwithM \effname M \equiv \letXbeboundtoYinZ M x \raiseOpwithM \effname x}


  \inferrule*[right=HandleEmpty]
  {}
  {\hndl M x N {\emptyset} \equiv \letXbeboundtoYinZ M x N}

  \inferrule*[right=HandleExt]
  {\forall\effname \in \dom(\phi).~ \psi(\effname) = \phi(\effname)
    \and
    \forall \effname \in \dom(\psi). \effname\not\in\dom(\phi) \Rightarrow \psi(\effname) = k(\raiseOpwithM \effname x)
  }
  {\hndl M y N \phi \equiv \hndl M y N \psi}
\end{mathpar}

Next, the congruence rules
\begin{mathpar}

  \inferrule*[right=VarCong]
  {  x_1 \ltdyn x_2 : c \in \Gamma^\ltdyn}
  {\sg^\ltdyn\vdash_{d_\sigma} x_1 \ltdyn x_2 : c}
  
  \inferrule*[right=TrueCong]
  { }
  { \vdash_{d_\sigma} \tru \ltdyn \tru : \boolty}
  
  \inferrule*[right=FalseCong]
  { }
  { \vdash_{d_\sigma} \fls \ltdyn \fls : \boolty}

  \inferrule*[right=LambdaCong]
  {
   x_1 \ltdyn x_2 : c \vdash_{d_{\sigma'}} M_1 \ltdyn M_2 : d
  }
  { \vdash_{d_\sigma}
  \lambda x_1.M_1 \ltdyn \lambda x_2.M_2 : c \to_{d_{\sigma'}} d}

  \inferrule*[right=AppCong]
  {
   \vdash_{d_\sigma} M_1 \ltdyn M_2 : c \to_{d_\sigma} d \and
   \vdash_{d_\sigma} N_1 \ltdyn N_2 : c
  }
  { \vdash_{d_\sigma} M_1\,N_1 \ltdyn M_2\,N_2 : d
  }

   \inferrule*[right=LetCong]
   {
    \vdash_{d_\sigma} M_1 \ltdyn M_2 : c \\\\
    x_1 \ltdyn x_2 : c \vdash_{d_\sigma} N_1 \ltdyn N_2 : d
   }
   { \vdash_{d_\sigma}
   \letXbeboundtoYinZ {M_1} {x_1} {N_1} \ltdyn
   \letXbeboundtoYinZ {M_2} {x_2} {N_2} : d
   }

  \inferrule*[right=IfCong]
  { \vdash_{d_\sigma} M \ltdyn M' : \boolty \\\\
    \vdash_{d_\sigma} N_t \ltdyn N'_t : c
   \and  \vdash_{d_\sigma} N_f \ltdyn N'_f : c
  }
  { \vdash_{d_\sigma}
  \ifXthenYelseZ {M} {N_t}{N_f} \ltdyn
  \ifXthenYelseZ {M'} {N'_t}{N'_f} : c
  }

  \inferrule*[right=RaiseCong]
  {c : A_1 \ltdyn A_2 \and d : B_1 \ltdyn B_2\\\\
  \effname @ c \leadsto d \in d_\sigma \and  \vdash_{d_\sigma} M_1 \ltdyn M_2 : c 
  }
  { \vdash_{d_\sigma}
   \raiseOpwithM {\effname} {M_1} \ltdyn
   \raiseOpwithM {\effname} {M_2} : d
  }

  \inferrule*[right=HandleCong]
  {\vdash_{d_\sigma}M \ltdyn M' : c\and
    y:c \vdash_{d_\tau} N \ltdyn N' : d \\\\
    \forall \effname @ d_i \leadsto d_o \in d_\sigma.
    (\effname \notin \dom(\phi) \wedge \effname \notin \dom(\phi') \wedge 
     \effname : d_i \leadsto d_o \in d_\tau) \vee \\\\ 
    x:d_i, k : d_o \effto {d_\tau} d \vdash_{d_\tau} \phi(\effname) \ltdyn \phi'(\effname) : d
  }
  {\vdash_{d_\tau} \hndl {M} y N \phi\ltdyn \hndl {M'} y {N'} {\phi'} :  d}
\end{mathpar}

Next, the rules for errors
\begin{mathpar}
  \inferrule*[right=ErrBot]
  {\vdash_{{d_\sigma}^r} M :  {c^r}}
  {\vdash_{d_\sigma} \err \ltdyn M : c}

  \inferrule*[right=ErrStrict]{}{E[\err] \equiv \err}
\end{mathpar}

The generic rules for casts
\begin{mathpar}
  \inferrule*[right=ValUpL]
  {\vdash_{d_\sigma} M \ltdyn N : (c : A \ltdyn B) \and c : A \ltdyn A}
  {\vdash_{d_\sigma} \upcast A B M \ltdyn N : B}
  
  \inferrule*[right=ValUpR]
  {\vdash_{\sigma} M : A \and c : A \ltdyn B}
  {\vdash_{\sigma} M \ltdyn \upcast A B M : c}
  
  \inferrule*[right=ValUpEval]{}{\upcast A B M \equiv \letXbeboundtoYinZ M x \upcast A B x\and}

  \inferrule*[right=ValDnL]
  {c : A\ltdyn B \and \vdash_{\sigma} N : B}
  {\vdash_{\sigma} \dncast A B N \ltdyn N : c}

  \inferrule*[right=ValDnR]
  {\vdash_{d_\sigma} M \ltdyn N : (c: A \ltdyn B)}
  {\vdash_{d_\sigma} M \ltdyn \dncast A B N : A}

  \inferrule*[right=ValDnEval]{}{\dncast A B M \equiv \letXbeboundtoYinZ M x \dncast A B x}

  \inferrule*[right=ValUpL]
  {\vdash_{d_\sigma} M \ltdyn N : c \and d_{\sigma} : \sigma \ltdyn \tau}
  {\vdash_\tau \upcast \sigma \tau M \ltdyn N : c}
  
  \inferrule*[right=ValUpR]
  {\vdash_{\sigma} M : A \and d_\sigma : \sigma \ltdyn \tau}
  {\vdash_{d_\sigma} M \ltdyn \upcast \sigma\tau M : c}

  \inferrule*[right=EffDnL]
  {d_\sigma : \sigma\ltdyn \tau \and \vdash_{\tau} N : A}
  {\vdash_{d_\sigma} \dncast \sigma \tau N \ltdyn N : A}

  \inferrule*[right=EffDnR]
  {d_\sigma : \sigma \ltdyn \tau
    \vdash_{d_\sigma} M \ltdyn N : c}
  {\vdash_{\sigma} M \ltdyn \dncast \sigma \tau N : c}
\end{mathpar}

And the subtyping rules
\begin{mathpar}
  \inferrule*[right=SubtyMon]
  {\vdash_{d_\sigma} M \ltdyn N : c\and
    d_\sigma : \sigma \ltdyn \tau\and
    c : A \ltdyn B\\\\
    d_\sigma' : \sigma' \ltdyn \tau'\and
    c' : A' \ltdyn B'\\\\
    \sigma \subty \sigma'\and
    A \subty A' \and
    \tau \subty \tau' \and
    B \subty B'
  }
  {\vdash_{d_\sigma'}M \ltdyn N : c'}

  \inferrule*[right=ValUpSub]
  {c : A \ltdyn B \and c' : A' \ltdyn B' \and c \subty c' \and \vdash_{\sigma} M : A}
  {\vdash_{\sigma}\upcast A B M \equiv \upcast {A'}{B'} M : {B'}}

  \inferrule*[right=ValDnSub]
  {c : A \ltdyn B \and c' : A' \ltdyn B' \and c \subty c' \and \vdash_\sigma N : B}
  {\vdash_\sigma \dncast A B N \equiv \dncast {A'}{B'} N : \compty \sigma {A'}}

  \inferrule*[right=EffUpSub]
  {c_\sigma : \sigma \ltdyn \tau \and c' : \sigma' \ltdyn \tau' \and c_\sigma \subty c'_\sigma \and \vdash_\sigma M : A}
  {\vdash_{\tau'} \upcast \sigma \tau M \equiv \upcast {\sigma'}{\tau'} M : A}

  \inferrule*[right=EffDnSub]
  {c_\sigma : \sigma \ltdyn \tau \and c' : \sigma' \ltdyn \tau' \and c_\sigma \subty c'_\sigma \and \vdash_\tau N : A}
  {\vdash_{\sigma'}\dncast \sigma \tau N \equiv \dncast {\sigma'}{\tau'} N : A}
\end{mathpar}

In Figure~\ref{fig:uniq}, we list some derivable reasoning principles
for our inequational theory, which follow by analogous proofs to prior
work.

\begin{figure}
  \begin{mathpar}
    \upcast{A}{A}M \equiv M\and
    \upcast{\sigma}{\sigma}M \equiv M\and
    \dncast{A}{A}M \equiv M\and
    \dncast{\sigma}{\sigma}M \equiv M\and

    \upcast{B}{C}\upcast{A}{B} M \equiv \upcast{A}{C} M \and
    \dncast{A}{B}\dncast{B}{C} M \equiv \dncast {A} {C} M \and
    \upcast{\sigma'}{\sigma''}\upcast{\sigma}{\sigma'} M \equiv \upcast{\sigma}{\sigma''} M \and
    \dncast{\sigma}{\sigma'}\dncast{\sigma'}{\sigma''} M \equiv \dncast {\sigma} {\sigma''} M \and

    \upcast{A}{B}\upcast{\sigma}{\sigma'}M \equiv \upcast{\sigma}{\sigma'}\upcast{A}{B}M\and
    \dncast{A}{B}\dncast{\sigma}{\sigma'}M \equiv \dncast{\sigma}{\sigma'}\dncast{A}{B}M\and
  \end{mathpar}
  \caption{Proveable Uniqueness Theorems}
  \label{fig:uniq}
\end{figure}

We can show the following properties of the interaction between subtyping and casts axiomatically:

\begin{lemma}\label{lem:gradual_subty_admissible}
  The following hold:

  \begin{enumerate}
      \item $\sg^\ltdyn \vDash_{d_\sigma} \upcast {A'} {B'} M \ltdyn \upcast {A} {B} N : B' $.
      \item $\sg^\ltdyn \vDash_{d_\sigma} \dncast {A} {B} M \ltdyn \dncast {A'} {B'} N : A' $.
      \item $\sg^\ltdyn \vDash_{\sigma_2'} \upcast {\sigma_1'} {\sigma_2'} P \ltdyn \upcast {\sigma_1} {\sigma_2} Q : c$.
      \item $\sg^\ltdyn \vDash_{\sigma_1'} \dncast {\sigma_1} {\sigma_2} P \ltdyn \dncast {\sigma_1'} {\sigma_2'} Q : c$.
  \end{enumerate}
\end{lemma}
\begin{proof}
  \item We have
  
  \begin{mathpar}
      \inferrule* [right = (\textsc {ValUpL})]
        {\inferrule* [right = (\text{Subtyping})]
          {\inferrule* [right = (\textsc {ValUpR})]
            {\vdash_{d_\sigma} M \ltdyn N : A}
            {\vdash_{d\sigma} M \ltdyn \upcast {A} {B} N : A \ltdyn B}
          }
          {\vdash_{d_\sigma} M \ltdyn \upcast {A} {B} N : A' \ltdyn B'}
        }
        {\vdash_{d\sigma} \upcast {A'} {B'} M \ltdyn \upcast {A} {B} N : B'}  
  \end{mathpar}
 
  \item Dual to the above.
  
  \item We have
  
  \begin{mathpar}
      \inferrule* [right = (\textsc {EffUpL})]
        {\inferrule* [right = (\text{Subtyping})]
          {\inferrule* [right = (\textsc {EffUpR})]
            {\vdash_{\sigma_1} P \ltdyn Q : c}
            {\vdash_{d_\sigma} P \ltdyn \upcast {\sigma_1} {\sigma_2} Q : c}
          }
          {\vdash_{d_\sigma'} P \ltdyn \upcast {\sigma_1} {\sigma_2} Q : c}
        }
        {\vdash_{\sigma_2'} \upcast {\sigma_1'} {\sigma_2'} P \ltdyn \upcast {\sigma_1} {\sigma_2} Q : c}  
  \end{mathpar}

  \item Dual to the above.
\end{proof}

\section{Operational Semantics}

\begin{figure}
\begin{mathpar}

    \inferrule
    {\effname \in \dom(\phi) \and E' \apart \effname}
    {E[\hndl 
      {E'[\raiseOpwithM \effname V]} x N {\phi}] \\
    \stepsin {}
    E[\phi(\effname)[V/x][(\lambda y.
      \hndl {(E'[y])} x N {\phi}
    )/k]]
    }

    \inferrule
    { }
    { E[\hndl {V} {x} {N} {\phi}] 
    \stepsin {} E[N[V/x]]
    }
    \quad \textsc {HandleVal}

    \inferrule
    { }
    { E[(\lambda x. M) V] \stepsin {} E[M[V/x]] }
    \quad \textsc{Lam}

    \inferrule
    { }
    { E[\letXbeboundtoYinZ{V}{x}{M}] \stepsin {} E[M[V/x]] }
    \quad \textsc{Let}

    \inferrule
    { }
    { E[\err] \stepsin {} \err}
    \quad \textsc{Err}

    \inferrule
    { }
    { E[\ifXthenYelseZ{\tru}{N_t}{N_f}] \stepsin {} E[N_t] }
    \quad \textsc{IfTrue}

    \inferrule
    { }
    { E[\ifXthenYelseZ{\fls}{N_t}{N_f}] \stepsin {} E[N_f] }
    \quad \textsc{IfFalse}

    \inferrule
    { }
    { E[ \upcast \sigma {\sigma'} {V} ]
    \stepsin {}
    E[ V ]
    }
    \quad \textsc {EffUpDnCastVal}

    \inferrule
    { }
    { E[ \dncast \sigma {\sigma'} {V} ]
    \stepsin {}
    E[ V ]
    }
    \quad \textsc {EffDnCastVal}


    \inferrule
    { \effname \in \sigma' \and E' \apart \effname}
    { E[ \upcast \sigma {\sigma'} E'[
        \raiseOpwithM \effname {V}] ]
    \stepsin {} \\
    E[ \letXbeboundtoYinZ
      {\dncast {B} {B'} \raiseOpwithM{\effname}{\upcast {A} {A'} V}}
      {x}
      {\upcast {\sigma} {\sigma'} E'[x]} ]
    }
    \quad \textsc {EffUpCast}

    \inferrule
    { \effname \in \sigma \and E' \apart \effname}
    { E[ \dncast \sigma {\sigma'} E'[
        \raiseOpwithM \effname {V}] ]
    \stepsin {} \\
    E[ \letXbeboundtoYinZ
      {\upcast {B} {B'} \raiseOpwithM{\effname}{\dncast {A} {A'} V}}
      {x}
      {\dncast {\sigma} {\sigma'} E'[x]} ]
    }
    \quad \textsc{GoodEffDnCast}

    \inferrule
    { \effname \notin \sigma \and E' \apart \effname}
    { E[ \dncast \sigma {\dyn} E'[ 
        \raiseOpwithM \effname {V}] ]
    \stepsin {}
    E[ \err ]
    }
    \quad \textsc{BadEffDnCast}

    \inferrule
    {}
    { E[ \updownarrow \boolty M ]
    \stepsin {}
    E [M]
    }
    \quad \textsc{BoolUpDnCast}




    \inferrule
    { }
    { E[ (\upcast {(A \to_\sigma B)} {(A' \to_{\sigma'} B')} V_f)\, V ] 
    \stepsin {}
    E [\upcast B {B'} \upcast \sigma {\sigma'} (V_f \dncast {A} {A'} V) ]
    }
    \quad \textsc{FunUpCast}

    \inferrule
    { }
    { E[ (\dncast {(A \to_\sigma B)} {(A' \to_{\sigma'} B')} V_f)\, V ] 
    \stepsin {}
    E [\dncast B {B'} \dncast \sigma {\sigma'} (V_f \upcast {A} {A'} V) ]
    }
    \quad \textsc{FunDnCast}

  \end{mathpar}  
  \caption{Full Operational Semantics}
\end{figure}

\begin{figure}
  \begin{mathpar}
    \Delta ::= \holeRhoT{\sigma}{A}

    \inferrule
    {\hastyDRhoMT{\Delta}{\sigma}{E}{A \to_\sigma B} \and \hastyDRhoMT{\cdot}{\sigma}{N}{A}}
    {\hastyDRhoMT{\Delta}{\sigma}{E\, N}{B}}

    \inferrule
    {\sg \vdash_{\sigma} {V} : {A \to_\sigma B} \and 
     \hastyDRhoMT{\holeRhoT {\sigma_i} C}{\sigma}{E}{A}}
    {\hastyDRhoMT{\holeRhoT {\sigma_i} C}{\sigma}{V\, E}{B}}

    \inferrule
    {\hasty E \boolty \and
      \sg \vdash_{\sigma}{N_t}{B} \and
      \sg \vdash_{\sigma}{N_f}{B}}
    {\hasty {\ifXthenYelseZ E {N_t}{N_f}} B}

    \inferrule
    {\hastyDRhoMT{\Delta}{\sigma}{E}{A} \and
      \effname @ A \leadsto B \in \sigma
    }
    {\hastyDRhoMT {\Delta}{\sigma}{\raiseOpwithM \effname E} {B}
    }

    \inferrule
    {\hasty E A\\\\
     \TmhastySGRhoMT \sig {\Gamma, x : A} \tau   N B\\\\
     (\forall (\effarr \effname {A_\effname}{B_{\effname}}) \in \sigma.~
     (\effname \not\in \dom(\phi) \wedge(\effarr \effname {A_\effname}{B_{\effname}}) \in \tau)\\\\
     \quad\vee(\TmhastySGRhoMT \sig {\Gamma, x:A_\effname, k:B_\effname \effto\tau B} \tau {\phi(\effname)} B))
    }
    {\TmhastyRhoMT {\tau}{\hndl E x N \phi} B}

    \inferrule
    {\hastyRhoMT \sigma E A \and
      \hastySGRhoMT{\sig}{\Gamma, x : A}{\sigma}{N}{B}}
    {\hasty {\letXbeboundtoYinZ E x N} {B}}

    \inferrule
    {\hasty{E}{A'} \and \sg \vdash A' \subty A}
    {\hastyRhoMT{\sigma}{E}{A}}

    \inferrule
    {\hastyRhoMT{\Delta}{\sigma'}{E}{A} \and \sg \vdash \sigma' \subty \sigma}
    {\hastyDRhoMT{\Delta}{\sigma}{E}{A}}

    \inferrule
    {\hasty E A \and \sig \vdash A \ltdyn B}
    {\hasty {\upcast A B E} {B}}

    \inferrule
    {\hasty E B \and \sig \vdash A \ltdyn B}
    {\hasty {\dncast A B E} {A}}

    \inferrule
    {\hastyRhoMT{\sigma} E A \and \sig \vdash \sigma \ltdyn \sigma'}
    {\hastyRhoMT{\sigma'} {\upcast \sigma {\sigma'} E} A}

    \inferrule
    {\hastyRhoMT{\sigma'} E A \and \sig \vdash \sigma \ltdyn \sigma'}
    {\hastyRhoMT{\sigma} {\dncast \sigma {\sigma'} E} A}

  \end{mathpar}
  \caption{Typing Rules for Evaluation Contexts}
  \label{fig:ev-ctx-typing}
\end{figure}

An evaluation context $E_{\apart \effname}$ is one in which none of
the handler clauses in the spine of the context handles $\effname$.
\begin{figure}
  \begin{mathpar}

    \inferrule
    {}
    {\effname \apart \bullet}

    \inferrule
    {\effname \apart E}
    {\effname \apart (\upcast A B E) }

    \inferrule
    {\effname \apart E}
    {\effname \apart (\dncast A B E) }

    \inferrule
    {\effname \apart E \and \effname \notin \sigma \and \effname \notin \sigma'}
    {\effname \apart (\upcast \sigma {\sigma'} E) }

    \inferrule
    {\effname \apart E \and \effname \notin \sigma'}
    {\effname \apart (\dncast \sigma {\sigma'} E)}

    \inferrule
    {\effname \apart E \and \effname' \text { any effect}}
    {\effname \apart (\raiseOpwithM{\effname'}{E}) }

    \inferrule
    {\effname \apart E \wedge \effname\not\in\dom(\phi)}
    {\effname \apart (\hndl E x N \phi)}

    \inferrule
    {\effname \apart E}
    {\effname \apart (E\, M)}

    \inferrule
    {\effname \apart E}
    {\effname \apart (V\, E)}

    \inferrule
    {\effname \apart E}
    {\effname \apart (\ifXthenYelseZ E {N_t} {N_f})}

    \inferrule
    {\effname \apart E}
    {\effname \apart (\letXbeboundtoYinZ E x N)}

  \end{mathpar}
  \caption{Apartness of Effect from an Evaluation Context}
\end{figure}

\subsection{Operational Semantics from First Principles}

Now we show that every operational reduction is justified by our
inequational theory.

\begin{lemma}[Effect Casts are Handlers]
  \label{lem:eff-casts-are-handlers}
  Let $\sigma \ltdyn \tau$ where $\sigma$ is a concrete effect set.

  Then the upcast $\upcast \sigma \tau$ is equivalent to a handler in that for any $M : \compty \sigma A$:
  \[ \upcast \sigma \tau M \equiv \hndl M x x {\phi_{\upcast \sigma \tau}}\]
  where for each $\effname \in \dom(\sigma)$
  \[ x,k \vdash \phi_{\upcast{\sigma}\tau}(\effname) = k(\dncast{B_\sigma}{B_\tau}\raiseOpwithM \effname {\upcast {A_\sigma}{A_\tau}}) \]
  where $\effname @ A_\sigma \leadsto B_\sigma \in \sigma$ and
  $\effname @ A_\tau \leadsto B_\tau \in \tau$.

  Similarly, the downcast $\dncast \sigma \tau$ is equivalent to a
  handler in that for any $N : \compty \tau A$:
  \[
  \dncast \sigma \tau M \equiv
  \hndl M x x {\phi_{\dncast \sigma \tau}}
  \]
  where for each $\effname \in \dom(\tau)$, if $\effname
  \in\dom(\sigma)$, then
  \[ x,k \vdash \phi_{\dncast\sigma\tau}(\effname)
  = k(\upcast{B_\sigma}{B_\tau}\raiseOpwithM \effname {\dncast {A_\sigma}{A_\tau}})
  \]
  and if $\effname \not\in\dom(\sigma)$, then
  \[ \phi_{\dncast\sigma\tau}(\effname) = \err \]
\end{lemma}
\begin{proof}
  First for the upcast case
  \begin{itemize}
  \item We want to show
    \[ \upcast \sigma \tau M \ltdyn \hndl M x x {\phi_{\upcast \sigma \tau}} \]
    By UpL, it is sufficient to show
    \[ M \ltdyn \hndl M x x {\phi_{\upcast \sigma \tau}} \]
    But by the handler $\eta$ rule, this is equivalent to showing
    \[ \hndl M x x {\phi_{\sigma}} \ltdyn \hndl M x x {\phi_{\upcast \sigma \tau}} \]
    where $\dom(\phi_\sigma) = \dom(\sigma)$ and $\phi_{\sigma}(\effname) =
    k(\raiseOpwithM \effname x)$.
    Then by congruence, we need to show that for each $\effname \in \dom(\sigma)$,
    \[ k(\raiseOpwithM \effname x) \ltdyn k(\dncast {B_\sigma}{B_\tau}\raiseOpwithM \effname \upcast {A_\sigma}{A_\tau} x)\]
    which follows from UpR/DnR and congruence rules
  \item We want to show
    \[ \hndl M x x {\phi_{\upcast \sigma \tau}} \ltdyn \upcast \sigma \tau M \]
    By handler $\eta$ it is sufficient to show 
    \[ \hndl M x x {\phi_{\upcast \sigma \tau}} \ltdyn \hndl {\upcast \sigma \tau M} x x {\phi_{\tau}} \]
    where $\dom(\phi_{\tau}) = \dom(\tau)$ and $\phi_\tau(\effname) =
    k(\raiseOpwithM \effname x)$.
    Then $M \ltdyn {\upcast \sigma \tau M}$ by UpR and so
    by congruence we need only to show for each $\effname \in \sigma$
    that
    \[ \phi_{\upcast\sigma\tau}(\effname)\ltdyn \phi_\tau(\effname) \]
    which follows by a similar argument to the previous case.
  \end{itemize}

  Next, the downcast cases.
  \begin{itemize}
  \item We want to show
    \[ \hndl N x x {\phi_{\dncast\sigma\tau }} \ltdyn \dncast \sigma \tau N \]
    By DnR, it is sufficient to show
    \[ \hndl N x x {\phi_{\dncast\sigma\tau }} \ltdyn N \]
    By handler $\eta$ this is equivalent to showign
    \[ \hndl N x x {\phi_{\dncast\sigma\tau }} \ltdyn \hndl N x x {\phi_{\tau}} \]
    That is, for any $\effname \in \dom(\tau)$ that
    \[ \phi_{\dncast\sigma\tau}(\effname) \ltdyn \phi_\tau(\effname) \]
    There are two cases
    \begin{enumerate}
    \item If $\effname \in \dom(\sigma)$, then we need to show
      \[ k(\upcast {B_\sigma}{B_\tau}\raiseOpwithM \effname {\upcast {A_\sigma}{A_\tau} x})
      \ltdyn k(\raiseOpwithM \effname x)\]
      which follows by congruence and DnL/UpL rules.
    \item If $\effname \notin \dom(\sigma)$, then we need to show
      \[ \err \ltdyn k(\raiseOpwithM \effname x)\]
      which is immediate.
    \end{enumerate}
  \item We want to show
    \[ \dncast \sigma \tau N \ltdyn \hndl N x x {\phi_{\dncast\sigma\tau}} \]
    By handler $\eta$ this is equivalent to showing
    \[ \hndl {(\dncast \sigma \tau N)} x x {\phi_{\sigma}} \ltdyn \hndl N x x {\phi_{\dncast\sigma\tau}} \]
    By congruence and DnL this reduces to showing for each $\effname
    \in \dom(\sigma)$ that
    \[ \phi_\sigma(\effname) \ltdyn \phi_{\dncast\sigma\tau}(\effname)\]
    since $\effname \in \dom(\sigma)$, these are each of the form:
    \[ k(\raiseOpwithM \effname x) \ltdyn k(\upcast {B_\sigma}{B_\tau}\raiseOpwithM \effname {\upcast {A_\sigma}{A_\tau} x}) \]
    which follows by congruence and DnR/UpR rules.
  \end{itemize}
\end{proof}

\begin{lemma}[Derivation of Function Casts]
  \label{lem:derivation-fun-casts}
  \[ \upcast {A \effto \sigma B}{A' \effto \tau B'} f
  \equiv
  \lambda x. \upcast B{B'}\upcast \sigma \tau (f (\dncast A {A'} x))
  \]
  And similarly,
  \[ \dncast {A \effto \sigma B}{A' \effto \tau B'} f
  \equiv
  \lambda x. \dncast B{B'}\dncast \sigma \tau (f (\upcast A {A'} x))
  \]
\end{lemma}
\begin{proof}
  We show the upcast cases, the downcast cases are precisely dual.
  \begin{enumerate}
  \item We want to show
    \[ \upcast {A \effto \sigma B}{A' \effto \tau B'} f
    \ltdyn
    \lambda x. \upcast B{B'}\upcast \sigma \tau (f (\dncast A {A'} x))
    \]

    By UpL, it is sufficient to show
    \[ f
    \ltdyn
    \lambda x. \upcast B{B'}\upcast \sigma \tau (f (\dncast A {A'} x))
    \]
    By $\eta$ equivalence for functions it is sufficient to show
    \[ \lambda x. f x
    \ltdyn
    \lambda x. \upcast B{B'}\upcast \sigma \tau (f (\dncast A {A'} x))
    \]
    Which follows by congruence rules and UpR/DnR rules.
  \item We want to show
    \[
    \lambda x. \upcast B{B'}\upcast \sigma \tau (f (\dncast A {A'} x))
    \ltdyn
     \upcast {A \effto \sigma B}{A' \effto \tau B'} f
    \]
    By function $\eta$ it is sufficient to show
    \[
    \lambda x. \upcast B{B'}\upcast \sigma \tau (f (\dncast A {A'} x))
    \ltdyn
    \lambda y. (\upcast {A \effto \sigma B}{A' \effto \tau B'} f) y
    \]
    Which follows by congruence and UpL/DnL/UpR rules.
  \end{enumerate}
\end{proof}

\begin{lemma}
  If $x,k \vdash \phi(\effname) = k(\raiseOpwithM \effname x)$, then
  \[
  \hndl {\raiseOpwithM \effname x} y N \phi
  \equiv \letXbeboundtoYinZ {(\raiseOpwithM \effname x)} o
  \hndl o y N \phi
  \]
\end{lemma}
\begin{proof}
  \begin{align*}
    \hndl {\raiseOpwithM \effname x} y N \phi
    &\equiv \hndl {(\letXbeboundtoYinZ {\raiseOpwithM \effname x} o o} y N \phi\\
    &\equiv (\lambda o. \hndl o y N \phi)(\raiseOpwithM \effname x)\\
    &\equiv \letXbeboundtoYinZ {(\raiseOpwithM \effname x)} o \hndl o y N \phi
  \end{align*}
\end{proof}

This lemma is useful for the cast cases of the following, as it
reduces to showing the cast is equivalent to one whose $\effname$ case
is just a re-raise.

\begin{lemma}
  \label{lem:apart-raise}
  If $E \apart \effname$, then
  \[ E[\raiseOpwithM \effname x] \equiv \letXbeboundtoYinZ {\raiseOpwithM \effname x} y {E[y]}\]
\end{lemma}
\begin{proof}
  By induction on $\effname \apart E$

  \begin{itemize}
  \item $\inferrule {} {\effname \apart \bullet}$
    \begin{align*}
      \raiseOpwithM \effname x &\equiv \letXbeboundtoYinZ {\raiseOpwithM \effname x} y y
    \end{align*}
  \item $\inferrule
    {\effname \apart E}
    {\effname \apart (\upcast A B E) }$

    \begin{align*}
      \upcast A B E[\raiseOpwithM \effname x]
      &\equiv \letXbeboundtoYinZ {E[\raiseOpwithM \effname x]} y \upcast A B y\\
      &\equiv \letXbeboundtoYinZ {(\letXbeboundtoYinZ {(\raiseOpwithM \effname x)} z E[z])} y \upcast A B y\\
      &\equiv
      \letXbeboundtoYinZ {(\raiseOpwithM \effname x)} z
      \letXbeboundtoYinZ {E[z]} y \upcast A B y\\
      &\equiv
      \letXbeboundtoYinZ {(\raiseOpwithM \effname x)} z
      \upcast A B {E[z]}
    \end{align*}
  \item $\inferrule
    {\effname \apart E}
    {\effname \apart (\dncast A B E) }$

    Similar to previous.



  \item $\inferrule
    {\effname \apart E}
    {\effname \apart (\raiseOpwithM{\effname'}{E})}$

    \begin{align*}
      \raiseOpwithM{\effname'}{E[\raiseOpwithM {\effname'} x]}
      &\equiv\raiseOpwithM{\effname'}{(\letXbeboundtoYinZ {\raiseOpwithM {\effname'} x} z E[z])}\\
      &\equiv
      \letXbeboundtoYinZ {\raiseOpwithM {\effname'} x} z E[z]
      \raiseOpwithM{\effname'}{()}\\
    \end{align*}

  \item $\inferrule
    {\effname \apart E \and
     \effname\not\in\dom(\phi)}
    {\effname \apart (\hndl {E} {y} {N} \phi)}$

    Define $\psi$ to be the extension of $\phi$ with the case
    $\psi(\effname) = k(\raiseOpwithM \effname x)$.

    \begin{align*}
      \hndl {E[\raiseOpwithM \effname x]} {y} {N} \phi
      &\equiv\hndl {E[\raiseOpwithM \effname x]} {y} {N} \psi\\
      &\equiv\hndl {(\letXbeboundtoYinZ {(\raiseOpwithM \effname x)} z E[z])} {y} {N} \psi\\
      &\equiv (\lambda o. \hndl {E[o]} y N \psi)(\raiseOpwithM \effname x)\\
      &\equiv (\letXbeboundtoYinZ {(\raiseOpwithM \effname x)} o \hndl {E[o]} y N \psi)\\
      &\equiv (\letXbeboundtoYinZ {(\raiseOpwithM \effname x)} o \hndl {E[o]} y N \phi)\\
    \end{align*}
    
  \item $\inferrule
    {\effname \apart E}
    {\effname \apart (E\, M)}$

    \begin{align*}
      (E[\raiseOpwithM \effname x]) M
      &\equiv
      \letXbeboundtoYinZ {E[\raiseOpwithM \effname x]} f
      \letXbeboundtoYinZ M y {f\,y}\\
      &\equiv
      \letXbeboundtoYinZ {(\letXbeboundtoYinZ {\raiseOpwithM \effname x} z {E[z]})} f
      \letXbeboundtoYinZ M y {f\,y}\\
      &\equiv
      \letXbeboundtoYinZ {\raiseOpwithM \effname x} z
      \letXbeboundtoYinZ {E[z]} f
      \letXbeboundtoYinZ M y {f\,y}\\
      &\equiv
      \letXbeboundtoYinZ {\raiseOpwithM \effname x} z
      {(E[z])\,M}
    \end{align*}

  \item $\inferrule
    {\effname \apart E}
    {\effname \apart (V\, E)}$

    \begin{align*}
      (V\,E[\raiseOpwithM \effname x])
      &\equiv
      \letXbeboundtoYinZ V f
      \letXbeboundtoYinZ {E[\raiseOpwithM \effname x]} y
      f\,y\\
      &\equiv
      \letXbeboundtoYinZ V f
      \letXbeboundtoYinZ {(\letXbeboundtoYinZ {\raiseOpwithM \effname x} z E[z])} y
      f\,y\\
      &\equiv
      \letXbeboundtoYinZ {(\letXbeboundtoYinZ {\raiseOpwithM \effname x} z E[z])} y
      V\,y\\
      &\equiv
      \letXbeboundtoYinZ {\raiseOpwithM \effname x} z
      \letXbeboundtoYinZ {(E[z])} y
      V\,y\\
      &\equiv
      \letXbeboundtoYinZ {\raiseOpwithM \effname x} z
      V\,(E[z])\\
    \end{align*}

  \item $\inferrule
    {\effname \apart E}
    {\effname \apart (\ifXthenYelseZ E {N_t} {N_f})}$

    \begin{align*}
      \ifXthenYelseZ {E[\raiseOpwithM \effname x]} {N_t} {N_f}
      &\equiv \letXbeboundtoYinZ {(E[\raiseOpwithM \effname x])} y
      \ifXthenYelseZ y {N_t} {N_f}\\
      &\equiv \letXbeboundtoYinZ {(\letXbeboundtoYinZ {\raiseOpwithM \effname x} z E[z])} y
      \ifXthenYelseZ y {N_t} {N_f}\\
      &\equiv
      \letXbeboundtoYinZ {\raiseOpwithM \effname x} z
      \letXbeboundtoYinZ {(E[z])} y
      \ifXthenYelseZ y {N_t} {N_f}\\
      &\equiv
      \letXbeboundtoYinZ {\raiseOpwithM \effname x} z
      \ifXthenYelseZ {E[z]} {N_t} {N_f}\\
    \end{align*}
  \item $\inferrule
    {\effname \apart E}
    {\effname \apart (\letXbeboundtoYinZ E x N)}$

    \begin{align*}
      \letXbeboundtoYinZ {E[\raiseOpwithM \effname x]} y N
      &\equiv \letXbeboundtoYinZ {\letXbeboundtoYinZ {(\raiseOpwithM \effname x)} z E[z]} y N\\
      &\equiv
      \letXbeboundtoYinZ {(\raiseOpwithM \effname x)} z
      \letXbeboundtoYinZ {E[z]} y N
    \end{align*}
  \end{itemize}
\end{proof}

\begin{theorem}[Soundness of Operational Semantics]
  If $M \stepstar M'$ then $M \equiv M'$ is derivable in the
  inequational theory.
\end{theorem}
\begin{proof}
  \begin{enumerate}
  \item The value handle, boolean/function $\beta$ reductions and error reduction are immediate by axioms.
  \item
    \[ \inferrule
  {E \apart \effname}
  {\hndl {E[\raiseOpwithM \effname V]} y N \phi \equiv
    \phi(\effname)[V/x,\lambda o. \hndl {E[o]} y N \phi/k]
  } \]

  \begin{align*}
    \hndl {E[\raiseOpwithM \effname V]} y N \phi
    &\equiv \hndl {(\letXbeboundtoYinZ {\raiseOpwithM \effname V} z E[z])} y N \phi\tag{Lemma\ref{lem:apart-raise}}\\
    &\equiv \phi(\effname)[V/x,\lambda o. \hndl {E[o]} y N \phi/k]
  \end{align*}

  \item \[   \upcast \sigma {\tau} {V} \equiv V \]
    By the following:
    \begin{align*}
      \upcast \sigma {\tau} {V}
      &\equiv \hndl V x x {\phi_{\upcast\sigma\tau}}\tag{Lemma~\ref{lem:eff-casts-are-handlers}}\\
      &\equiv V\tag{Handle $\beta$}
    \end{align*}
  \item \[ \dncast \sigma {\tau} {V} \equiv V \] is similar to the
    previous.
  \item \[  \inferrule
  { \effname @ A \leadsto B \in \sigma \and
    \effname @ A' \leadsto B' \in \tau \and
    E \apart \effname}
  { \upcast \sigma {\tau} E[\raiseOpwithM {\effname} {V}]
      \equiv
    \letXbeboundtoYinZ
      {\dncast {B} {B'} \raiseOpwithM{\effname}{\upcast {A} {A'} V}}
      {x}
      {\upcast {\sigma} {\tau} E[x]}
  }\]

  \begin{align*}
    \upcast \sigma {\tau} E[\raiseOpwithM {\effname} {V}]
    &\equiv
    \hndl {(E[\raiseOpwithM {\effname} {V}])} x x {\phi_{\upcast \sigma \tau}}\tag{Lemma\ref{lem:eff-casts-are-handlers}}\\
    &\equiv 
    \hndl {(\letXbeboundtoYinZ {\raiseOpwithM {\effname} {V}} z E[z])} x x {\phi_{\upcast \sigma \tau}}\tag{Lemma\ref{lem:apart-raise}}\\
    &\equiv
    \phi_{\upcast \sigma \tau}(\effname)[V/x,\lambda o. \hndl {E[o]} x x {\phi_{\upcast \sigma \tau}}]\\
    &= (\lambda o. \hndl {E[o]} x x {\phi_{\upcast \sigma \tau}})(\dncast {B} {B'}\raiseOpwithM \effname {\upcast A {A'} V})\\
    &\equiv (\lambda o. \upcast\sigma\tau E[o])(\dncast {B} {B'}\raiseOpwithM \effname {\upcast A {A'} V})\\
    &\equiv
    \letXbeboundtoYinZ {(\dncast {B} {B'}\raiseOpwithM \effname {\upcast A {A'} V})} o
    \upcast \sigma \tau {E[o]}\\
  \end{align*}
  \item \[  \inferrule
  { \effname @ A \leadsto B \in \sigma \and
    \effname @ A' \leadsto B' \in \tau \and
    E \apart \effname}
  { \dncast \sigma {\tau} E[\raiseOpwithM {\effname} {V}]
      \equiv
    \letXbeboundtoYinZ
      {\upcast {B} {B'} \raiseOpwithM{\effname}{\dncast {A} {A'} V}}
      {x}
      {\dncast {\sigma} {\tau} E[x]}
  } \]
  Similar to previous
\item \[    \inferrule
    { \effname \notin \sigma \and E \apart \effname}
    { \dncast \sigma {\dyn} E[ \raiseOpwithM {\effname} {V}]
      \equiv
      \err
    }
    \]
    \begin{align*}
      \dncast \sigma {\dyn} E[ \raiseOpwithM {\effname} {V}]
      &\equiv \hndl {(E[ \raiseOpwithM {\effname} {V}])} x x {\phi_{\dncast\sigma\dyn}}\tag{Lemma\ref{lem:eff-casts-are-handlers}}\\
      &\equiv \hndl {(\letXbeboundtoYinZ {(\raiseOpwithM {\effname} {V})} z E[z])} x x {\phi_{\dncast\sigma\dyn}}\tag{Lemma\ref{lem:apart-raise}}\\
      &\equiv \err
    \end{align*}
 \item \[ \upcast \boolty \boolty V \equiv V \] By the identity rule.
 \item \[ \dncast \boolty \boolty V \equiv V\] By the identity rule.
 \item \[ { (\upcast {(A \to_\sigma B)} {(A' \to_{\tau} B')} V_f)\, V
    \equiv
    \upcast B {B'} \upcast \sigma {\tau} (V_f \dncast {A} {A'} V)
    }\]
   By the following:
   \begin{align*}
     (\upcast {(A \to_\sigma B)} {(A' \to_{\tau} B')} V_f)\, V &\equiv
     ((\lambda x.  \upcast{B}{B'}\upcast{\sigma}{\tau}(V_f (\dncast{A}{A')}x)))\, V\tag{Lemma\ref{lem:derivation-fun-casts}}\\
     &\equiv\upcast{B}{B'}\upcast{\sigma}{\tau}(V_f (\dncast{A}{A')}V) \tag{$\beta\to$}
   \end{align*}
 \item Similar to previous.
  \end{enumerate}
\end{proof}

\begin{theorem}[Adequacy]
  If $\cdot\vdash_\emptyset M \equiv M' : \boolty$ is derivable in the equational theory
  than for any $R \in \{ \tru,\fls,\mho \}$
  \[ M \stepstar R \iff M' \stepstar R \]
\end{theorem}

\begin{corollary}[Consistency]
  $\tru \equiv \fls$ is not derivable.
\end{corollary}

\begin{theorem}[Graduality]
  If $\cdot \vdash_\emptyset M \ltdyn M' : \boolty$
  Then for any $R \in \{\tru,\fls\}$,
  \[ M \stepstar R \Rightarrow M' \stepstar R \]
  and for any $R' \in \{\tru,\fls,\err\}$,
  \[ M' \stepstar R' \implies M \stepstar R' \]
\end{theorem}

\section{Elaboration}

\begin{lemma}
  If $A \gsubty B$ then there exist types $A_h,D_h,D_l, B_l$ with 
  \begin{enumerate}
  \item $c_l : A \ltdyn D_l$ and $c_h : A_h \ltdyn D_h$ satisfying $c_l \subty c_h$
  \item $d_l : B_l \ltdyn D_l$ and $d_h : B \ltdyn D_h$ satisfying $d_l \subty d_h$
  \item $e_l : D_l \ltdyn D$ and $e_h : D_h \ltdyn D$ with $e_l \subty
    e_h$ where $D = \lceil A \rceil = \lceil B \rceil$.
  \end{enumerate}
\end{lemma}
\begin{proof}
  By induction on the proof of $A \gsubty A'$.
\end{proof}

Then the four different choices of cast are all
equivalent in the inequational theory:
\begin{lemma}
  Given $A,A_h,B,B_l,D_l,D_h,D, c_l,c_h,d_l,d_h,e_l,e_h$ as in the
  output of the previous lemma, for any $\Gamma\vdash M : \compty
  {\sigma} A$, the following four terms are equivalent at type $B$.
  \begin{enumerate}
  \item $\dncast{B}{D_h}\upcast{A_h}{D_h}M$
  \item $\dncast{B}{D_h}\upcast{A}{D_l}M$
  \item $\dncast{B_l}{D_l}\upcast{A}{D_l}M$
  \item $\dncast{B}{D}\upcast{A}{D}M$
  \end{enumerate}
\end{lemma}
\begin{proof}
  \begin{enumerate}
  \item To show (1) is equivalent to (2), it suffices to show
    \[ \upcast{A_h}{D_h}M \equiv \upcast{A}{D_l}M \]
    which is an instance of the subtyping/cast rule since $c_l \ltdyn c_h$.
  \item Similarly to show (2) is equivalent to (3) follows from
    $d_l \subty d_h$
  \item Lastly we show (4) is equivalent to (2).
    By cast functoriality,
    \[ \dncast{B}{D}\upcast{A}{D}M \equiv \dncast{B}{D_h}\dncast{D_h}{D}\upcast{D_l}{D}\upcast{A}{D_l}M\]
    And by retraction the middle cast $\dncast{D_h}{D}\upcast{D_l}{D}$
    is the identity.
  \end{enumerate}
\end{proof}

\subsection{Graduality}

\begin{figure}[h!]
  \begin{mathpar}
    \begin{array}{rcl}
      \valatomlhs{c} & := & 
      \{ (V^l, V^r) : \texttt{val}(V^l) \wedge \texttt{val}(V^r) \wedge \\&& \quad
        (\hastyGDRhoMT {\cdot} {\cdot} {\emptyset} {V^l} {c^l}) \wedge
        (\hastyGDRhoMT {\cdot} {\cdot} {\emptyset} {V^r} {c^r}) \} \\
      \\
      
      \termatomlhs{A^l} {A^r} {d_\sigma} & := &
      \{ (M^l, M^r) : \\&& \quad
        (\hastyGDRhoMT {\cdot} {\cdot} {d_\sigma^l} {M^l} {A^l}) \wedge
        (\hastyGDRhoMT {\cdot} {\cdot} {d_\sigma^r} {M^r} {A^r}) \} \\
      \\


      \ecatomlhs {c} {\RbngT{\sigma^l}{A^l}} {\RbngT{\sigma^r}{A^r}} & := &
      \{ (x^l.M^l, x^r.M^r) : \\&& \quad
        (\hastyGDRhoMT {x^l : c^l} {\cdot} {\sigma^l} {M^l} {A^l}) \wedge
        (\hastyGDRhoMT {x^r : c^r} {\cdot} {\sigma^r} {M^r} {A^r}) \}
 
    \end{array}
  \end{mathpar}
  \caption{Well typed atoms}
\end{figure}

Our main goal is to prove the soundness of the inequational theory
with respect to the logical relation. That is

\begin{theorem}[Graduality]
  If $\Gamma^\ltdyn \vdash_{d_\sigma} M \ltdyn N : c$ then
  $\Gamma^\ltdyn \vDash_{d_\sigma} M \ltdyn N : c$
\end{theorem}
\begin{proof}
  By induction on the term precision derivation.
  \begin{enumerate}
  \item (ValSubst) Lemma \ref{lem:val_subst}
  \item (MonadUnitL) Lemma \ref{lem:monad_unit_l}
  \item (MonadUnitR) Lemma \ref{lem:monad_unit_r}
  \item (MonadAssoc) Lemma \ref{lem:monad_assoc}
  \item (BoolBeta) Lemmas \ref{lem:bool_beta_true} and \ref{lem:bool_beta_false}
  \item (BoolEta) Lemma \ref{lem:bool_eta}
  \item (IfEval) Lemma \ref{lem:if_eval}
  \item (FunBeta) Lemma \ref{lem:fun_beta}
  \item (FunEta) Lemma \ref{lem:eta_expansion_functions}
  \item (AppEval) Lemma \ref{lem:app_eval}
  \item (HandleBetaRet) Lemma \ref{lem:handle_beta_ret}
  \item (HandleBetaRaise) Lemma \ref{lem:handle_beta_raise} 
  \item (HandleEmpty) Lemma \ref{lem:handle_empty}
  \item (HandleExt) Lemma \ref{lem:handle_ext} 
  \item (RaiseEval) Lemma \ref{lem:raise_eval}
  

  \item (Variable) Lemma \ref{lem:cong_var}
  \item (Let) Lemma \ref{lem:cong_let}
  \item (Boolean) Lemma \ref{lem:cong_boolean}
  \item (If) Lemma \ref{lem:cong_if}
  \item (Lambda) Lemma \ref{lem:cong_lam}
  \item (App) Lemma \ref{lem:cong_app}
  \item (Raise) Lemma \ref{lem:cong_raise}
  \item (HandleCong) Lemma \ref{lem:cong_handle}


  \item (Transitivity) Lemma \ref{lem:heterogeneous-transitivity}

  \item (ErrBot) Lemma \ref{lem:error_bot}
  \item (ErrStrict) Lemma \ref{lem:error_strict}
  \item (SubtyMon) Lemma \ref{lem:subty_mono} 
  \item (ValUpSub) Lemma \ref{lem:gradual_subty_non_admissible}
  \item (ValDnSub) Lemma \ref{lem:gradual_subty_non_admissible}
  \item (EffUpSub) Lemma \ref{lem:gradual_subty_non_admissible}
  \item (EffDnSub) Lemma \ref{lem:gradual_subty_non_admissible}
  \item (ValUpL) Follows from Lemma \ref{lem:ValUpL_general}.
  \item (ValUpR) Follows from Lemma \ref{lem:ValUpR_general}.
  \item (ValUpEval) Lemma \ref{lem:ValUpEval}
  \item (ValDnR) Follows from Lemma \ref{lem:ValDnR_general}.
  \item (ValDnL) Follows from Lemma \ref{lem:ValDnL_general}.
  \item (ValDnEval) Lemma \ref{lem:ValDnEval}
  \item (ValRetract) Lemma \ref{lem:cast-retraction}.
  \item (EffUpL) Follows from Lemma \ref{lem:EffUpL_general} 
  \item (EffUpR) Follows from Lemma \ref{lem:EffUpR_general}
  \item (EffDnR) Follows from Lemma \ref{lem:EffDnR_general}
  \item (EffDnL) Follows from Lemma \ref{lem:EffDnL_general}
  \item (EffRetract) Lemma \ref{lem:cast-retraction}.
  \end{enumerate}
\end{proof}

We begin with a few lemmas that will be useful in our proofs.

\subsubsection{Lemmas}

\begin{lemma}\label{lem:vals_in_R_implies_vals_in_Result}
  If $(V_1, V_2) \in R$, and $V_1$ and $V_2$ are values of type $A^l$ and $A^r$ respectively, then $(V_1, V_2) \in \simirrel {d_\sigma} j {(R, A^l, A^r)}$.
\end{lemma}
\begin{proof}
  We will establish the first disjunct in the definition of $\simirrel{\cdot}{}{}$. This follows by assumption.
\end{proof}

\begin{lemma}\label{lem:vals_in_Result_implies_vals_in_E}
  If $(V_1, V_2) \in \simirrel {d_\sigma} j (R, A^l, A^r)$, then $(V_1, V_2) \in \simierel {d_\sigma} j (R, A^l, A^r)$.
\end{lemma}
\begin{proof}
  Let $\sim\, \in \{<, >\}$, and suppose $(V_1, V_2) \in \simirrel {d_\sigma} j (R, A^l, A^r)$. Notice that regardless of whether $\sim$ is $<$ or $>$, we will be able to show the last clause in the definition of $\ltierel {d_\sigma} {j} (R, A^l, A^r)$ or $\gtierel {d_\sigma} {j} (R, A^l, A^r)$. In particular, we can take $k = j$, $V_1 = V_1$, and $V_2 = V_2$, noting that $V_1$ steps to itself in $0$ steps, as does $V_2$. Thus, it remains to show that $V_1$ and $V_2$ are related by $\ltirrel {d_\sigma} {j} (R, A^l, A^r)$ or $\gtirrel {d_\sigma} {j} (R, A^l, A^r)$. This is true by assumption.
\end{proof}

\begin{lemma}\label{lem:vals_in_V_implies_vals_in_E}
  If $(V_1, V_2) \in \simivrel {c} j$, then $(V_1, V_2) \in \simierel {d_\sigma} j {\simivrel {c} {}}$. 
\end{lemma}
\begin{proof}
  By Lemma \ref{lem:vals_in_Result_implies_vals_in_E} (with $R = \simivrel c {}$), it suffices to show that $(\sigma_1 V_1, \sigma_1 V_2) \in \simirrel {d_\sigma} j {\simivrel {c} {}}$. This is true by Lemma \ref{lem:vals_in_R_implies_vals_in_Result}, again with $R = \simivrel c {}$.
\end{proof}

\begin{lemma}[anti-reduction, one-sided]\label{lem:anti-reduction-one-sided}

  Suppose $M_1 \stepsin {i_1} M_1'$ and $M_2 \stepsin {i_2} M_2'$.

  If $(M_1', M_2') \in \gtierel {d_\sigma} {j - i_2} (R, A^l, A^r)$, then 
  $(M_1, M_2) \in \gtierel {d_\sigma} j (R, A^l, A^r)$.

  Similarly, if $(M_1', M_2') \in \ltierel {d_\sigma} {j - i_1} (R, A^l, A^r)$, then 
  $(M_1, M_2) \in \ltierel {d_\sigma} j (R, A^l, A^r)$.
\end{lemma}
\begin{proof}
  We prove the first statement; the second is analogous (and in fact easier).
  The assumption that $(M_1', M_2') \in \gtierel {d_\sigma} {j - i_2} (R, A^l, A^r)$ has four cases:

  \begin{enumerate}
    \item $M_2' \stepsin {j-i_2+1}$. In this case, $M_2 \stepsin {i_2} M_2' \stepsin {j-i_2+1}$, i.e, $M_2 \stepsin {j+1}$. Thus, we may assert the first disjunct in the definition of $\gtierel {d_\sigma} j (R, A^l, A^r)$.
     
    \item There exists $k \le j-i_2$ such that $M_1' \stepsin {j-i_2-k} \err$, and furthermore $M_1' \stepsin * \err$. In this case, we have that $M_2 \stepsin {i_2} M_2' \stepsin {j-i_2-k} \err$, so $M_2 \stepsin {j-k} \err$. Also, $M_1 \stepsin {i_1} M_1' \stepsin * \err$, so $M_1 \stepsin * \err$. Thus, we may assert the second disjunct.
    
    \item There exists $k \le j-i_2$ and $N_2$ such that $M_2' \stepsin {j-i_2-k} N_2$ and $M_1' \stepsin * \err$. In this case we have $M_2 \stepsin {i_2} M_2' \stepsin {j-i_2-k} N_2$, so $M_2 \stepsin {j-k} N_2$. Thus, we may assert the third disjunct.
    
    \item Similar to previous case.
  \end{enumerate}
\end{proof}

\begin{lemma}[anti-reduction]\label{lem:anti-reduction}
  Suppose $M_1 \stepsin {i_1} M_1'$ and $M_2 \stepsin {i_2} M_2'$, and that 
  $(M_1', M_2') \in \simierel {d_\sigma} {j - m} (R, A^l, A^r)$, where $m = \min\{i_1, i_2\}$.
  Then $(M_1, M_2) \in \simierel {d_\sigma} j (R, A^l, A^r)$.
\end{lemma}
\begin{proof}
  Follows from one-sided anti-reduction (Lemma \ref{lem:anti-reduction-one-sided}) and downward closure.
\end{proof}

\begin{lemma}[forward reduction, one-sided]\label{lem:forward-reduction-one-sided}
  Suppose $M_1 \stepsin {i_1} M_1'$ and $M_2 \stepsin {i_2} M_2'$.
  
  If $(M_1, M_2) \in \gtierel {d_\sigma} {j + i_2} (R, A^l, A^r)$, then
  $(M_1', M_2') \in \gtierel {d_\sigma} j (R, A^l, A^r)$.

  Similarly, if $(M_1, M_2) \in \ltierel {d_\sigma} {j + i_1} (R, A^l, A^r)$, then
  $(M_1', M_2') \in \ltierel {d_\sigma} j (R, A^l, A^r)$.
\end{lemma}
\begin{proof}
  Follows from determinism of evaluation and a case analysis on the assumption that $M_1$ and $M_2$ are related.
\end{proof}

\begin{lemma}[forward reduction]\label{lem:forward-reduction}
  Suppose $M_1 \stepsin {i_1} M_1'$ and $M_2 \stepsin {i_2} M_2'$, and that 
  $(M_1, M_2) \in \simierel {d_\sigma} {j + m} (R, A^l, A^r)$, where $m = \max\{i_1, i_2\}$.
  Then $(M_1', M_2') \in \simierel {d_\sigma} j (R, A^l, A^r)$.
\end{lemma}
\begin{proof}
  Follows from one-sided forward reduction (Lemma \ref{lem:forward-reduction-one-sided}) and downward closure.
\end{proof}

Frequently in our proofs we will encounter a situation where we know that two
evaluation contexts are related in the $\simikrel{\cdot}{}{}$ relation, that is,
substituting related values gives related outputs.
On the other hand, as a cast applied to a value is not necessarily itself a value,
we cannot reason directly about what happens when such semantic values are substituted into
related evaluation contexts. We therefore introduce the following lemma.

\begin{lemma}\label{lem:reduction-technical}
    Suppose $E_1$ and $E_2$ are evaluation contexts that take values to values.
    Let $V_1$ and $V_2$ be values (not necessarily related) such that

    \[ 
        (E_1[V_1], E_2[V_2]) \in \simierel {d_\sigma'} {j} {\simivrel {c} {}}.
    \]

    Furthermore, let $(E^l[x^l], E^r[x^r] \in \simikrel {c} {j} {\simierel {d_\sigma} {} {\simivrel {d} {}}})$.

    Then

    \[ (E^l[E_1[V_1]], E^r[E_2[V_2]]) \in \simierel {d_\sigma} {j} {\simivrel {d} {}}. \]

\end{lemma}
\begin{proof}
    We show the proof for $\sim\, =\, >$.

    By assumption, we have that there exist values $V_1'$ and $V_2'$ such that
    $E_1[V_1] \stepsin {i_1} V_1'$ and $E_2[V_2] \stepsin {i_2} V_2'$, for some $i_1$ and $i_2$.

    Thus, $E^l[E_1[V_1]] \stepsin {i_1} E^l[V_1']$ and likewise $E^r[E_2[V_2]] \stepsin {i_2} E^r[V_2']$.

    By one-sided anti-reduction (Lemma \ref{lem:anti-reduction-one-sided}), it suffices to show that

    \[ (E^l[V_1'], E^r[V_2']) \in \gtierel {d_\sigma} {j - i_2} {\gtivrel {d} {}}. \]

    By assumption on $E^l$ and $E^r$ being related, it suffices to show that
    $(V_1', V_2') \in \gtivrel {c} {j - i_2}$.

    Now by one-sided forward reduction (Lemma \ref{lem:forward-reduction-one-sided}), it suffices to show

    \[ (E_1[V_1], E_2[V_2]) \in \gtierel {d_\sigma} {(j - i_2) + i_2} {\simivrel {c} {}}. \]

    But this is precisely our assumption, so we are finished.

\end{proof}

\textbf{Remark:}
The reason why we needed to consider cases on $\sim$ separately is that the more ``generic''/two-sided
anti-reduction and forward-reduction lemmas involve the $\min$ or $\max$ of the number of steps
taken by the two terms. These may not be equal, in which case the arithmetic wouldn't work out.
But this doesn't mean the above lemma is false.
Conceptually, what is happening is that in the two-sided variants of the lemmas,
$\sim$ could be either $>$ or $<$.
On the other hand, the key here is that $\sim$ stays the same throughout the application of
anti-reduction and forward reduction, so we are able to use the more specific, one-sided lemmas.


\begin{lemma}[time-out]\label{lem:time-out}
  If $M_1 \stepsin {(i+1)}$, then $(M_1, M_2) \in \ltierel {d_\sigma} i R$.
  Similarly, if $M_2 \stepsin {(i+1)}$, then $(M_1, M_2) \in \gtierel {d_\sigma} i R$.
\end{lemma}
\begin{proof}
Suppose $M_1 \stepsin {(i+1)}$. Then we may assert the first disjunct in the definition of $\ltierel {d_\sigma} i R$ to conclude that $(M_1, M_2) \in \ltierel {d_\sigma} i R$. Likewise, if $M_2 \stepsin {(i+1)}$, then we may assert the first disjunct in the defintion of $\gtierel {d_\sigma} i R$ to conclude that $(M_1, M_2) \in \gtierel {d_\sigma} i R$.
\end{proof}

We present two trivial lemmas about the later modality. We do this to cut down on tedious reasoning about step indices within other proofs.

\begin{lemma}[]\label{lem:later1}
Let $R$ be a monotone step-indexed relation. If $(M_1, M_2) \in R_j$, then $(M_1, M_2) \in (\later R)_j$.
\end{lemma}
\begin{proof}
  Suppose $(M_1, M_2) \in R_j$. If $j = 0$, then $(M_1, M_2) \in (\later R)_0$ trivially.

  Otherwise, let $j = j' + 1$. By monotonicity of $R$, we have $(M_1, M_2) \in R_{j'}$, from which it follows that $(M_1, M_2) \in (\later R)_j$.
\end{proof}

\begin{lemma}[]\label{lem:later2}
  Let $R$ be a monotone step-indexed relation, and let $j$ be of the form $j = j' + 1$. If $(M_1, M_2) \in (\later R)_j$, then $(M_1, M_2) \in (\later R)_{j'}$.
\end{lemma}
\begin{proof}
  Suppose $(M_1, M_2) \in (\later R)_j$. Since $j = j' + 1$, by definition of $\later$ we must have that $(M_1, M_2) \in R_{j'}$. By the previous lemma (Lemma \ref{lem:later1}), we conclude $(M_1, M_2) \in (\later R)_{j'}$, which is what we needed to show.
\end{proof}

\begin{lemma}[Reasoning with ``later'' when both sides step]\label{lem:later_both_sides_step}
  Suppose $M \stepsin 1 M'$ and $N \stepsin 1 N'$, and that
  $(M', N') \in (\later \simierel {d_\sigma} {} {})_{k} {R}$.
  Then $(M, N) \in \simierel {d_\sigma} {k} {R}$.
\end{lemma}
\begin{proof}
  First suppose $k = 0$. Then by the time-out lemma (Lemma \ref{lem:time-out}),
  regardless of whether $\sim$ is $<$ or $>$, we have $(M, N) \in \simierel {d_\sigma} {0} R$.

  Now suppose $k \ge 1$. Then by the definition of later, we have that
  $(M', N') \in \simierel {d_\sigma} {k - 1} {R}$, so by anti-reduction
  we have that $(M, N) \in \simierel {d_\sigma} {k} {R}$.
\end{proof}

\begin{lemma}[L\"{o}b-induction]\label{lem:lob-induction}
  Let $P(n)$ be a predicate indexed by a natural number $n$.
  Suppose for all natural numbers $n$, we have that $(\later^m P)(n)$ implies $P(n)$ for all $m \ge 1$.
  Then $P(n)$ is true for all natural numbers $n$.
\end{lemma}
\begin{proof}
  The proof is by induction on $n$. When $n = 0$, the assumption says that $(\later P)(0)$ implies $P(0)$ (we have taken $m = 1$).
  So, it suffices to show that $(\later P)(0)$ holds. This is true by the definition of later.

  Now let $n \ge 1$ be fixed, and suppose $P(n)$ is true. We claim that $P(n+1)$ is true.
  By our assumption, it will suffice to show that $(\later P)(n+1)$ is true. (We have again chosen $m = 1$.)
  By definition of later, we must show $P(n)$ is true. But $P(n)$ is true by assumption.
\end{proof}

We now introduce a key lemma about evaluation contexts.

\textbf{Note}: In the below, we omit explicit mention of the types associated to the relations that parameterize $\simierel{\cdot}{}{}$ and $\simirrel{\cdot}{}{}$.

\begin{lemma}\label{lem:bind1}
  If
  \begin{enumerate}
    \item $(M_1, M_2) \in \simierel {d_\sigma'} j {S'}$
    \item For all $k \le j$ and $(N_1, N_2) \in \simirrel {d_\sigma'} k {S'}$, we have $(E_1[N_1], E_2[N_2]) \in \simierel {d_\sigma} k S$,
  \end{enumerate}

  \vspace{2ex}

  then $( E_1[M_1], E_2[M_2] ) \in \simierel {d_\sigma} j S$.

\end{lemma}
\begin{proof}
  We prove the lemma for $\sim\, =\, >$; the other case is similar. Based on assumption (1), there are four cases:

  \begin{enumerate}
    \item Case $M_2 \stepsin {j+1}$. We have $E_2[M_2] \stepsin {j+1}$, so we may assert the first disjunct in the definition of $\simierel {d_\sigma} j S$ to conclude that $( E_1[M_1], E_2[M_2] ) \in \simierel {d_\sigma} j S$.
    
    \item Case $\exists k \le j$ such that $M_2 \stepsin {j-k} \err$ and $M_1 \stepsin * \err$. We have $E_2[M_2] \stepsin {j-k+1} \err$. If $k = 0$, then we have $E_2[M_2] \stepsin {j+1}$, so we may assert the first disjunct. Otherwise, if $k \ge 1$, then we may take $k' = k-1$ and observe that $E_2[M_2] \stepsin {j-k'} \err$.
    
    \item Case $\exists k \le j$, $\exists V_2$ such that $M_2 \stepsin {j-k} N_2$ and $M_1 \stepsin * \err$. We have $E_2[M_2] \stepsin {j-k} E_2[N_2]$, so we may assert the third disjunct with $k = k$ and $N_2 = E_2[N_2]$.
    
    \item Case $\exists k \le j, \exists (N_1, N_2) \in \gtirrel {d_\sigma} k {S'}$ such that $M_2 \stepsin {j-k} N_2$ and $M_1 \stepsin * N_1$. We have $E_1[M_1] \stepsin {i_1} E_1[N_1]$ for some $i_1$, and $E_2[M_2] \stepsin {j-k} E_2[N_2]$. By assumption (2), we have $(E_1[N_1], E_2[N_2]) \in \simierel {d_\sigma} k S$.
    Thus, we may assert the fourth disjunct with $V_1 = E_1[N_1]$ and $V_2 = E_2[N_2]$.

  \end{enumerate}

\end{proof}

\begin{lemma}[``Semantic bind'']\label{lem:bind_general}
  Let $c : A \ltdyn A'$ and $d : B \ltdyn B'$.
  Let $E_1$ and $E_2$ be evaluation contexts such that 
  $\hastyDRhoMT{\holeRhoT{d_\sigma'^l}{A}}{d_\sigma^l}{E_1}{B}$ and
  $\hastyDRhoMT{\holeRhoT{d_\sigma'^r}{A'}}{d_\sigma^r}{E_2}{B'}$.
  Suppose
  \begin{enumerate}
    \item $(M_1, M_2) \in \simierel {d_\sigma'} j (S', A, A')$.
    \item For all $k \le j$ and $(V_1, V_2) \in S'_k$, we have $(E_1[V_1], E_2[V_2]) \in \simierel {d_\sigma} k (S, B, B')$.
    \item For all $k \le j$ and for all $\effname : c_\effname \leadsto d_\effname \in d_\sigma'$, if $E_1$ catches 
    $\effname$ or $E_2$ catches $\effname$, then for all $V^l, V^r \in (\later \simivrel {c_\effname} {})_k$
    and all evaluation contexts $E^l \apart \effname$ and $E^r \apart \effname$ such that 
    $(x^l.E^l[x^l], x^r.E^r[x^r]) \in
    (\later \simikrel{d_\effname}{})_k 
      (\simierel{d_\sigma'} {} {(S', A, A')},
      \RbngT{d_\sigma'^l}{A},
      \RbngT{d_\sigma'^r}{A'})$, we have \\
    $(E_1[E^l[\raiseOpwithM{\effname}{V^l}]],
      E_2[E^r[\raiseOpwithM{\effname}{V^r}]]) 
      \in \simierel {d_\sigma} k (S, B, B')$.
  \end{enumerate}

  Then $( E_1[M_1], E_2[M_2] ) 
    \in \simierel{d_\sigma} j (S, B, B')$.
\end{lemma}
\begin{proof}

  We use L\"{o}b induction (Lemma \ref{lem:lob-induction}). We assume that if the premises
  of the lemma are satisfied ``later'', then the conclusion holds later.
  We show under this assumption that the lemma holds ``now".

  We first apply Lemma \ref{lem:bind1}. The first hypothesis is immediate.
  Now let $k \le j$ and let $(N_1, N_2) \in \simirrel {d_\sigma'} k (S', A, A')$. We need to show that

  \[ (E_1[N_1], E_2[N_2]) \in \simierel {d_\sigma} k (S, B, B'). \]

  There are two cases to consider. In the first case, $N_1$ and $N_2$ are values and $(N_1, N_2) \in \simivrel c j$.
  Then by assumption (2) with $k = j$, we have $(E_1[N_1], E_2[N_2]) \in \simierel {d_\sigma} j (S, B, B')$, as needed.

  In the second case, there exist
  $\effname' : c' \leadsto d' \in d_\sigma'$, $E^l \apart \effname', E^r \apart \effname'$, and $V^l, V^r$ such that 
  $(V^l, V^r) \in (\later \simivrel{c'}{})_j$, and
  $(x^l.E^l[x^l], x^r.E^r[x^r]) \in\\ 
  (\later \simikrel{d'}{})_j 
    (\simierel{d_\sigma'} {} {(S', A, A')},
    \RbngT{d_\sigma'^l}{A},
    \RbngT{d_\sigma'^r}{A'})$,
  and
  $N_1 = E^l[\raiseOpwithM{\effname'}{V^l}]$ and 
  $N_2 = E^r[\raiseOpwithM{\effname'}{V^r}]$.

  Let 
  $N_1' = E_1[N_1] = E_1[E^l[\raiseOpwithM{\effname'}{V^l}]]$ and 
  $N_2' = E_2[N_2] = E_2[E^r[\raiseOpwithM{\effname'}{V^r}]]$.

  We need to show that

  \[
    ( N_1', N_2' ) 
      \in \simierel {d_\sigma} j {(S, B, B')}.
  \]

  We now consider whether one of $E_1$ or $E_2$ catches $\effname'$, or whether neither catches it.
  In the former case, assumption (3) immediately implies the desired result.

  Now suppose neither $E_1$ nor $E_2$ catches $\effname$. In this case, note that since
  $\effname' \apart E^l$ and $\effname' \apart E_1$, we have $\effname' \apart E_1[E^l]$.
  %
  %
  Likewise, we have $\effname' \apart E_2[E^r]$.
  It follows that $N_1'$ and $N_2'$ are stuck terms, i.e., they do not step. 
  Thus, it suffices to show that

  \[
    ( N_1', N_2' ) 
      \in \simirrel {d_\sigma} j {(S, B, B')}.
  \]

  We first claim $(V^l, V^r) \in (\later \simivrel{c'}{})_{j}$.
  Since $(V^l, V^r) \in (\later \simivrel{c'}{})_j$,
  this follows by Lemma \ref{lem:later2}.
  
  We now claim that 
  
  \[ 
    ( x^l.(E_1[E^l[x^l]]), x^r.(E_2[E^r[x^r]]) ) \in\\ 
      (\later \simikrel{d'}{})_{j} 
        (\simierel{d_\sigma} {} {(S, B, B')},
        \RbngT{d_\sigma^l}{B},
        \RbngT{d_\sigma^r}{B'}).
  \]
   
  To this end, let $k \le j$ and let $(V'^l, V'^r) \in (\later \simivrel {d'} {})_k$. We need to show that

  \[
    ( E_1[E^l[V'^l]], E_2[E^r[V'^r]] ) \in 
    (\later \simierel{d_\sigma} {} {})_{k} {(S, B, B')}.
  \]

  By the L\"{o}b induction hypothesis, it suffices to show that the three hypotheses of the lemma hold
  later.
  We claim that $(E^l[V'^l], E^r[V'^r]) \in (\simierel{d_\sigma'} {k} {\simivrel c {}})$.
  To see this, recall our assumption that

  \[
    ( x^l.E^l[x^l], x^r.E^r[x^r] ) \in
      (\later \simikrel{d'}{})_j 
        (\simierel{d_\sigma'} {} {\simivrel c {}}).
  \]

  Thus, we have that $( E^l[V'^l], E^r[V'^r] ) \in (\later \simierel{d_\sigma'} {} {})_{k} ({\simivrel c {}})$,
  which is what we needed to show.

\end{proof}

We now introduce a few lemmas about precision derivations. We first show how we may ``compose'' precision derivations:

\begin{lemma}[cut admissibility for precision derivations]\label{lem:cut_admissibility}
\begin{itemize}
  \item If $c : A \ltdyn B$ and 
           $d : B \ltdyn C$ then 
           $c \circ d : A \ltdyn C$.

  \item If $d_\sigma : \sigma \ltdyn \sigma'$ and
           $d_\sigma' : \sigma' \ltdyn \sigma''$ then 
           $d_\sigma \circ d_\sigma' : \sigma \ltdyn \sigma''$.
\end{itemize}
\end{lemma}
\begin{proof}
We prove these statements simultaneously by induction on $d$ and $d_\sigma'$.

\begin{itemize}
  \item Case $d = \boolty$. We have $B = C = \boolty$, so $c = \boolty$ (the reflexivity derivation). Thus, we may take $c \circ d = \boolty$.

  \item Case $d = d_i \to_{d_\sigma} d_o$. Inspecting the rules in figure \ref{fig:precision-term-assignment},
  we see that $B = B_i \to_{B_\sigma} B_o$ and $C = C_i \to_{C_\sigma} C_o$. Thus, we must have $A = A_i \to_{A_\sigma} A_o$,
  which means that $c = c_i \to_{c_\sigma} c_o$.
  

  We may take $c \circ d = (c_i \circ d_i) \to_{c_\sigma \circ d_\sigma} (c_o \circ d_o)$.
  By our inductive hypotheses, we have
  (1) $c_i \circ d_i : A_i \ltdyn C_i$, 
  (2) $c_\sigma \circ d_\sigma : A_\sigma \ltdyn C_\sigma$, and
  (3) $c_o \circ d_o : A_o \ltdyn C_o$.
  Now, using the type precision formation rule for functions, we get that $(c_i \circ d_i) \to_{c_\sigma \circ d_\sigma} (c_o \circ d_o) : (A_i \to_{A_\sigma} A_o \ltdyn C_i \to_{C_\sigma} C_o)$.
  
  \item Case $d_\sigma' = \dyn$. Define $\dyn \circ \dyn = \dyn$. Define $\texttt{Inj}(d) \circ \dyn = \texttt{Inj}(d)$.
  An concrete effect set cannot be composed with $\dyn$.
  \item Case $d_\sigma' = \texttt{Inj}(d)$. Note that $\sigma'' = \dyn$.
  We define $d_\sigma \circ \texttt{Inj}(d) = \texttt{Inj}(d_\sigma \circ d)$.
  \item Case $d_\sigma' = d_c'$: Define $(d_c \circ d_c')$ by
  $\effarr \effname c d \in (d_c \circ d_c')$ if and only if $c = c_1 \circ c_2$ and $d = d_1 \circ d_2$ with
  $\effarr \effname c_1 d_1 \in d_c$ and $\effarr \effname c_2 d_2 \in d_c'$.
\end{itemize}
\end{proof}

\begin{lemma}[reflexivity of composition]\label{lem:precision-reflexivity}
  Let $c : A \ltdyn B$ and $d_\sigma : \sigma \ltdyn \sigma'$. The following hold.
  \begin{itemize}
    \item $c \circ B = A \circ c = c$.
    \item $d_\sigma \circ \sigma' = \sigma \circ d_\sigma = d_\sigma$.
  \end{itemize}
\end{lemma}
\begin{proof}
  Follows from the uniquenes of precision derivations. That is, $c \circ B$, $A \circ c$, and $c$
  all are all proofs of $A \ltdyn B$, hence are equal.
\end{proof}


\begin{lemma}[decomposition]\label{lem:eff-precision-decomp}
  Suppose $\effname @ c \leadsto d \in d_\sigma \circ d_\sigma'$.
  Then there exist $c_1, c_2$ and $d_1, d_2$ such that
  $\effname @ c_1 \leadsto d_1 \in d_\sigma$ and
  $\effname @ c_2 \leadsto d_2 \in d_\sigma'$ and
  $c = c_1 \circ c_2$ and $d = d_1 \circ d_2$.
\end{lemma}
\begin{proof}
  By induction on $d_\sigma'$.
  \begin{itemize}
    \item Case $d_\sigma' = \dyn$.
    If $d_\sigma = \dyn$, then our assumption becomes $\effname @ c \leadsto d \in \dyn \circ \dyn = \dyn$.
    By definition of membership in $\dyn$, this means that $\effname @ c^r \leadsto d^r \in \sig$.

    We may take $c_1 = c$ and take $c_2$ to be the reflexivity derivation for $c^r \ltdyn c^r$.
    Likewise, we take $d_1 = d$ and $d_2$ to be the relfexivity derivation for $d^r \ltdyn d^r$.
    Note that $\effname @ c_2 \leadsto d_2 \in \dyn$, because $c_2^r = c^r$ and $d_2^r = d^r$,
    and we know $\effname @ c^r \leadsto d^r \in \sig$.
    We also have that
    $c = c_1 \circ c_2$ and $d = d_1 \circ d_2$, using Lemma \ref{lem:precision-reflexivity}.

    If $d_\sigma = \texttt{inj}(d_\sigma)$, then our assumption becomes
    $\effname @ c \leadsto d \in \texttt{inj}(d_\sigma)$. By definition of membership in $\inj{}{}$,
    we have that $\effname @ c \leadsto d \in d_\sigma$.
    We may again take $c_1 = c$ and $c_2$ to be the reflexivity derivation for $c^r \ltdyn c^r$,
    and likewise for $d_1$ and $d_2$. The same reasoning as above applies.

    \item Case $d_\sigma' = \texttt{inj}(d_\sigma)$.
    By definition of composition, our assumption becomes 
    $\effname @ c \leadsto d \in (d_\sigma \circ \texttt{inj}(d_\sigma)) = \texttt{inj}(d_\sigma \circ d_\sigma)$.

    By the induction hypothesis, there are $c_1, c_2$ and $d_1, d_2$ such that
    $\effname @ c_1 \leadsto d_1 \in d_\sigma$ and
    $\effname @ c_2 \leadsto d_2 \in d_\sigma$ and
    $c = c_1 \circ c_2$ and $d = d_1 \circ d_2$.
    By definition of membership in $\inj{}{}$, we have 
    $\effname @ c_2 \leadsto d_2 \in \texttt{inj}(d_\sigma) = d_\sigma'$.

    \item Case $d_\sigma' = d_c'$ (concrete effect set).
    Similar to previous case.
  \end{itemize}
\end{proof}

\subsubsection{Congruence Rules}

With these lemmas, we can prove the soundness of the term precision congruence rules.
The proofs are by induction on the term precision derivation.

\begin{lemma}[Congruence for Booleans]\label{lem:cong_boolean}\end{lemma}
\begin{proof}
        
        We need to show that $\Gamma^{\ltdyn} \vDash_{d_{\sigma}} \sem{\tru} \ltdyn \sem{\tru} \in \boolty$, and likewise for $\fls$ (we will show this for $\tru$ only; the reasoning for $\fls$ is exactly the same.)

        Let $\sim \, \in \{<, >\}$ and let $(\gamma_1, \gamma_2) \in \simigrel {\Gamma^\ltdyn} i$. We need to show 
        
        \[ 
          \left(\tru[\gamma_1], \tru[\gamma_2]\right) \in 
            \simierel {d_\sigma} i {\simivrel {\boolty} {}},
        \]

        i.e.,

        \[ \left(\tru, \tru\right) \in \simierel {d_\sigma} i {\simivrel \boolty {}}. \]

        By Lemma \ref{lem:vals_in_V_implies_vals_in_E}, it suffices to show that $(\tru, \tru) \in \simivrel {\boolty} {i}$. This is true according to the definition of the logical relation.

\end{proof}

  \begin{lemma}[Congruence for Variables]\label{lem:cong_var}\end{lemma}
  \begin{proof}

        We need to show that $\Gamma^\ltdyn, x_1 \ltdyn x_2 : c, \Gamma'^\ltdyn \vDash_{d_\sigma} x_1 \ltdyn x_2 \in c$.

        Let $\sim \, \in \{<, >\}$, and let $\widehat{\Gamma}^\ltdyn = \Gamma^\ltdyn, x_1 \ltdyn x_2 : c, \Gamma'^\ltdyn$. Let $(\gamma_1, \gamma_2) \in \simigrel {\widehat{\Gamma}^\ltdyn} i$. We need to show

        \[ \left(x_1[\gamma_1], x_2[\gamma_2]\right) \in \simierel {d_\sigma} i {\simivrel {c} {}}. \]

        By Lemma \ref{lem:vals_in_V_implies_vals_in_E}, it suffices to show that $(\gamma_1(x_1), \gamma_2(x_2)) \in \simivrel c i$. But this follows from the fact that $(\gamma_1, \gamma_2) \in \simigrel {\widehat{\Gamma}^\ltdyn} i$. In particular, by the definition of the logical relation, since $(x_1 \ltdyn x_2 : c) \in \widehat{\Gamma}^\ltdyn$, we have $(\gamma_1(x_1), \gamma_2(x_2)) \in \simivrel c i$.
  \end{proof}

  \begin{lemma}[Congruence for Lambdas]\label{lem:cong_lam}\end{lemma}
  \begin{proof}
        
        Suppose $\Gamma^{\ltdyn}, x \ltdyn y : c \vDash_{d_{\sigma'}} {M} \ltdyn {N} \in d$.
        We need to show that $\Gamma^{\ltdyn} \vDash_{d_{\sigma}} {\lambda x.M} \ltdyn {\lambda y.N} \in c \to_{d_{\sigma'}} d$.
        
        Let $\sim \, \in \{<, >\}$ and let $(\gamma_1, \gamma_2) \in \simigrel {\Gamma^\ltdyn} i$. We need to show 
        
        \[ 
          \left((\lambda x.M)[\gamma_1], (\lambda y.N)[\gamma_2]\right) \in
            \simierel {d_\sigma} i {\simivrel {c \to_{d_{\sigma'}} d} {}}. 
        \]

        Let $V_1 = \lambda x.{M}[\gamma_1]$ and $V_2 = \lambda y.{N}[\gamma_2]$.
        By Lemma \ref{lem:vals_in_V_implies_vals_in_E}, it will suffice to show that $(V_1, V_2) \in \simivrel {c \to_{d_{\sigma'}} d} i$. To this end, let $k \le i$ and let $(V_{i1}, V_{i2}) \in \simivrel c k$. We will show that $(V_1\, V_{i1}, V_2\, V_{i2}) \in \simierel {d_{\sigma'}} k {\simivrel d {}}$.

        Let $M' = ({M}[\gamma_1])(V_{i1}/x)$ and let $N' = ({N}[\gamma_2])(V_{i2}/y)$. Note that $(V_1\, V_{i1}) \stepsin 1 M'$, and similarly $(V_2\, V_{i2}) \stepsin 1 N'$. Thus, if $k = 0$, then by the Time-out Lemma (Lemma \ref{lem:time-out}), we conclude that $(V_1\, V_{i1}, V_2\, V_{i2}) \in \simierel {d_{\sigma'}} k {\simivrel d {}}$. 
        
        Hence, from now on, we assume $k \ge 1$. By the Anti-reduction lemma (Lemma \ref{lem:anti-reduction}) (with $i_1 = i_2 = 1$ and $j = k$), it will suffice to show that $(M', N') \in \simierel {d_{\sigma'}} {k-1} {\simivrel d {}}$.

        This will follow by our inductive hypothesis, which says that for any $\sim \,\in \{<, >\}$, any natrual number $n$, and any $(\gamma_1', \gamma_2') \in \simigrel {\Gamma^\ltdyn, x \ltdyn y : c} n$, we have
        
        \[ ({M}[\gamma_1'], {N}[\gamma_2']) \in \simierel {d_{\sigma'}} {n} {\simivrel {d} {}}. \]
        
        Let $\gamma_1' = \gamma_1, V_{i1}/x$, let $\gamma_2' = \gamma_2, V_{i2}/y$. It is easily verified that $(\gamma_1', \gamma_2') \in \simigrel {\Gamma^\ltdyn, x \ltdyn y : c} {k-1}$. (Doing so requires the monotonicity lemma, combined with the fact that $(\gamma_1, \gamma_2) \in \simigrel {\Gamma^\ltdyn} i$ and that $k-1 < k \le i$).
        %
        Taking $n = k-1$ above, and noting that $M' = {M}[\gamma_1']$ and $N' = {N}[\gamma_2']$, it follows that $(M', N') \in \simierel {d_{\sigma'}} {k-1} {\simivrel d {}}$, as we wanted to show.

  \end{proof}

  \begin{lemma}[Congruence for Function Application]\label{lem:cong_app}\end{lemma}
  \begin{proof}
       
        
        Suppose $\Gamma^{\ltdyn} \vDash_{d_{\sigma}} {M_1} \ltdyn {M_2} \in c \to_{d_{\sigma}} d$, and that $\Gamma^{\ltdyn} \vDash_{d_{\sigma}} {N_1} \ltdyn {N_2} \in c$.

        We need to show that $\Gamma^{\ltdyn} \vDash_{d_{\sigma}} {M_1\, N_1} \ltdyn {M_2\, N_2} \in d$.
        
        Let $\sim \, \in \{<, >\}$ and let $(\gamma_1, \gamma_2) \in \simigrel {\Gamma^\ltdyn} i$. We need to show 
        
        \[ 
          \left({M_1 \, N_1}[\gamma_1], {M_2 \, N_2}[\gamma_2]\right) \in 
            \simierel {d_\sigma} i {\simivrel {d} {}}.
        \]

        By Lemma \ref{lem:bind_general}, it will suffice to show that 
        
        (1) $(M_1[\gamma_1], M_2[\gamma_2]) \in 
          \simierel {d_\sigma} i {\simivrel {c \to_{d_\sigma} d} {}}$,
        and that (2) for all $k \le i$ and 
        $(V_1, V_2) \in 
          \simivrel {c \to_{d_\sigma} d} k$, 
        we have $(V_1\, N_1, V_1\, N_2) \in \simierel {d_\sigma} k {\simivrel {d} {}}$.

        (1) follows immediately from our first top-level assumption.

        To show (2), we again apply Lemma \ref{lem:bind_general}. It follows from our second top-level assumption that $(N_1[\gamma_1], N_2[\gamma_2]) \in \simierel {d_\sigma} k {\simivrel c {}}$. Now let $k' \le k$ and $(V_1', V_2') \in \simivrel {c} {k'}$. We claim that

        \[ 
          (V_1\, V_1', V_2\, V_2') \in
            \simierel {d_\sigma} {k'} {\simivrel {d} {}}.  
        \]

        This holds since $(V_1, V_2) \in \simivrel {c \to_{d_\sigma} d} k$ and
        $(V_1', V_2') \in \simivrel {c} {k'}$.

  \end{proof}

  \begin{lemma}[Congruence for If]\label{lem:cong_if}\end{lemma}
  \begin{proof}
        Suppose:
        
        \begin{enumerate}
          \item $\Gamma^{\ltdyn} \vDash_{d_{\sigma}} \sem{M} \ltdyn \sem{M'} \in \boolty$
          \item $\Gamma^{\ltdyn} \vDash_{d_{\sigma}} \sem{N_t} \ltdyn \sem{N_t'} \in c$
          \item $\Gamma^{\ltdyn} \vDash_{d_{\sigma}} \sem{N_f} \ltdyn \sem{N_f'} \in c$
        \end{enumerate}
        
        Let $\sim \, \in \{<, >\}$ and let $(\gamma_1, \gamma_2) \in \simigrel {\Gamma^\ltdyn} i$. We need to show 
        
        \[ 
          \left({\ifXthenYelseZ{M}{N_t}{N_f}}[\gamma_1], 
                {\ifXthenYelseZ{M'}{N_t'}{N_f'}}[\gamma_2]\right) 
            \in \simierel {d_\sigma} i {\simivrel {c} {}}.
        \]

        By Lemma \ref{lem:bind_general}, it will suffice to show that
        (1) $(\sem{M}[\gamma_1], \sem{M'}[\gamma_2]) \in \simierel {d_\sigma} i {\simivrel {\boolty} {}}$, and
        (2) for all $k \le i$ and $(V_1, V_2) \in \simivrel {\boolty} {k}$, we have

        \[ 
            (\ifXthenYelseZ {V_1} {{N_t}[\gamma_1]}{{N_f}[\gamma_1]}),  
            (\ifXthenYelseZ {V_2} {{N_t'}[\gamma_2]}{{N_f'}[\gamma_2]})
              \in \simierel {d_\sigma} k {\simivrel {c} {}}.
        \]

        We note that (1) follows by our first top-level assumption. For (2), the assumption $(V_1, V_2) \in \simivrel {\boolty} {k}$ has two cases. If $V_1 = V_2 = \tru$, then by anti-reduction (Lemma \ref{lem:anti-reduction}), it will suffice to show $({{N_t}[\gamma_1]}, {{N_t'}[\gamma_2]}) \in \simierel {d_\sigma} k {\simivrel {c} {}}$. But this follows from our second top-level assumption. Similarly, if $V_1 = V_2 = \fls$, then it suffices to show that $({{N_f}[\gamma_1]}, {{N_f'}[\gamma_2]}) \in \simierel {d_\sigma} k {\simivrel {c} {}}$, which follows from our third top-level assumption.
        
  \end{proof}

  \begin{lemma}[Congruence for Let]\label{lem:cong_let}\end{lemma}
  \begin{proof} 
    This proof is similar to the function abstraction proof and is hence omitted.
  \end{proof}

  \begin{lemma}[Congruence for Raise]\label{lem:cong_raise}\end{lemma}
  \begin{proof}
    Let $c : A_1 \ltdyn A_2$ and $d : B_1 \ltdyn B_2$. Suppose $\effname@c \leadsto d \in d_\sigma$ and

    \[
      \Gamma^\ltdyn \vDash_{d_\sigma} M_1 \ltdyn M_2 \in c.
    \]

    We need to show that

    \[
    \Gamma^\ltdyn \vDash_{d_\sigma}
    {\raiseOpwithM \effname {M_1}} \ltdyn
    {\raiseOpwithM \effname {M_2}} \in d.
    \]

    Let $\sim\, \in \{<, >\}$ and $(\gamma_1, \gamma_2) \in \simigrel {\Gamma} {j}$. We will show

    \[
      ({\raiseOpwithM \effname {M_1}}[\gamma_1],
        {\raiseOpwithM \effname {M_2}}[\gamma_2])
        \in \simierel {d_\sigma} {j} {\simivrel d {}}.
    \]

    We apply Lemma \ref{lem:bind_general}. We first claim that $({M_1}[\gamma_1], {M_2}[\gamma_2]) \in \simierel {d_\sigma} j {\simivrel c {}}$. This follows by assumption. Now, let $k \le j$ and $(V_1, V_2) \in \simivrel c k$. We claim that

    \[
      ({\raiseOpwithM \effname {V_1}}[\gamma_1],
        {\raiseOpwithM \effname {V_2}}[\gamma_2])
        \in \simierel {d_\sigma} {k} {\simivrel d {}}.
    \]

    By Lemma \ref{lem:vals_in_Result_implies_vals_in_E}, it suffices to show that

    \[
      ({\raiseOpwithM \effname {V_1}}[\gamma_1],
        {\raiseOpwithM \effname {V_2}}[\gamma_2])
        \in \simirrel {d_\sigma} {k} {\simivrel d {}}.
    \]

    We assert the second disjunct in the definition of $\simirrel{\cdot}{}{}$, where we take $\effname$ to be $\effname$ (which we know by assumption is in $d_\sigma$), and we take $E^l = E^r = \hole$ and $V^l = V_1$, $V^r = V_2$.

    We need to show that $(V_1, V_2) \in (\later \simivrel{c}{})_k$, and that

    \[
      ( x^l.(\hole [x^l]), x^r.(\hole [x^r]) ) \in 
        (\later \simikrel {d} {})_k
          (\simierel {d_\sigma} {} {\simivrel d {}})
    \]

    To this end, let $k' \le k$ and let $(V^l, V^r) \in \simivrel {c} {k'}$. We need to show

    \[
      (V^l, V^r) \in \simierel {d_\sigma} {k'} {\simivrel d {}}.
    \]

    But this follows by Lemma \ref{lem:vals_in_V_implies_vals_in_E}.

  \end{proof}
    
  \begin{lemma}[Congruence for Handle]\label{lem:cong_handle}
    \begin{mathpar}
      \inferrule*[]
      {M \ltdyn M' : \compty {d_\sigma} c\and
        y:c \vdash N \ltdyn N' : \compty {d_\tau} d \\\\
        \forall \effname @ d_i \leadsto d_o \in d_\sigma.
        (\effname \notin \dom(\phi) \wedge \effname \notin \dom(\phi') \wedge 
         \effname : d_i \leadsto d_o \in d_\tau) \vee \\\\ 
        x:d_i, k : d_o \effto {d_\tau} d \vdash \phi(\effname) \ltdyn \phi'(\effname) : \compty {d_\tau} d
      }
      {\hndl {M} y N \phi\ltdyn \hndl {M'} y {N'} {\phi'} : \compty {d_\tau} d}
    \end{mathpar}
  \end{lemma}
  \begin{proof}

    We use L\"{o}b induction (Lemma \ref{lem:lob-induction}). Assume that for
    all $k \le j$ and all $(\gamma_1, \gamma_2) \in (\later \simigrel {\Gamma^\ltdyn} {})_{k}$ and all 
    $(M, M') \in (\later \simierel {d_\sigma} {} {})_{k} (\simivrel c {})$, we have

    \begin{align*} (
      &{\hndl M x N \phi}[\gamma_1], \\
      &{\hndl {M'} {x'} {N'} {\phi'}}[\gamma_2] )
      \\ & \quad \quad \in (\later \simierel {d_{\tau}} j {}_{k} (\simivrel d {}).
    \end{align*}

    Let $(M, M') \in \simierel {d_\sigma} {j} {\simivrel c {}}$.

    Let $\sim\, \in \{<, >\}$ and let $(\gamma_1, \gamma_2) \in \simigrel {\Gamma^\ltdyn} {j}$.
    We need to show that
    
    \begin{align*} (
        &{\hndl M x N \phi}[\gamma_1], \\
        &{\hndl {M'} {x'} {N'} {\phi'}}[\gamma_2] )
      \\ & \quad \quad \in \simierel {d_{\tau}} j {\simivrel d {}}.
    \end{align*}

    By monadic bind (Lemma \ref{lem:bind_general}), it suffices to consider the following
    cases:

    \begin{itemize}
    
    \item Let $k \le j$ and let $(V_1, V_2) \in {\simivrel c {k}}$. We need to show that

    \begin{align*} (
      &{\hndl {V_1} x {N[\gamma_1]} {\phi[\gamma_1]}}, \\
      &{\hndl {V_2} {x'} {N'[\gamma_2]} {\phi'[\gamma_2]}} )
    \\ & \quad \quad \in \simierel {d_{\tau}} j {\simivrel d {}}.
    \end{align*}

    By anti-reduction (Lemma \ref{lem:anti-reduction}), it suffices to show that

    \[
      ( N[\gamma_1][V_1/x], N'[\gamma_2][V_2/x'] ) \in
      \simierel {d_{\tau}} k {\simivrel d {}}.
    \]

    This follows from the premise: if we let 
    $\gamma_1' = \gamma_1, V_1/x$ and $\gamma_2' = \gamma_2, V_2/x'$,
    then it is easily checked that
    $(\gamma_1', \gamma_2') \in \simigrel {\Gamma^\ltdyn, x \ltdyn x' : c} j$.
    Furthermore, $N[\gamma_1][V_1/x] = N[\gamma_1']$ and likewise for
    $N[\gamma_2][V_2/x']$. The premise then implies that
    $(N[\gamma_1][V_1/x], N'[\gamma_2][V_2/x']) \in \simierel {d_{\tau}} k {\simivrel d {}}$,
    as needed.

    \item Let $k \le j$ and let $\effname @ d_i \leadsto d_o \in d_\sigma$
    be an effect that is caught by either handler -- i.e., $\effname \in \dom(\phi)$ or
    $\effname \in \dom(\phi')$. By the premise, it follows that $\effname$ is in both
    $\dom(\phi)$ and $\dom(\phi')$.
    
    Let $(V^l, V^r) \in (\later \simivrel {c_i} {})_k$. Let $E^l \apart \effname$ and
    $E^r \apart \effname$ be evaluation contexts such that
    
    \[ (x^l.E^l[x^l], x^r.E^r[x^r]) \in
    (\later \simikrel{d_o}{})_k 
      (\simierel{d_{\sigma}} {} {\simivrel c {}}). \]
    
    We need to show that

    \begin{align*} (
      &{\hndl {E^l[\raiseOpwithM{\effname}{V^l}} x {N[\gamma_1]} {\phi[\gamma_1]}}, \\
      &{\hndl {E^r[\raiseOpwithM{\effname}{V^r}} {x'} {N'[\gamma_2]} {\phi'[\gamma_2]}} )
      \\ & \quad \quad \in \simierel {d_{\tau}} k {\simivrel d {}}.
    \end{align*}

    By anti-reduction, it suffices to show that

    \begin{align*}
      (
        &\phi(\effname)[\gamma_1]
          [V^l/x]
          [(\lambda y.\hndl {E^l[y]} x {N[\gamma_1]} {\phi[\gamma_1]})/k], \\
        &\phi'(\effname)[\gamma_2]
          [V^r/x']
          [(\lambda y.\hndl {E^r[y]} {x'} {N'[\gamma_2]} {\phi'[\gamma_2]})/k']
      ) \\ & \quad \quad \in (\later \simierel {d_{\tau}} {} {})_{k} (\simivrel d {}).
    \end{align*}

    To show this, we apply the premise, as follows.
    Let $H_1 = \hndl {E^l[y]} x {N[\gamma_1]} {\phi[\gamma_1]}$ and
        $H_2 = \hndl {E^r[y]} {x'} {N'[\gamma_2]} {\phi'[\gamma_2]}$. Let
    $\gamma_1' = \gamma_1, V^l/x_i,  (\lambda y. H_1)/k_i$ and let 
    $\gamma_2' = \gamma_2, V^r/x_i', (\lambda y. H_2)/k_i'$.
    In order to apply the premise, we must prove that 
    $(\gamma_1', \gamma_2') \in \simigrel {\Gamma^\ltdyn, x_i \ltdyn x_i' : d_i, k_i \ltdyn k_i' : d_o \to_{d_{\tau}} d} {k'}$.
    
    We first need to show that $(V^l, V^r) \in (\later \simivrel {c_i} {})_{k}$.
    This holds by assumption. We now need to show that

    \[ 
      ( (\lambda y. H_1), (\lambda y. H_2) ) \in 
        (\later \simivrel {d_o \to_{d_{\tau}} d} {})_{k}.
    \]

    To this end, let $k' \le k$ and let $(V_A, V_B) \in (\later \simivrel {d_o} {})_{k'}$. We need to show that

    \begin{align*}
      ( 
        &(\lambda y. H_1)\, V_A, (\lambda y. H_2)\, V_B ) 
        \\ &\quad \quad \in (\later \simierel {d_{\tau}} {} {})_{k'} (\simivrel d {})
    \end{align*}

    By anti-reduction, it suffices to show that

    \begin{align*}\label{eq:handle_goal_final}
      (
        &\hndl {E^l[V_A]} x N {\phi},
         \hndl {E^r[V_B]} {x'} {N'} {\phi'}
      ) \\ & \quad \quad \in (\later \simierel {d_{\tau}} {} {})_{k'} (\simivrel d {}).
    \end{align*}

    By the L\"{o}b induction hypothesis, it will suffice to show that

    \begin{align*}
      (
        &{E^l[V_A]}, {E^r[V_B]}
      ) \quad \quad \in (\later \simierel {d_{\sigma}} {} {})_{k'} (\simivrel c {}).
    \end{align*}

    Recall that by assumption, we have
    \[ (x^l.E^l[x^l], x^r.E^r[x^r]) \in
    (\later \simikrel{d_o}{})_k 
      (\simierel{d_{\sigma}} {} {\simivrel c {}}). \]

    Thus, it suffices to show that $(V_A, V_B) \in (\later \simivrel {d_o} {})_{k'}$,
    which is precisely our assumption.
    
    \end{itemize}

\end{proof}

Note that we do not need to show soundness of the term precision congruence rules involving casts.
This will follow from the soundness of the upper and lower bound rules for casts.

\begin{corollary}[reflexivity]\label{cor:reflexivity}
  Let $M$ be a term such that $\hastyRhoMT{\sigma}{M}{A}$. We have
  $\sg^\ltdyn \vDash_\sigma M \ltdyn M : A$.
\end{corollary}
\begin{proof} 
  By induction on $M$, using the soundness of the term precision relation already proven.
\end{proof}

\subsubsection{Equational Rules}

\begin{lemma}[Value substitution]\label{lem:val_subst}
  \begin{mathpar}
    \inferrule
    { x_1 \ltdyn x_2 : c \vDash_{d_\sigma} M \equiv N : d \and 
      V \equiv V' : c}
    {M[V/x_1] \equiv N[V'/x_2]}
  \end{mathpar}
\end{lemma}
\begin{proof}
  Suppose for all $j$ and all $(\gamma_1, \gamma_2) \in \simigrel {\Gamma^\ltdyn, x^l \ltdyn x^r : c} {j}$, that

  \[
    (x_1.M, x_2.N) \in \simierel {d_\sigma} {j} {\simivrel {d} {}}
  \]

  and

  \[
    (x_2.N, x_1.M) \in \simierel {d_\sigma} {j} {\simivrel {d} {}}.
  \]

  Further suppose that for all $j$,

  \[(V, V') \in \simivrel {c} {j} \]

  and

  \[(V', V) \in \simivrel {c} {j}. \]

  Let $j$ be arbitrary, and let $(\gamma_1, \gamma_2) \in \simigrel {\Gamma^\ltdyn} {}$.
  We need to show

  \[
    (M[V/x_1], N[V'/x_2]) \in \simierel {d_\sigma} {j} {\simivrel {d} {}}
  \]
  
  and

  \[
    (N[V'/x_2], M[V/x_1]) \in \simierel {d_\sigma} {j} {\simivrel {d} {}}.
  \]

  The second statement is symmetric to the first, so we show only the first.

  Let $\gamma_1' = (\gamma_1, x_1 = V)$ and let $\gamma_2' = (\gamma_2, x_2 = V')$.

  Note that we have $M[\gamma_1'] = M[\gamma_1][V/x_1]$ and $N[\gamma_2'] = N[\gamma_2][V'/x_2]$,
  by definition of substitution. 

  By our assumption, it is sufficient to show that $(\gamma_1', \gamma_2') \in
  \simigrel {\Gamma^\ltdyn, x_1 \ltdyn x_2 : c} {j}$.

  For this, it sufficies to show that
  $(\gamma_1'(x_1), \gamma_2'(x_2)) \in \simivrel {c} {j}$. But $\gamma_1'(x_1) = V$ and
  $\gamma_2'(x_2) V'$, so we are finished.


\end{proof}

\begin{lemma}[Monad Unit Left]\label{lem:monad_unit_l}
  \begin{mathpar}
    \inferrule*[]{}{\letXbeboundtoYinZ y x N \equiv N[y/x]}
  \end{mathpar}
\end{lemma}
\begin{proof}
  We show one direction of the equivalence; the other is symmetric.
  Let $j$ be arbitrary and let $(\gamma_1, \gamma_2) \in \simigrel {\Gamma} {j}$.
  We need to show

  \[ 
    (\letXbeboundtoYinZ{y}{x}{N}, N[y/x]) \in \simierel {\sigma} {j} {\simivrel {B} {}}.
  \]

  Since $y$ is a variable and hence a value, we have by the operational semantics that

  \[ \letXbeboundtoYinZ{y}{x}{N} \stepsin 1 N[y/x]. \]

  Thus, by anti-reduction, it suffices to show that

  \[
    (N[y/x], N[y/x]) \in \simierel {\sigma} {j} {\simivrel {B} {}}.
  \]

  But this follows by reflexivity (Corollary \ref{cor:reflexivity}).

\end{proof}

\begin{lemma}[Monad Unit Right]\label{lem:monad_unit_r}
  \begin{mathpar}
    \inferrule*[]{}{\letXbeboundtoYinZ M x x \equiv M}
  \end{mathpar}
\end{lemma}
\begin{proof}
  We show one direction of the equivalence; the other is symmetric.
  Let $j$ be arbitrary and let $(\gamma_1, \gamma_2) \in \simigrel {\Gamma} {j}$.
  We need to show

  \[ 
    (\letXbeboundtoYinZ M x x, M) \in \simierel {\sigma} {j} {\simivrel {B} {}}.
  \]

  Since $x$ is a variable and hence a value, we have by the operational semantics that

  \[\letXbeboundtoYinZ{M}{x}{x} \stepsin 1 M[x/x]. \]
  
  By definition of substitution, $M[x/x] = M$. Thus, by anti-reduction, it suffices
  to show that

  \[
    (M, M) \in \simierel {\sigma} {j} {\simivrel {B} {}}.
  \]

  This follows by reflexivity (Corollary \ref{cor:reflexivity}).

\end{proof}

\begin{lemma}[Monad Associativity]\label{lem:monad_assoc}
  \begin{mathpar}
    \inferrule*[]{}
    {\letXbeboundtoYinZ {(\letXbeboundtoYinZ M x N)} y P \equiv \letXbeboundtoYinZ M x \letXbeboundtoYinZ N y P}
  \end{mathpar}
\end{lemma}
\begin{proof}

  We show one direction of the equivalence; the other is symmetric.
  Let $j$ be arbitrary and let $(\gamma_1, \gamma_2) \in \simigrel {\Gamma} {j}$.
  We need to show

  \[ 
    (\letXbeboundtoYinZ {(\letXbeboundtoYinZ M x N)} y P, 
     \letXbeboundtoYinZ M x \letXbeboundtoYinZ N y P) \in 
       \simierel {\sigma} {j} {\simivrel {B} {}}.
  \]

  We apply Lemma \ref{lem:bind_general}, taking 
  $E_1 = \letXbeboundtoYinZ {(\letXbeboundtoYinZ{\hole}{x}{N})} {y} {P}$ and 
  $E_2 = \letXbeboundtoYinZ {\hole} {x} {\letXbeboundtoYinZ N y P}$.

  We first need to show that $(M, M) \in \simierel {\sigma} {j} {\simivrel {A} {}}$,
  which is true by reflexivity (Corolarry \ref{cor:reflexivity}).

  Now, let $k \le j$ and $(V_1, V_2) \in \simivrel {A} {k}$. We need to show that

  \[
    (
      \letXbeboundtoYinZ {(\letXbeboundtoYinZ{V_1}{x}{N})} {y} {P},
      \letXbeboundtoYinZ {V_2} {x} {\letXbeboundtoYinZ N y P}
    ) \in \simierel {\sigma} {k} {\simivrel {B} {}}.
  \]

  According to the operational semantics, we have

  \[ (\letXbeboundtoYinZ{V_1}{x}{N}) \stepsin 1 N[V_1/x]. \]

  Thus,

  \[
    \letXbeboundtoYinZ {(\letXbeboundtoYinZ{V_1}{x}{N})} {y} {P} \stepsin 1
    \letXbeboundtoYinZ {N[V_1/x]} {y} {P}.
  \]

  Similarly, we have

  \[ \letXbeboundtoYinZ {V_2} {x} {\letXbeboundtoYinZ N y P} 
      \stepsin 1 (\letXbeboundtoYinZ N y P)[V_2/x]
      = \letXbeboundtoYinZ {N[V_2/x]} {y} {P[V_2/x]}. 
  \]

  Note that since $x$ does not occur in $P$, we have $P[V_2/x] = P$. 

  Now, by anti-reduction, it suffices to show

  \[
    (
      \letXbeboundtoYinZ {N[V_1/x]} {y} {P},
      \letXbeboundtoYinZ {N[V_2/x]} {y} {P}
    ) \in \simierel {\sigma} {k} {\simivrel {B} {}}.
  \]

  We again apply Lemma \ref{lem:bind_general}, this time with
  $E_1 = \letXbeboundtoYinZ {\hole} {y} {P}$ and
  $E_2 = \letXbeboundtoYinZ {\hole} {y} {P}$.

  We first need to show that 
  $(N[V_1/x], N[V_2/x]) \in \simierel {\sigma} {k} {\simivrel {A} {}}$.
  This follows from reflexivity (Corollary \ref{cor:reflexivity}) and
  value substitution (Lemma \ref{lem:val_subst}) applied to our assumption
  on $V_1$ and $V_2$.

  Now let $k' \le k$ and $(V_1', V_2') \in \simivrel {k'} {A}$. We need to show that

  \[
    ( \letXbeboundtoYinZ {V_1'} {y} {P},
      \letXbeboundtoYinZ {V_2'} {y} {P}
    ) \in \simierel {\sigma} {k'} {\simivrel {B} {}}.
  \]

  By anti-reduction, it suffices to show

  \[ (P[V_1'/y], P[V_2'/y]) \in \simierel {\sigma} {k'} {\simivrel {B} {}}. \]

  This again follows from reflexiviy and value substitution.

\end{proof}

\begin{lemma}[$\eta$-expansion for Booleans]\label{lem:bool_eta}
  \begin{mathpar}
    \inferrule*[]{}{M[x : \boolty] \equiv \ifXthenYelseZ x {M[\tru/x]}{M[{\fls/x}]}}
  \end{mathpar}
\end{lemma}
\begin{proof}

  We show one direction of the equivalence; the other is symmetric.
  Let $j$ be arbitrary and let $(\gamma_1, \gamma_2) \in \simigrel {\Gamma, x_1 \ltdyn x_2 : \boolty} {j}$.
  We need to show

  \[ 
    ( M[\gamma_1], (\ifXthenYelseZ x {M[\tru/x]}{M[{\fls/x}]})[\gamma_2] ) \in 
       \simierel {\sigma} {j} {\simivrel {B} {}}.
  \]

  By definition of substitution, this is equivalent to

  \[ 
    ( M[\gamma_1],
      (\ifXthenYelseZ \gamma_2(x) {M[\tru/x][\gamma_2]}{M[{\fls/x}][\gamma_2]}) ) \in 
       \simierel {\sigma} {j} {\simivrel {B} {}}.
  \]
  
  By our assumption on $\gamma_1$ and $\gamma_2$, we have that either
  $\gamma_1(x_1) = \gamma_2(x_2) = \tru$ or $\gamma_1(x_1) = \gamma_2(x_2) = \fls$.

  We show only the former case; the latter is symmetric.
  In the former case, we need to show

  \[
    ( M[\tru/x][\gamma_1], 
      (\ifXthenYelseZ \tru {M[\tru/x][\gamma_2]}{M[{\fls/x}][\gamma_2]})
    ) \in 
    \simierel {\sigma} {j} {\simivrel {B} {}}.
  \]

  By anti-reduction, it is sufficient to show

  \[
    ( M[\tru/x][\gamma_1], M[\tru/x][\gamma_2]
    ) \in 
    \simierel {\sigma} {j} {\simivrel {B} {}}.
  \]

  This follows by reflexivity.

\end{proof}

\begin{lemma}[Boolean $\beta$ reduction - true]\label{lem:bool_beta_true}
  \begin{mathpar}
    \inferrule*[]{}{\ifXthenYelseZ \tru {N_t}{N_f} \equiv N_t}
  \end{mathpar}
\end{lemma}
\begin{proof}
  
  We show one direction of the equivalence; the other is symmetric.
  Let $j$ be arbitrary and let $(\gamma_1, \gamma_2) \in \simigrel {\Gamma} {j}$.
  We need to show

  \[ 
    ( (\ifXthenYelseZ {\tru} {N_t} {N_f})[\gamma_1], N_t[\gamma_2] ) \in
      \simierel {\sigma} {j} {\simivrel {B} {}}.
  \]

  By anti-reduction, it suffices to show

  \[
    ( N_t[\gamma_1], N_t[\gamma_2] ) \in \simierel {\sigma} {j} {\simivrel {B} {}}.
  \]

  This holds by reflexivity.

\end{proof}

\begin{lemma}[Boolean $\beta$ reduction - false]\label{lem:bool_beta_false}
  \begin{mathpar}
    \inferrule*[]{}{\ifXthenYelseZ \fls {N_t}{N_f} \equiv N_f}
  \end{mathpar}
\end{lemma}
\begin{proof}
  Precisely dual to the above proof.
\end{proof}

\begin{lemma}[Eval for If]\label{lem:if_eval}
  \begin{mathpar}
    \inferrule*[Right=IfEval]{}{\ifXthenYelseZ {M} {N_t}{N_f} \equiv \letXbeboundtoYinZ M x \ifXthenYelseZ x {N_t} {N_f}}
  \end{mathpar}
\end{lemma}
\begin{proof}

  We show one direction of the equivalence; the other is symmetric.
  Let $j$ be arbitrary and let $(\gamma_1, \gamma_2) \in \simigrel {\Gamma} {j}$.
  We need to show

  \[ 
    ( 
      (\ifXthenYelseZ {M} {N_t}{N_f})[\gamma_1], 
      (\letXbeboundtoYinZ M x \ifXthenYelseZ x {N_t} {N_f})[\gamma_2]
    ) \in \simierel {\sigma} {j} {\simivrel {B} {}}.
  \]

  We apply Lemma \ref{lem:bind_general}, with 
  $E_1 = \ifXthenYelseZ {\hole} {N_t[\gamma_1]} {N_f[\gamma_1]}$ and
  $E_2 = \letXbeboundtoYinZ {\hole} {x} {\ifXthenYelseZ {\gamma_2(x)} {N_t[\gamma_2]} {N_f[\gamma_2]}}$.

  We first need to show that $(M[\gamma_1], M[\gamma_2]) \in \simierel {\tau} {j} {\simivrel {\boolty} {}}$.
  This follows by reflexivity (Corollary \ref{cor:reflexivity}).

  Now let $k \le j$ and let $(V_1, V_2) \in \simivrel {\boolty} {k}$. We need to show that

  \[
    (
      (\ifXthenYelseZ{V_1}{N_t[\gamma_1]}{N_f[\gamma_1]}),
      (\letXbeboundtoYinZ {V_2} {x} {\ifXthenYelseZ {\gamma_2(x)} {N_t[\gamma_2]} {N_f[\gamma_2]}})
    ) \in \simierel {\sigma} {k} {\simivrel {B} {}}.
  \]

 By definition of $\simivrel {\boolty} {}$, either $V_1 = V_2 = \tru$ or $V_1 = V_2 = \fls$.
 We consider the first case; the second is symmetric.

 We need to show

 \[
  (
    (\ifXthenYelseZ{\tru}{N_t[\gamma_1]}{N_f[\gamma_1]}),
    (\letXbeboundtoYinZ {\tru} {x} {\ifXthenYelseZ {\gamma_2(x)} {N_t[\gamma_2]} {N_f[\gamma_2]}})
  ) \in \simierel {\sigma} {k} {\simivrel {B} {}}.
\]

By anti-reduction, it suffices to show

\[
  (
    N_t[\gamma_1],
    N_t[\gamma_2]
  ) \in \simierel {\sigma} {k} {\simivrel {B} {}}.
\]

This follows by reflexivity.

\end{proof}

\begin{lemma}[$\beta$-reduction for functions]\label{lem:fun_beta}
  \begin{mathpar}
    \inferrule*[Right=FunBeta]{}{(\lambda x. M) V \equiv M[V/x]}
  \end{mathpar}
\end{lemma}
\begin{proof}

  We show one direction of the equivalence; the other is symmetric.
  Let $j$ be arbitrary and let $(\gamma_1, \gamma_2) \in \simigrel {\Gamma} {j}$.
  We need to show

  \[ 
    ( 
      ((\lambda x. M)\, V)[\gamma_1], 
      (M[V/x])[\gamma_2]
    ) \in \simierel {\sigma} {j} {\simivrel {B} {}}.
  \]

  Since $V$ is a value, it suffices by anti-reduction to show that

  \[
    ( M[V/x][\gamma_1], M[V/x][\gamma_2] ) 
      \in \simierel {\sigma} {j} {\simivrel {B} {}}.
  \]

  This follows by reflexivity.

\end{proof}

\begin{lemma}[$\eta$-expansion for functions]\label{lem:eta_expansion_functions}
  Let $V_f$ be a value such that $\hastyRhoMT{\emptyset}{V}{A \to_{\sigma'} B}$. We have
  $\sg^\ltdyn \vDash_\sigma V_f \equiv (\lambda x. V_f x) : (A \to_{\sigma'} B)$.
\end{lemma}
\begin{proof}
  Let $j$ be arbitrary. We need to show

  \[ 
    ( V_f, (\lambda x. V_f x) ) \in 
      \simierel {\emptyset} {j} {\simivrel {A \to_{\sigma'} B} {}}.
  \]
  
  As these are values, it suffices by Lemma \ref{lem:vals_in_V_implies_vals_in_E} to show that they
  are related in $\simivrel {A \to_{\sigma'} B} {j}$.
  To this end, let $k \le j$ and let $(V_{i1}, V_{i2}) \in \simivrel {A} {k}$. We claim that

  \[ 
    ( V_f\, V_{i1}, (\lambda x. V_f x)\, V_{i2} ) \in 
      \simierel {\sigma'} {k} {\simivrel {B} {}}.
  \]

  By anti-reduction, it will suffice to show that

  \[ 
    ( V_f\, V_{i1}, V_f\, V_{i2} ) \in 
      \simierel {\sigma'} {k} {\simivrel {B} {}}.
  \]

  By reflexivity (Corollary \ref{cor:reflexivity}), we know that $(V_f, V_f) \in \simierel {\emptyset} {k} {A \to_{\sigma'} B}$,
  and since $V_f$ is a value, this means that $(V_f, V_f) \in \simivrel {A \to_{\sigma'} B} {k}$.
  This immediately implies the desired result, since $(V_{i1}, V_{i2}) \in \simivrel {A} {k}$.

\end{proof}

\begin{lemma}[AppEval]\label{lem:app_eval}
  \begin{mathpar}
    \inferrule*[]{}{M\,N \equiv \letXbeboundtoYinZ M x \letXbeboundtoYinZ N y x\,y}
  \end{mathpar}
\end{lemma}
\begin{proof}

  We show one direction of the equivalence; the other is symmetric.
  Let $j$ be arbitrary and let $(\gamma_1, \gamma_2) \in \simigrel {\Gamma} {j}$.
  We need to show

  \[ 
    ( 
      (M\, N)[\gamma_1], 
      (\letXbeboundtoYinZ M x \letXbeboundtoYinZ N y x\,y)[\gamma_2]
    ) \in \simierel {\tau_A} {j} {\simivrel {A_o} {}}.
  \]

  We apply Lemma \ref{lem:bind_general}, with $E_1 = (\hole \, N[\gamma_2])$ and
  $E_2 = \letXbeboundtoYinZ \hole x \letXbeboundtoYinZ {N[\gamma_2]} y x\,y$.

  We first need to show that
  $(M[\gamma_1], M[\gamma_2]) \in \simierel {\tau} {j} {\simivrel {A_i \to_{\tau_A} A_o} {}}$.
  This follows by reflexivity.

  Now let $k \le j$ and let $(V_1, V_2) \in \simivrel {A_i \to_{\sigma_A} A_o} {k}$.
  We need to show that

  \[ 
    ( 
      (V_1\, N[\gamma_1]), 
      (\letXbeboundtoYinZ {V_2} x \letXbeboundtoYinZ {N[\gamma_2]} y x\,y)
    ) \in \simierel {\tau_A} {k} {\simivrel {A_o} {}}.
  \]

  By anti-reduction, it suffices to show

  \[ 
    ( 
      (V_1\, N[\gamma_1]), 
      (\letXbeboundtoYinZ {N[\gamma_2]} y V_2\,y)
    ) \in \simierel {\tau_A} {k} {\simivrel {A_o} {}}.
  \]

  We again apply Lemma \ref{lem:bind_general}, this time with
  $E_1 = (V_1\, \hole)$ and $E_2 = \letXbeboundtoYinZ {\hole} {y} {V_2\, y}$.

  We need to show  $(N[\gamma_1], N[\gamma_2]) \in \simierel {\tau} {k} {\simivrel {A_i} {}}$,
  which holds by reflexivity. Now let $k' \le k$ and let $(V_1', V_2') \in \simivrel {A_i} {k'}$.
  We need to show that

  \[ 
    ( 
      (V_1\, V_1'), 
      (\letXbeboundtoYinZ {V_2'} y V_2\,y)
    ) \in \simierel {\tau_A} {k'} {\simivrel {A_o} {}}.
  \]

  By anti-reduction, it suffices to show

  \[ 
    ( 
      (V_1\, V_1'), 
      (V_2\, V_2')
    ) \in \simierel {\tau_A} {k'} {\simivrel {A_o} {}}.
  \]

  This follows from our assumptions on $V_1$ and $V_2$ and on $V_1'$ and $V_2'$.

\end{proof}

\begin{lemma}[HandleBetaRet]\label{lem:handle_beta_ret}
  \begin{mathpar}
    \inferrule*[]{}{\hndl x y M \phi \equiv M[x/y]}
  \end{mathpar}
\end{lemma}
\begin{proof}

  We show one direction of the equivalence; the other is symmetric.
  Let $j$ be arbitrary and let $(\gamma_1, \gamma_2) \in \simigrel {\Gamma} {j}$.
  We need to show

  \[ 
    ( 
      (\hndl x y M \phi)[\gamma_1], 
      (M[x/y])[\gamma_2]
    ) \in \simierel {\sigma} {j} {\simivrel {B} {}}.
  \]

  Since $x$ is a value, the above handle term steps, and by anti-reduction it is sufficient to show
  
  \[ 
    ( 
      (M[x/y][\gamma_1]), 
      (M[x/y])[\gamma_2]
    ) \in \simierel {\sigma} {j} {\simivrel {B} {}}.
  \]

  This follows by reflexivity.

\end{proof}

\begin{lemma}[HandleBetaRaise]\label{lem:handle_beta_raise}
  \begin{mathpar}
    \inferrule*[]
    {}
    {\hndl {({\letXbeboundtoYinZ {\raiseOpwithM \effname x} o {N_k}})} y M \phi \equiv 
      \phi(\effname)[\lambda o. \hndl {N_k} y M \phi/k]}
  \end{mathpar}
\end{lemma}
\begin{proof}


  We show one direction of the equivalence; the other is symmetric.
  Let $j$ be arbitrary and let $(\gamma_1, \gamma_2) \in \simigrel {\Gamma} {j}$.
  We need to show

  \begin{align*}
    ( 
      & (\hndl {({\letXbeboundtoYinZ {\raiseOpwithM \effname x} o {N_k}})} y M \phi)[\gamma_1], \\
      & (\phi(\effname)[\lambda o. \hndl {N_k} y M \phi /k])[\gamma_2]
    ) \\ & \quad\quad \in \simierel {\sigma} {j} {\simivrel {B} {}}.
  \end{align*}

  Let $E = \letXbeboundtoYinZ {\hole} {o} {N_k[\gamma_1]}$. Our goal is to show

  \begin{align*}
    ( 
      & (\hndl {E[\raiseOpwithM{\effname}{x}]} y {M[\gamma_1]} {\phi[\gamma_1]}), \\
      & (\phi(\effname)
          [\gamma_2][\lambda o . \hndl 
            {N_k[\gamma_2]} y {M[\gamma_2]} {\phi[\gamma_2]} /k])
    ) \\ & \quad\quad \in \simierel {\sigma} {j} {\simivrel {B} {}}.
  \end{align*}

  Note that $E \apart \effname$. By anti-reduction, it suffices to show

  \begin{align*}
    ( 
      & (\phi(\effname)[\gamma_1]
          [\lambda o' . \hndl
            {E[o']} {y} {M[\gamma_1]} {\phi[\gamma_1]} / k]), \\
      & (\phi(\effname)[\gamma_2]
          [\lambda o . \hndl 
            {N_k[\gamma_2]} y {M[\gamma_2]} {\phi[\gamma_2]} / k])
    ) \\ & \quad\quad \in \simierel {\sigma} {j} {\simivrel {B} {}}.
  \end{align*}

  That is, we need to show

  \begin{align*}
    ( 
      & (\phi(\effname)[\gamma_1]
          [\lambda o' . \hndl 
            {\letXbeboundtoYinZ {o'} {o} {N_k[\gamma_1]}} 
            {y} {M[\gamma_1]} {\phi[\gamma_1]} / k]), \\
      & (\phi(\effname)[\gamma_2]
          [\lambda o . \hndl {N_k[\gamma_2]} y {M[\gamma_2]} {\phi[\gamma_2]} / k])
    ) \\ & \quad\quad \in \simierel {\sigma} {j} {\simivrel {B} {}}.
  \end{align*}

  By ValSubst, it suffices to show (1)
  for all related $(V_{f1}, V_{f2}) \in \simivrel {A_i \to_{\sigma} B} {j}$ and
  $\gamma_1' = \gamma_1, V_{f1}/k$ and $\gamma_2' = \gamma_2, V_{f2}/k$, we have

  \[ 
    (
      \phi(\effname)[\gamma_1'],
      \phi(\effname)[\gamma_2']
    ) \in \simierel {\sigma} {j} {\simivrel {B} {}},
  \]

  and (2),

  \begin{align*}
    ( 
      & (\lambda o' . \hndl 
            {\letXbeboundtoYinZ {o'} {o} {N_k[\gamma_1]}} 
            {y} {M[\gamma_1]} {\phi[\gamma_1]}), \\
      & (\lambda o . \hndl {N_k[\gamma_2]} y {M[\gamma_2]} {\phi[\gamma_2]})
    ) \\ & \quad\quad \in \simierel {\sigma} {j} {\simivrel {A_i \to_{\sigma} B} {}}.
  \end{align*}

  (1) follows from relfexivity. To show (2), we will use transitivity 
  (Lemma \ref{lem:mixed-transitivity-terms}).
  If $\sim$ is $<$, then note that by MonadUnitL we have

  \[ (\letXbeboundtoYinZ{o'}{o}{N_k[\gamma_1]}, N_k[\gamma_2][o'/o]) \in 
    \simierel {\sigma} {j} {\simivrel {B} {}}, \]

  and by soundness of the congruence rules we have

  \begin{align*}
    ( 
      & (\lambda o' . \hndl 
            {\letXbeboundtoYinZ {o'} {o} {N_k[\gamma_1]}} 
            {y} {M[\gamma_1]} {\phi[\gamma_1]}), \\
      & (\lambda o' . \hndl 
            {N_k[\gamma_2][o'/o]} 
            {y} {M[\gamma_2]} {\phi[\gamma_2]})
    ) \\ & \quad\quad \in \simierel {\sigma} {j} {\simivrel {A_i \to_{\sigma} B} {}}.
  \end{align*}

  Then by transitivity, it will suffice to show that

  \begin{align*}
    ( 
      & (\lambda o' . \hndl 
            {N_k[\gamma_2][o'/o]} 
            {y} {M[\gamma_2]} {\phi[\gamma_2]}), \\
      & (\lambda o . \hndl 
            {N_k[\gamma_2]} y {M[\gamma_2]} {\phi[\gamma_2]})
    ) \\ & \quad\quad \in \simierel {\sigma} {\omega} {\simivrel {A_i \to_{\sigma} B} {}}.
  \end{align*}

  By congruence for lambdas, it suffices to show that, given related values
  $(V_1, V_2) \in \simivrel {A_i} {\omega}$, we have

  \begin{align*}
    ( 
      & (\hndl 
            {N_k[\gamma_2][o'/o][V_1/o']} {y} {M[\gamma_2]} {\phi[\gamma_2]}), \\
      & (\hndl {N_k[\gamma_2][V_2/o]} y {M[\gamma_2]} {\phi[\gamma_2]})
    ) \\ & \quad\quad \in \simierel {\sigma} {\omega} {\simivrel {A_i \to_{\sigma} B} {}}.
  \end{align*}

  This follows from the soundness of the congruence rules.

  On the other hand, if $\sim$ is $>$, then similarly by MonadUnitL we have

  \[ (\letXbeboundtoYinZ{o'}{o}{N_k[\gamma_1]}, N_k[\gamma_1][o'/o]) \in 
  \simierel {\sigma} {\omega} {\simivrel {B} {}}. \]

  It then suffices to show that

  \begin{align*}
    ( 
      & (\lambda o' . \hndl 
            {N_k[\gamma_1][o'/o]} 
            {y} {M[\gamma_1]} {\phi[\gamma_1]}), \\
      & (\lambda o . \hndl 
            {N_k[\gamma_2]} y {M[\gamma_2]} {\phi[\gamma_2]})
    ) \\ & \quad\quad \in \simierel {\sigma} {j} {\simivrel {A_i \to_{\sigma} B} {}},
  \end{align*}

  which again follows from the soundness of the congruence rules.

\end{proof}

\begin{lemma}[RaiseEval]\label{lem:raise_eval}
  \begin{mathpar}
    \inferrule*[]{}
    {\raiseOpwithM \effname M \equiv \letXbeboundtoYinZ M x \raiseOpwithM \effname x}
  \end{mathpar}
\end{lemma}
\begin{proof}

  We show one direction of the equivalence; the other is symmetric.
  Let $j$ be arbitrary and let $(\gamma_1, \gamma_2) \in \simigrel {\Gamma} {j}$.
  We need to show

  \[ 
    ( 
      (\raiseOpwithM \effname M)[\gamma_1], \,
      (\letXbeboundtoYinZ M x \raiseOpwithM \effname x)[\gamma_2]
    ) \in \simierel {\sigma} {j} {\simivrel {B} {}}.
  \]

  We apply Monadic Bind (Lemma \ref{lem:bind_general}), with 
  $E_1 = \raiseOpwithM {\effname} {\hole}$ and
  $E_2 = \letXbeboundtoYinZ {\hole} {x} {\raiseOpwithM \effname x}$.

  We first need to show that $(M[\gamma_1], M[\gamma_2]) \in \simierel {\tau} {j} {\simivrel {A} {}}$.
  This follows from reflexivity (Corollary \ref{cor:reflexivity}).

  Now let $k \le j$ and let $(V_1, V_2) \in \simivrel {A} {k}$. We need to show that

  \[ 
    ( 
      (\raiseOpwithM \effname {V_1})[\gamma_1], \,
      (\letXbeboundtoYinZ {V_2} x \raiseOpwithM \effname x)[\gamma_2]
    ) \in \simierel {\sigma} {k} {\simivrel {B} {}}.
  \]

  As $V_2$ is a value, the above let term steps. By anti-reduction, it suffices to show

  \[ 
    ( 
      (\raiseOpwithM \effname {V_1}), \,
      (\raiseOpwithM \effname {V_2})
    ) \in \simierel {\sigma} {k} {\simivrel {B} {}}.
  \]

  This follows from our assumption on $V_1$ and $V_2$ and
  the soundness of the term congruence rule for raise (Lemma \ref{lem:cong_raise}).

\end{proof}

\begin{lemma}[HandleEmpty]\label{lem:handle_empty}
  \begin{mathpar}
    \inferrule
    {}
    {\hndl M x N {\emptyset} \equiv \letXbeboundtoYinZ M x N}
  \end{mathpar}
\end{lemma}
\begin{proof}
  We show one direction of the equivalence; the other is symmetric.

  Let $j$ be arbitrary and let $(\gamma_1, \gamma_2) \in \simigrel {\Gamma} {j}$.
  We need to show

  \begin{align*}
    ( 
      &(\hndl M x N {\emptyset})[\gamma_1], \\
      &(\letXbeboundtoYinZ M x N)[\gamma_2]
    ) \\ & \quad \quad \in \simierel {\sigma} {j} {\simivrel {B} {}}.
  \end{align*}

  By Monadic Bind (Lemma \ref{lem:bind_general}) and the fact that 
  neither evaluation context catches any effects, it suffices to show that

  \begin{align*}
    ( 
      &\hndl {V_1} x {N[\gamma_1]} {\emptyset}, \\
      &\letXbeboundtoYinZ {V_2} x N[\gamma_2]
    ) \\ & \quad \quad \in \simierel {\sigma} {k} {\simivrel {B} {}},
  \end{align*}

  where $k \le j$ and $(V_1, V_2) \in \simivrel {B} {k}$.
  By anti-reduction, it will suffice to show that

  \begin{align*}
    ( 
      &{N[\gamma_1][V_1/x]},
        N[\gamma_2][V_2/x]
    ) \in \simierel {\sigma} {k} {\simivrel {B} {}}.
  \end{align*}

  Using ValSubst, the result follows by reflexivity and our assumption on $V_1$ and $V_2$.

\end{proof}

\begin{lemma}\label{lem:handle_ext}
  \begin{mathpar}
    \inferrule*[Right=HandleExt]
    {\forall\effname \in \dom(\phi).~ \psi(\effname) = \phi(\effname)
      \and
      \forall \effname \in \dom(\psi). \effname\not\in\dom(\phi) \Rightarrow 
        \psi(\effname) = k(\raiseOpwithM \effname x)
    }
    {\hndl M y N \phi \equiv \hndl M y N \psi : \compty \sigma B}
  \end{mathpar}
\end{lemma}
\begin{proof}

  We show one direction of the equivalence; the other is symmetric.

  The proof is by L\"{o}b induction. We assume that

  \[ 
    ( 
      (\hndl M_1' y N \phi)[\gamma_1], 
      (\hndl M_2' y N \psi)[\gamma_2]
    ) \in (\later \simierel {\sigma} {} {})_{j} (\simivrel {B} {}).
  \]

  for all $k \le j$, $(\gamma_1, \gamma_2) \in (\later \simigrel {\Gamma} {})_{k}$ and 
  $(M_1', M_2') \in (\later \simierel {\sigma} {} {})_{k} (\simivrel {A} {})$.

  Let $(\gamma_1, \gamma_2) \in \simigrel {\Gamma} {j}$.
  We need to show

  \[ 
    ( 
      (\hndl M_1 y N \phi)[\gamma_1], 
      (\hndl M_2 y N \psi)[\gamma_2]
    ) \in \simierel {\sigma} {j} {\simivrel {B} {}}
  \]

  for all $(M_1, M_2) \in \simierel {\sigma} {j} {\simivrel {A} {}}$.

  We apply Monadic Bind (Lemma \ref{lem:bind_general}).
  It suffices to consider the following cases:

  \begin{itemize}

    \item Let $k \le j$ and $(V_1, V_2) \in \simivrel {A} {k}$. We need to show that

    \[ 
      ( 
        (\hndl {V_1} y {N[\gamma_1]} {\phi[\gamma_1]}), 
        (\hndl {V_2} y {N[\gamma_2]} {\psi[\gamma_2]})
      ) \in \simierel {\sigma} {k} {\simivrel {B} {}}.
    \]
  
    This follows by anti-reduction and reflexivity.
  
    \item Let $k \le j$ and let $\effname \in \sigma$ be an effect caught by either handler,
    i.e., $\effname$ is in $\dom(\phi)$ or $\dom(\psi)$.
    Let $(V^l, V^r) \in (\later \simivrel {c_\effname} {})_{k}$, and let $E^l \apart \effname$
    and $E^r \apart \effname$ such that
    $(x^l.E^l[x^l], x^r.E^r[x^r]) \in
    (\later \simikrel{d_\effname}{})_{k} (\simierel {\sigma} {} {\simivrel {B} {}})$.

    We need to show

    \begin{align*} 
      ( 
        &(\hndl {E^l[\raiseOpwithM{\effname}{V^l}]} y {N[\gamma_1]} {\phi[\gamma_1]}), \\
        &(\hndl {E^r[\raiseOpwithM{\effname}{V^r}]} y {N[\gamma_2]} {\psi[\gamma_2]})
      ) \\ & \quad \quad \in \simierel {\sigma} {k} {\simivrel {B} {}}.
    \end{align*}

    If $\effname \in \dom(\phi)$, then by the premise, we have
    $\psi(\effname) = \phi(\effname)$, so both sides step, and it
    suffices by anti-reduction to show

    \begin{align*} 
      ( 
        & \phi(\effname)[\gamma_1][V^l/x][(\lambda z. \hndl {E^l[z]} y {N[\gamma_1]} {\phi[\gamma_1]}) / k], \\
        & \phi(\effname)[\gamma_2][V^r/x][(\lambda z. \hndl {E^r[z]} y {N[\gamma_2]} {\psi[\gamma_2]}) / k]
      ) \\ & \quad \quad \in (\later \simierel {\sigma} {} {})_{k} (\simivrel {B} {}).
    \end{align*}

    By ValSubst, it suffices to show that $(V^l, V^r) \in (\later \simivrel {c_\effname} {})_{k}$,
    which is true by assumption, and that

    \begin{align*} 
      ( 
        & (\lambda z. \hndl {E^l[z]} y {N[\gamma_1]} {\phi[\gamma_1]}), \\
        & (\lambda z. \hndl {E^r[z]} y {N[\gamma_2]} {\psi[\gamma_2]})
      ) \\ & \quad \quad \in (\later \simivrel {d_\effname \to_{\sigma} B} {})_{k}.
    \end{align*}

    By congruence for lambdas, it suffices to show that, given values
    $(V_1, V_2) \in (\later \simivrel {d_\effname} {})_{k}$, we have

    \begin{align*} 
      ( 
        & \hndl {E^l[V_1]} y {N[\gamma_1]} {\phi[\gamma_1]}, \\
        & \hndl {E^r[V_2]} y {N[\gamma_2]} {\psi[\gamma_2]}
      ) \\ & \quad \quad \in (\later \simierel {\sigma} {} {})_{k} (\simivrel {B} {}).
    \end{align*}

    This follows by the L\"{o}b induction hypothesis and our assumption on
    $E^l$ and $E^r$.    

    Now assume that $\effname \notin \dom(\phi)$. Then note that the first handle term
    does not step, while the second handle term steps to

    \[ \psi(\effname)[\gamma_2][V^r/x][(\lambda z. \hndl {E^r[z]} y {N[\gamma_2]} {\psi[\gamma_2]}) / k]. \]
    
    By the premise, we have $\psi(\effname) = k(\raiseOpwithM \effname x)$.
    Thus, by anti-reduction, it suffices to show

    \begin{align*} 
      ( 
        &(\hndl {E^l[\raiseOpwithM{\effname}{V^l}]} y {N[\gamma_1]} {\phi[\gamma_1]}), \\
        & (k(\raiseOpwithM \effname x)[\gamma_2)[V^r/x][(\lambda z. \hndl {E^r[z]} y {N[\gamma_2]} {\psi[\gamma_2]}) / k]
      ) \\ & \quad \quad \in \simierel {\sigma} {k} {\simivrel {B} {}}.
    \end{align*}

    That is, it will suffice to show

    \begin{align*} 
      ( 
        &(\hndl {E^l[\raiseOpwithM{\effname}{V^l}]} y {N[\gamma_1]} {\phi[\gamma_1]}), \\
        & ((\lambda z. \hndl {E^r[z]} y {N[\gamma_2]} {\psi[\gamma_2]})\, (\raiseOpwithM \effname V^r))
      ) \\ & \quad \quad \in \simierel {\sigma} {k} {\simivrel {B} {}}.
    \end{align*}

    Neither term steps, so it suffices to show they are related
    in $\simirrel {\sigma} {k} {\simivrel {B} {}}$.

    We need to show that $(V^l, V^r) \in (\later \simivrel {c_\effname} {})_{k}$,
    which is true by assumption, and that given $k' \le k$ and related values
    $(V_1, V_2) \in (\later \simivrel {d_\effname} {})_{k'}$, we have

    \begin{align*} 
      ( 
        &(\hndl {E^l[V_1]} y {N[\gamma_1]} {\phi[\gamma_1]}), \\
        & ((\lambda z. \hndl {E^r[z]} y {N[\gamma_2]} {\psi[\gamma_2]})\, V_2)
      ) \\ & \quad \quad \in (\later \simierel {\sigma} {} {})_{k'} (\simivrel {B} {}).
    \end{align*}

    By anti-reduction, it suffices to show

    \begin{align*} 
      ( 
        &(\hndl {E^l[V_1]} y {N[\gamma_1]} {\phi[\gamma_1]}), \\
        &(\hndl {E^r[V_2]} y {N[\gamma_2]} {\psi[\gamma_2]})
      ) \\ & \quad \quad \in (\later \simierel {\sigma} {} {})_{k'} (\simivrel {B} {}).
    \end{align*}

    This follows by the L\"{o}b induction hypothesis and our assumption on $E^l$ and $E^r$.




  \end{itemize}

\end{proof}

\subsubsection{Cast, Error, and Subtyping Properties}

\begin{lemma}[Err-bot]\label{lem:error_bot}
  \begin{mathpar}
    \inferrule
    {M : \compty {{d_\sigma}^r} {c^r}}
    {\err \ltdyn M : \compty {d_\sigma} c}
  \end{mathpar}
\end{lemma}
\begin{proof}
  Let $(\gamma_1, \gamma_2) \in \simigrel {\Gamma^\ltdyn} {j}$. We need to show

  \[
    (\err[\gamma_1], M[\gamma_2]) \in
      \simierel {d_\sigma} {j} {\simivrel {c} {}}.
  \]

  This follows from the definition of the logical relation:
  If $\sim$ is $<$ (counting steps on the left), then we are finished by
  the definition of the $\ltierel{}{}{}$ relation, because $\err \stepsin 0 \err$.

  If $\sim$ is $>$ (counting steps on the right), then we are similarly
  finished, because $M \stepsin 0 M$ and the left-hand term is $\err$.

\end{proof}

\begin{lemma}[Err-strict]\label{lem:error_strict}
  $E[\err] \equiv \err$
\end{lemma}
\begin{proof}
  We show one direction of the equivalence; the other is symmetric.
  Let $j$, $d_\sigma$, and $c$ be arbitrary. We need to show

  \[
    (E[\err], \err) \in \simierel {d_\sigma} {j} {\simivrel {c} {}}.
  \]

  By anti-reduction, it is sufficient to show

  \[
    (\err, \err) \in \simierel {d_\sigma} {j} {\simivrel {c} {}},
  \]

  which is easily seen to hold by definition of the logical relation.

\end{proof}

\begin{lemma}[Monotonicity of Subtyping]\label{lem:subty_mono}
  If $c \subty d$ then $\Vrel c \subseteq \Vrel d$

  Further, if $R \subseteq S$ then $\Krel d R \subseteq \Krel c S$,
  
  Further, if $c_\sigma \subty d_\sigma$ then both
  \begin{itemize}
  \item $\Erel{c} R \subseteq \Erel{d} S$
  \item $\Rrel{c} R \subseteq \Rrel{d} S$
  \end{itemize}
\end{lemma}
\begin{proof}
  By mutual induction on the subtyping proofs.
  First the type subtyping cases:
  \begin{enumerate}
  \item $\boolty \subty \boolty$: trivial.
  \item $c_i \to_{c_e} c_o \subty d_i to_{d_e} d_o$.  Assume
    $(V_f,V_f') \in \Vrel{c_i \to_{c_e} c_o}$, we need to show
    $(V_f,V_f') \in \Vrel{d_i \to_{d_e} d_o}$.  Let $(V_i,V_i') \in
    \Vrel{d_i}$. Then by inductive hypothesis, $(V_i, V_i')\in
    \Vrel{c_i}$. Therefore $(V_f V_i, V_f' V_i') \in \Erel{c_e}
    \Vrel{c_o}$ and the result follows by the two inductive
    hypotheses.
  \end{enumerate}
  The $\Krel{\cdot}$ case follows by a similar argument to the
  function case.

  The $\Erel{\cdot}$ case follows by inductive hypothesis.
  
  Next the $\Rrel{\cdot}$ cases:
  \begin{enumerate}
  \item $\dyn \subty \dyn$: trivial
  \item $\inferrule{c \subty \sig}{c \subty \dyn}$: trivial by definition of $\Rrel{\dyn}$
  \item $\inferrule{c \subty d}{c \subty Inj(d)}$: trivial by definition of $\Rrel{Inj(i,d)}$
  \item $\inferrule{c \subty d}{Inj(c) \subty Inj(d)}$: trivial by definition of $\Rrel{Inj(i,d)}$
  
  \item $\inferrule
    {\dom(d_c) \subseteq \dom(d'_c) \\\\
      \forall \effarr \effname c d \in d_c.
      \effarr \effname {c'} {d'} \in d'_c \wedge
      c \subty c' \wedge d' \subty d}
    {d_c \subty d'_c}$:
    Follows using L\"{o}b induction by the monotonicity of subtyping for the $\simivrel{\cdot}{}$ and
    $\simikrel{\cdot}{}$ relations.

  \end{enumerate}
\end{proof}

We next prove generalized versions of the cast properties ValUpL, ValUpR, ValDnL, ValDnR, EffUpL, EffUpR, EffDnL, EffDnR.
These are proved simultaneously by induction on the type precision derivation and by L\"{o}b-induction.

\begin{lemma}[ValUpR-general]\label{lem:ValUpR_general}
  \begin{mathpar}
  \inferrule
  {c : A \ltdyn A' \\\\ e : A' \ltdyn A''  \\\\
   \sg^\ltdyn \vDash_{d_\sigma} {M} \ltdyn {N} : c}
  {\sg^\ltdyn \vDash_{d_\sigma} {M} \ltdyn {\upcast {A'} {A''} N} : c \circ e}
  \end{mathpar}
\end{lemma}
\begin{proof}
  We need to show that

  \[ ({M}, {\upcast {A'} {A''} N}) \in 
    \simierel {d_\sigma} j {\simivrel {c \circ e} {}}. \]

  The proof is by induction on the precision derivation $e$.
  By monadic bind (Lemma \ref{lem:bind_general}), with $E_1 = \hole$ and 
  $E_2 = \upcast {A'} {A''} \hole$, it suffices to show

  \[ ({V_1}, {\upcast {A'} {A''} V_2}) \in 
    \simierel {d_\sigma} k {\simivrel {c \circ e} {}},
  \]

  where $k \le j$ and $(V_1, V_2) \in \simivrel {c} {k}$.
  We continue by cases on $e$.

  \begin{itemize}
    \item
    Case $e = \boolty$. We have $A = A' = A'' = \boolty$, and $c = \boolty$. Thus $c \circ e = \boolty$.
    
    Examining the operational semantics, we see that

    \[
      ({\upcast \boolty \boolty})({V_1}) \stepsin 1 {V_1}. 
    \]

    Thus, by anti-reduction, it suffices to show

    \[
      ({V_1}, {V_2}) \in 
        \simierel {d_\sigma} k {\simivrel {\boolty} {}}.
    \]

    This is true by assumption and Lemma \ref{lem:vals_in_V_implies_vals_in_E}.

    \item Case $e = e_i \to_{e_{\sigma}} e_o$. We have
    $A'  = A_i'  \to_{\sigma_A' } A_o'$ and
    $A'' = A_i'' \to_{\sigma_A''} A_o''$, and also
    $e_i : A_i' \ltdyn A_i''$ and 
    $e_o : A_o' \ltdyn A_o''$.

    By inversion, we see that $c = c_i \to_{c_\sigma} c_o$.
    Thus, we have that 
    $c \circ e = (c_i \to_{c_\sigma} c_o) \circ (e_i \to_{e_{\sigma}} e_o) = 
    (c_i \circ e_i) \to_{c_\sigma \circ e_\sigma} (c_o \circ e_o)$.

    We need to show that

    \[ 
      (
        {V_1}, 
        {\upcast {(A_i' \to_{\sigma_A'} A_o')} {(A_i'' \to_{\sigma_A''} A_o'')} V_2} 
      ) \in \simierel {d_\sigma} k {\simivrel {(c_i \circ e_i) \to_{c_\sigma \circ e_\sigma} (c_o \circ e_o)} {}}.
    \]

    As both terms are values, it suffices by Lemma \ref{lem:vals_in_V_implies_vals_in_E}
    to show they are related in ${\simivrel {(c_i \circ e_i) \to_{c_\sigma \circ e_\sigma} (c_o \circ e_o)} {k}}$.
    To this end, let $k' \le k$ and $(V^l, V^r) \in \simivrel {c_i \circ e_i} {k'}$.
    We need to show that

    \[
      ({V_1}\, V^l, 
       (\upcast {(A_i' \to_{\sigma_A'} A_o')} {(A_i'' \to_{\sigma_A''} A_o'')} V_2)\, V^r) \in 
      \simierel {c_\sigma \circ e_\sigma} {k'} {\simivrel {c_o \circ e_o} {}}.
    \]

    By anti-reduction, it suffices to show that

    \[
      ({V_1}\, V^l, 
       \upcast {A_o'} {A_o''} \upcast {\sigma_A'} {\sigma_A''} (V_2\, \dncast {A_i'} {A_i''} V^r) ) \in 
      \simierel {c_\sigma \circ e_\sigma} {k'} {\simivrel {c_o \circ e_o} {}}.
    \]

    By the induction hypothesis applied twice, it suffices to show

    \[
      ({V_1}\, V^l, 
       (V_2\, \dncast {A_i'} {A_i''} V^r) ) \in 
      \simierel {c_\sigma} {k'} {\simivrel {c_o} {}}.
    \]

    Finally, it suffices by the soundness of the term precision congruence rule for 
    function application (Lemma \ref{lem:cong_app} to show 
    that $(V_1, V_2) \in \simivrel {c_i \to_{c_\sigma} c_o} {k'}$, and that 
    
    \[ (V^l, \dncast {A_i'} {A_i''} V^r) \in \simierel {d_\sigma} {k'} {\simivrel {c} {}}. \]

    The former is true by our assumption on $V_1$ and $V_2$.
    The latter follows by the induction hypothesis and our assumption on $V^l$ and $V^r$. 

  \end{itemize}
  \end{proof}

\begin{lemma}[ValUpL-general]\label{lem:ValUpL_general}
  \begin{mathpar}
  \inferrule
  {\sg^\ltdyn \vdash_{d_\sigma} c : A \ltdyn A' \\\\
   \sg^\ltdyn \vdash_{d_\sigma} e : A' \ltdyn A'' \\\\
   \sg^\ltdyn \vdash_{d_\sigma} {M} \ltdyn {N} : c \circ e}
  {\sg^\ltdyn \vdash_{d_\sigma} {\upcast A {A'} M} \ltdyn {N} : e}
  \end{mathpar}
\end{lemma}
\begin{proof}
  Let $(\gamma_1, \gamma_2) \in \simigrel {\Gamma^\ltdyn} j$. We need to show that

  \[
    ( {\upcast A {A'} M}[\gamma_1], {N}[\gamma_2] ) \in 
      \simierel {d_\sigma} j {\simivrel e {}}.
  \]

  By monadic bind (Lemma \ref{lem:bind_general}), with 
  $E_1 = \upcast {A} {A'} \hole$ and 
  $E_2 = \hole$, it suffices to show

  \[ 
    ( {\upcast {A} {A'} V_1}, V_2 ) \in 
    \simierel {d_\sigma} j {\simivrel {e} {}},
  \]

  where $k \le j$ and $(V_1, V_2) \in \simivrel {c \circ e} {k}$.

  We continue by cases on $c$.
  The case $c = \boolty$ is similar to that in the previous lemma, so we skip to considering the case
  $c = c_i \to_{c_\sigma} c_o$. By inversion, we see that $e = e_i \to_{e_\sigma} e_o$.

  We have
  $A  = A_i  \to_{\hat{\sigma} } A_o$ and
  $A' = A_i' \to_{\hat{\sigma}'} A_o'$, and also
  Thus, we have that $c \circ e = (c_i \circ e_i) \to_{c_\sigma \circ e_\sigma} (c_o \circ e_o)$.

  We need to show that

  \[ 
    ( {\upcast {(A_i  \to_{\hat{\sigma} } A_o)} {(A_i' \to_{\hat{\sigma}'} A_o')} M}[\gamma_1],
      N[\gamma_2] ) \in 
    \simierel {d_\sigma} j {\simivrel {e_i \to_{e_\sigma} e_o} {}}.
  \]

  Similar to before, it suffices to show that these terms are related at $\simivrel {e_i \to_{e_\sigma} e_o} {k}$.
  This is similar to proof of the previous lemma, and hence omitted.

\end{proof}

\begin{lemma}[ValDnL-general]\label{lem:ValDnL_general}
  \begin{mathpar}
  \inferrule
  {\sg^\ltdyn \vdash_{d_\sigma} c : A \ltdyn A' \\\\    
   \sg^\ltdyn \vdash_{d_\sigma} e : A' \ltdyn A'' \\\\
   \sg^\ltdyn \vdash_{d_\sigma} M \ltdyn N : e}
 {\sg^\ltdyn \vdash_{d_\sigma} \dncast {A} {A'} M \ltdyn N : c \circ e}

  \end{mathpar}
\end{lemma}
\begin{proof}
  This proof is dual to the proof of ValUpR-general (Lemma \ref{lem:ValUpR_general}) and is hence omitted.
\end{proof}

\begin{lemma}[ValDnR-general]\label{lem:ValDnR_general}
  \begin{mathpar}
  \inferrule
  {\sg^\ltdyn \vdash_{d_\sigma} c : A \ltdyn A' \\\\
   \sg^\ltdyn \vdash_{d_\sigma} e : A' \ltdyn A'' \\\\
   \sg^\ltdyn \vdash_{d_\sigma} M \ltdyn N : c \circ e}
  {\sg^\ltdyn \vdash_{d_\sigma} M \ltdyn \dncast {A'} {A''} N : c}

  \end{mathpar}
\end{lemma}
\begin{proof}
  This proof is dual to the proof of ValUpL-general (Lemma \ref{lem:ValUpL_general}) and is hence omitted.
\end{proof}

\begin{lemma}[EffUpR-general]\label{lem:EffUpR_general}
  \begin{mathpar}
  \inferrule
  {d_\sigma  : \sigma  \ltdyn \sigma'  \\\\
   d_\sigma' : \sigma' \ltdyn \sigma'' \\\\
   \sg^\ltdyn \vdash_{d_\sigma} {M} \ltdyn {N} : c}
  {\sg^\ltdyn \vdash_{d_\sigma \circ d_\sigma'}
    {M} \ltdyn {\upcast {\sigma'} {\sigma''} N} : c}
  \end{mathpar}
\end{lemma}
\begin{proof}
  Let $(\gamma_1, \gamma_2) \in \simigrel {\Gamma^\ltdyn} j$. We need to show that

    \[
     \expandTermPrecisionDef
       {M}
       {{\upcast {\sigma'} {\sigma''} N}}
       {d_\sigma \circ d_\sigma'}{c}{j}.
    \]
    
    We prove this statement by L\"{o}b induction (Lemma \ref{lem:lob-induction}).
    That is, assume for all $k \le j$ and all
    $(M', N') \in (\later \simierel {d_\sigma} {} {})_k ({\simivrel {c} {}})$,
    we have

    \[ (M', \upcast {\sigma'} {\sigma''} N')
      \in (\later \simierel {d_\sigma \circ d_\sigma'} {} {})_{k} ({\simivrel {c} {}}).
    \]

    Let $(M, N) \in \simierel {d_\sigma} {j} {\simivrel {c} {}}$. We need to show
    
    \[ (M, \upcast {\sigma'} {\sigma''} N)
        \in \simierel {d_\sigma \circ d_\sigma'} {j} {\simivrel {c} {}}.
    \]

    We proceed by cases on $d_\sigma'$. The case $d_\sigma' = \dyn$ is immediate,
    so consider $d_\sigma' = \texttt{inj}(d_c)$,
    where $d_c : \sigma_c \ltdyn \sig \mid_{\textrm{supp}(\sigma_c)}$.
    In this case, we know that $\sigma'' = \dyn$. Furthermore, we have

    \[ d_\sigma \circ d_\sigma' = d_\sigma \circ (\texttt{inj} (d_c)) = \texttt{inj} (d_\sigma \circ d_c). \]

    Thus, we need to show

    \[
      \expandTermPrecisionDef
        {M}
        {{\upcast {\sigma'} {\dyn} N}}
        {\texttt{inj} (d_\sigma \circ d_c)}{c}{j}.
     \]

    By monadic bind (Lemma \ref{lem:bind_general}), it will suffice to consider the following cases:

    \begin{itemize}

    \item Let $k \le j$ and let $(V_1, V_2) \in \simivrel c k$. We need to show

    \[
      (V_1, {\upcast {\sigma'} {\dyn} V_2}) \in
      \simierel {\texttt{inj} (d_\sigma \circ d_c)} k {\simivrel c {}}.
    \]

    By anti-reduction, it suffices to show that

    \[
      (V_1, V_2) \in
      \simierel {\texttt{inj} (d_\sigma \circ d_c)} k {\simivrel c {}}.
    \]

    As $V_1$ and $V_2$ are values, it suffices by Lemma \ref{lem:vals_in_V_implies_vals_in_E} to
    show that $(V_1, V_2) \in \simivrel c k$, which is true by assumption.

    \item Let $k \le j$ and $\effname @ c_\effname \leadsto d_\effname \in d_\sigma$
    be an effect that is caught by ${\upcast {\sigma'} {\dyn} \hole}$.
    Let $(V^l, V^r) \in (\later \simivrel {c_\effname} {})_k$, and
    let $E^l \apart \effname$ and $E^r \apart \effname$ be evaluation contexts such that
    $(x^l.E^l[x^l], x^r.E^r[x^r]) \in 
    (\later \simikrel{d_\effname}{})_k 
      (\simierel{d_{\sigma}} {} {\simivrel c {}})$. We need to show that

    \begin{align*}
      ( &E^l[\raiseOpwithM{\effname}{V^l}], \\
        &\upcast {\sigma'} {\dyn}
          E^r[\raiseOpwithM{\effname}{V^r}] ) \in
        \simierel {\texttt{inj} (d_\sigma \circ d_c)} k {\simivrel c {}}.
    \end{align*}

    By anti-reduction, it suffices to show that

    \begin{align*}
      ( &E^l[\raiseOpwithM{\effname}{V^l}], \\
        &\letXbeboundtoYinZ
          {\dncast {d_\effname^r} {d_\effname^\dyn} 
            \raiseOpwithM
              \effname
              {\upcast {c_\effname^r} {c_\effname^\dyn} V^r}}
          {y}
          {\upcast {\sigma'} {\dyn} E^r[y]}
      ) \\ 
      &\quad\quad \in
      \simierel {\texttt{inj} (d_\sigma \circ d_c)} k {\simivrel c {}}.
    \end{align*}

    Let $V'^r$ be the term to which $\upcast {c_\effname^r} {c_\effname^\dyn} V^r$ steps.
    By anti-reduction, it suffices to show

    \begin{align*}
      ( &E^l[\raiseOpwithM{\effname}{V^l}], \\
        &\letXbeboundtoYinZ
          {\dncast {d_\effname^r} {d_\effname^\dyn} \raiseOpwithM{\effname}{V'^r}}
          {y}
          {\upcast {\sigma'} {\dyn} E^r[y]}
      ) \\ 
      &\quad\quad \in
      \simierel {\texttt{inj} (d_\sigma \circ d_c)} k {\simivrel c {}}.
    \end{align*}

    As neither term steps, it suffices to show they are related in
    $\simirrel {\texttt{inj} (d_\sigma \circ d_c)} k {\simivrel c {}}.$.
    To this end, we need to show (1) $(V^l, V'^r) \in (\later \simivrel {c_\effname \circ c'_\effname} {})_{k}$,
    and (2) given $k' \le k$ and $(V_1, V_2) \in (\later \simivrel {d_\effname \circ d'_\effname} {})_{k'}$, we
    have

    \begin{align*}
      ( &E^l[V_1], \\
        &\letXbeboundtoYinZ
          {\dncast {d_\effname^r} {d_\effname^\dyn} V_2}
          {y}
          {\upcast {\sigma'} {\dyn} E^r[y]}
      ) \\ 
      &\quad\quad \in
       (\later \simierel {\inj I {d_\sigma \circ d_c}} {} {})_{k'} ({\simivrel c {}}).
    \end{align*}

    To show (1), it suffices by forward reduction to show that
    $(V^l, \upcast {c_\effname^r} {c_\effname^\dyn} V^r) \in (\later \simivrel {c_\effname \circ c'_\effname} {})_{k}$.
    This follows inductively from ValUpR (which we are proving simultaneously and can
    therefore apply at smaller types), and our assumption on $V^l$ and $V^r$.

    To show (2), let $V_2'$ be the value to which
    $\dncast {d_\effname^r} {d_\effname^\dyn} V_2$ steps.
    It suffices by anti-reduction to show

    \begin{align*}
      ( &E^l[V_1],
         {\upcast {\sigma'} {\dyn} E^r[V_2']}
      ) \\ 
      &\quad\quad \in
       (\later \simierel {\inj I {d_\sigma \circ d_c}} {} {})_{k'} ({\simivrel c {}}).
    \end{align*}

    By the L\"{o}b induction hypothesis, it suffices to show that

    \begin{align*}
      ( &E^l[V_1], {E^r[V_2']} ) \\ 
      &\quad\quad \in
       (\later \simierel {d_\sigma} {} {})_{k'} ({\simivrel c {}}).
    \end{align*}

    By our assumption on $E^l$ and $E^r$, it suffices to show that
    $(V_1, V_2') \in (\later \simivrel {d_\effname} {})_{k'}$.
    By forward reduction, it suffices to show that

    \[
      (V_1, \dncast {d_\effname^r} {d_\effname^\dyn} V_2) \in
        (\later \simierel {d_\sigma} {} {})_{k'} (\simivrel {d_\effname} {}).
    \]

    Now inductively by ValDnR, it suffices to show 
    $(V_1, V_2) \in (\later \simivrel {d_\effname \circ d'_\effname} {})_{k'}$,
    which is our assumption.

    \vspace{4ex}

    The case where $d_\sigma'$ is a concrete effect precision derivation is similar to the above.

  \end{itemize}

\end{proof}

\begin{lemma}[EffUpL-general]\label{lem:EffUpL_general}
  \begin{mathpar}
  \inferrule
  {d_\sigma  : \sigma  \ltdyn \sigma'  \\\\
   d_\sigma' : \sigma' \ltdyn \sigma'' \\\\
   \sg^\ltdyn \vdash_{d_\sigma \circ d_\sigma'} {M} \ltdyn {N} : c}
  {\sg^\ltdyn \vdash_{d_\sigma'} \upcast \sigma {\sigma'} M \ltdyn N : c}
  \end{mathpar}
\end{lemma}
\begin{proof}
  This is proved similarly to the above.
\end{proof}

\begin{lemma}[EffDnL-general]\label{lem:EffDnL_general}
  \begin{mathpar}
  \inferrule
    {d_\sigma : \sigma \ltdyn \sigma' \\\\
     d_\sigma' : \sigma' \ltdyn \sigma'' \\\\
     \sg^\ltdyn \vdash_{d_\sigma'} M \ltdyn N : c}
    {\sg^\ltdyn \vdash_{d_{\sigma \circ d_\sigma'}} \dncast {\sigma} {\sigma'} M \ltdyn N : c}
  \end{mathpar} 
\end{lemma}
\begin{proof}
  We prove this by L\"{o}b induction (Lemma \ref{lem:lob-induction}).
  That is, assume for all $k \le j$ and all
  $(M', N') \in (\later \simierel {d_\sigma'} {} {})_k ({\simivrel {c} {}})$,
  we have

  \[ (\dncast {\sigma} {\sigma'} M', N')
    \in (\later \simierel {d_\sigma \circ d_\sigma'} {} {})_{k} ({\simivrel {c} {}}).
  \]

  Let $(\gamma_1, \gamma_2) \in \simigrel {\Gamma^\ltdyn} j$, and let
  $(M, N) \in \simierel {d_\sigma'} {j} {\simivrel {c} {}}$. We need to show
  
  \[ (\dncast {\sigma} {\sigma'} M, N)
      \in \simierel {d_\sigma \circ d_\sigma'} {j} {\simivrel {c} {}}.
  \]

  By monadic bind (Lemma \ref{lem:bind_general}) and the fact that effect casts are the
  identity on values, it will suffice to show the following:

  Let $k \le j$ and $\effname @ c_\effname \leadsto d_\effname \in d_\sigma'$
  be an effect that is caught by ${\dncast {\sigma} {\sigma'} \hole}$.
  Let $(V^l, V^r) \in (\later \simivrel {c_\effname} {})_k$, and
  let $E^l \apart \effname$ and $E^r \apart \effname$ be evaluation contexts such that
  $(x^l.E^l[x^l], x^r.E^r[x^r]) \in 
  (\later \simikrel{d_\effname}{})_k 
    (\simierel{d_\sigma'} {} {\simivrel c {}})$. We need to show that

    \begin{align*} 
      (
        & \dncast {\sigma} {\sigma'} 
          E^l[\raiseOpwithM{\effname}{V^l}], \\
        & E^r[\raiseOpwithM{\effname}{V^r}]N
      )
      \\ & \quad \quad \in \simierel {d_\sigma \circ d_\sigma'} {j} {\simivrel {c} {}}.
    \end{align*}

    Note that if $\effname \notin \sigma$, then the left hand side steps to $\err$,
    in which case we are finished by ErrBot (Lemma \ref{lem:error_bot}).
    Otherwise, the proof proceeds alalogously to EffUpR (Lemma \ref{lem:EffUpR_general}),
    with upcasts and downcasts interchanged.

\end{proof}

\begin{lemma}[EffDnR-general]\label{lem:EffDnR_general}
  \begin{mathpar}
  \inferrule
  {d_\sigma : \sigma \ltdyn \sigma' \\\\
   d_\sigma' : \sigma' \ltdyn \sigma'' \\\\
   \sg^\ltdyn \vdash_{d_\sigma \circ d_{\sigma'}} M \ltdyn N : c}
  {\sg^\ltdyn \vdash_{d_\sigma} M \ltdyn \dncast {\sigma'} {\sigma''} N : c}
  \end{mathpar}
\end{lemma}
\begin{proof}
  
  We prove this statement by L\"{o}b induction (Lemma \ref{lem:lob-induction}).
  That is, assume for all $k \le j$ and all
  $(M', N') \in (\later \simierel {d_\sigma \circ d_\sigma'} {} {})_k ({\simivrel {c} {}})$,
  we have

  \[ (M', \dncast {\sigma'} {\sigma''} N')
    \in (\later \simierel {d_\sigma} {} {})_{k} ({\simivrel {c} {}}).
  \]

  Let $(\gamma_1, \gamma_2) \in \simigrel {\Gamma^\ltdyn} j$, and let
  $(M, N) \in \simierel {d_\sigma \circ d_\sigma'} {j} {\simivrel {c} {}}$. We need to show
  
  \[ (M, \dncast {\sigma'} {\sigma''} N)
      \in \simierel {d_\sigma} {j} {\simivrel {c} {}}.
  \]

  By monadic bind (Lemma \ref{lem:bind_general}) and the fact that effect casts are the
  identity on values, it will suffice to show the following:

  Let $k \le j$ and $\effname @ c_\effname \leadsto d_\effname \in d_\sigma \circ d_\sigma'$
  be an effect that is caught by ${\dncast {\sigma'} {\sigma''} \hole}$.
  Let $(V^l, V^r) \in (\later \simivrel {c_\effname} {})_k$, and
  let $E^l \apart \effname$ and $E^r \apart \effname$ be evaluation contexts such that
  $(x^l.E^l[x^l], x^r.E^r[x^r]) \in 
  (\later \simikrel{d_\effname}{})_k 
    (\simierel{d_{\sigma} \circ d_\sigma'} {} {\simivrel c {}})$. We need to show that

  \begin{align*}
    ( &E^l[\raiseOpwithM{\effname}{V^l}], \\
      &\dncast {\sigma'} {\sigma''}
        E^r[\raiseOpwithM{\effname}{V^r}] ) \in
      \simierel {d_\sigma} k {\simivrel c {}}.
  \end{align*}

  First note that by Lemma \ref{lem:eff-precision-decomp}, there exist
  $c_1$, $c_2$, $d_1$, and $d_2$ such that
  $c_\effname = c_1 \circ c_2$ and $d_\effname = d_1 \circ d_2$ and
  $\effname @ c_1 \leadsto d_1 \in d_\sigma$ and $\effname @ c_2 \leadsto d_2 \in d_\sigma'$.
  In particular, this that $\effname \in \sigma'$, so the downcast from $\sigma''$ to $\sigma'$ does not fail.
  Let $c^L = c_1^l (= c_\effname^l)$, 
      $c^M = c_1^r = c_2^l$, 
  and $c^R = c_2^r (= c_\effname^r)$, and likewise define $d^L, d^M$ and $d^R$.

  By anti-reduction, it suffices to show that

  \begin{align*}
    ( &E^l[\raiseOpwithM{\effname}{V^l}], \\
      &\letXbeboundtoYinZ
        {\upcast {d^M} {d^R} 
          \raiseOpwithM
            \effname
            {\dncast {c^M} {c^R} V^r}}
        {y}
        {\dncast {\sigma'} {\sigma''} E^r[y]}
    ) \\ 
    &\quad\quad \in
    \simierel {d_\sigma} k {\simivrel c {}}.
  \end{align*}

  Let $V'^r$ be the term to which $\dncast {c^M} {c^R} V^r$ steps.
  By anti-reduction, it suffices to show

  \begin{align*}
    ( &E^l[\raiseOpwithM\effname{V^l}], \\
      &\letXbeboundtoYinZ
        {\upcast {d^M} {d^R} 
          \raiseOpwithM
            \effname
            {V'^r}}
        {y}
        {\dncast {\sigma'} {\sigma''} E^r[y]}
    ) \\ 
    &\quad\quad \in
    \simierel {d_\sigma} k {\simivrel c {}}.
  \end{align*}

  As neither term steps, it suffices to show they are related in
  $\simirrel {d_\sigma} k {\simivrel c {}}$.
  To this end, we need to show (1) $(V^l, V'^r) \in (\later \simivrel {c_1} {})_{k}$,
  and (2) given $k' \le k$ and $(V_1, V_2) \in (\later \simivrel {d_1} {})_{k'}$, we
  have

  \begin{align*}
    ( &E^l[V_1], \\
      &\letXbeboundtoYinZ
        {\upcast {d^M} {d^R} V_2}
        {y}
        {\dncast {\sigma'} {\sigma''} E^r[y]}
    ) \\ 
    &\quad\quad \in
     (\later \simierel {d_\sigma} {} {})_{k'} ({\simivrel c {}}).
  \end{align*}

  (1) follows from forward reduction and the inductive hypothesis for value types.
  To show (2), let $V_2'$ be the value to which
  $\upcast {d^M} {d^R} V_2$ steps.
  It suffices by anti-reduction to show

  \begin{align*}
    ( &E^l[V_1],
       {\upcast {\sigma'} {\sigma''} E^r[V_2']}
    ) \\ 
    &\quad\quad \in
     (\later \simierel {d_\sigma} {} {})_{k'} ({\simivrel c {}}).
  \end{align*}

  By the L\"{o}b induction hypothesis, it suffices to show that

  \begin{align*}
    ( &E^l[V_1], {E^r[V_2']} ) \\ 
    &\quad\quad \in
     (\later \simierel {d_\sigma \circ d_\sigma'} {} {})_{k'} ({\simivrel c {}}).
  \end{align*}

  By our assumption on $E^l$ and $E^r$, it suffices to show that
  $(V_1, V_2') \in (\later \simivrel {d_\effname} {})_{k'}$.
  By forward reduction, it suffices to show that

  \[
    (V_1, \upcast {d^M} {d^R} V_2) \in
      (\later \simierel {d_\sigma} {} {})_{k'} (\simivrel {d_\effname} {}).
  \]

  Now inductively by ValUpR, it suffices to show 
  $(V_1, V_2) \in (\later \simivrel {d_1} {})_{k'}$,
  which is our assumption.

  \vspace{4ex}

  The case where $d_\sigma'$ is a concrete effect precision derivation is similar to the above.
\end{proof}

\begin{lemma}[ValUpEval]\label{lem:ValUpEval}
  \begin{mathpar}
    \inferrule*[]{}{\upcast A B M \equiv \letXbeboundtoYinZ M x \upcast A B x\and}
  \end{mathpar}
\end{lemma}
\begin{proof}
  We show one direction of the equivalence; the other is symmetric.
  Let $j$ be arbitrary and let $(\gamma_1, \gamma_2) \in \simigrel {\Gamma} {j}$.
  We need to show

  \[ 
    ( 
      (\upcast A B M)[\gamma_1], \,
      (\letXbeboundtoYinZ M x \upcast A B x)[\gamma_2]
    ) \in \simierel {\sigma} {j} {\simivrel {B} {}}.
  \]

  By Monadic Bind (Lemma \ref{lem:bind_general}) and reflexivity, it will suffice to
  show that for all $k \le j$ let $(V_1, V_2) \in \simivrel {A} {k}$, we have

  \[ 
    ( 
      (\upcast A B V_1), \,
      (\letXbeboundtoYinZ {V_2} x \upcast A B x)
    ) \in \simierel {\sigma} {k} {\simivrel {B} {}}.
  \]

  By anti-reduction, it suffices to show

  \[ 
    ( 
      (\upcast A B V_1), \,
      (\upcast A B V_2)
    ) \in \simierel {\sigma} {k} {\simivrel {B} {}}.
  \]

  By congruence, it suffices to show

  \[ 
    ( 
      V_1, \, V_2
    ) \in \simierel {\sigma} {k} {\simivrel {B} {}}.
  \]

  This follows from our assumption on $V_1$ and $V_2$.

\end{proof}

\begin{lemma}[ValDnEval]\label{lem:ValDnEval}
  \begin{mathpar}
    \inferrule*[]{}{\dncast A B M \equiv \letXbeboundtoYinZ M x \dncast A B x}
  \end{mathpar}
\end{lemma}
\begin{proof} Dual to the above. \end{proof}

\begin{lemma}[cast-retraction]\label{lem:cast-retraction}
    
  let $A \ltdyn B$ and $\sigma \ltdyn \sigma'$, and let $c : A \ltdyn B$
  and $d_\sigma : \sigma \ltdyn \sigma'$.
  Let $\sg^\ltdyn \vdash_{\sigma} M \ltdyn N : A$.
  The following hold:

  \begin{enumerate}
      \item $\sg^\ltdyn \vDash_{\sigma} \dncast{A}{B} \upcast{A}{B} M \ltdyn N : A$
      \item $\sg^\ltdyn \vDash_{\sigma} \dncast{\sigma}{\sigma'} \upcast{\sigma}{\sigma'} M \ltdyn N : A$
  \end{enumerate}
\end{lemma}
\begin{proof}
  We prove stronger, ``pointwise" version of the above statemenets. Namely, we assume
  $(M, N) \in \simierel {\sigma} {j} {\simivrel {A} {}}$,
  and show, for example, that
  $(\dncast{A}{B} \upcast{A}{B} M, N) \in \simierel {\sigma} {j} {\simivrel {A} {}}$.

  The proof is by simultaneous induction on the derivations $c$ and $d_\sigma$.
  \begin{enumerate}
      \item Let $(\gamma_1, \gamma_2) \in \simigrel {A} {j}$.
      Suppose $(M, N) \in \simierel {\sigma} {j} {\simivrel {A} {}}$.
      We need to show
      
      \[ (\dncast{A}{B} \upcast{A}{B} M, N) 
        \in \simierel {\sigma} {j} {\simivrel {A} {}}. \]

      By monadic bind (Lemma \ref{lem:bind_general}), it suffices to show that

      \[ (\dncast{A}{B} \upcast{A}{B} V_1, V_2) \in \simierel {\sigma} {k} {\simivrel {A} {}}, \]

      where $k \le j$ and $(V_1, V_2) \in \simivrel {A} {k}$.

      We proceed by induction on the precision derivation $c$. If $c = \boolty$, then we need to show

      \[ (\dncast{\boolty}{\boolty} \upcast{\boolty}{\boolty} V_1, V_2) \in
          \simierel {\sigma} {k} {\simivrel {\boolty} {}}. \]

      According to the operational semantics, we have that

      \[ \dncast{\boolty}{\boolty} \upcast{\boolty}{\boolty} V_1 \stepsin 2 V_1. \]

      So by anti-reduction (Lemma \ref{lem:anti-reduction}), it suffices to show that
      $(V_1, V_2) \in \simierel {\sigma} {k} {\simivrel {\boolty} {}}$, which follows from our assumption.

      \vspace{2ex}

      If $c = c_i \to_{c_\sigma} c_o$, then $A = A_i \to_{\sigma_A} A_o$ and $B = B_i \to_{\sigma_B} B_o$.
      We need to show

      \begin{align*} (\dncast{(A_i \to_{\sigma_A} A_o)} {(B_i \to_{\sigma_B} B_o)} 
            &\upcast{(A_i \to_{\sigma_A} A_o)} {(B_i \to_{\sigma_B} B_o)} V_1, 
          V_2) \\
          &\quad \quad \in \simierel {\sigma} {k} {\simivrel {A_i \to_{\sigma_A} A_o} {}}.
      \end{align*}

      As both of these are values, it suffices to show that they are related in
      $\simivrel {A_i \to_{\sigma_A} A_o} {}$. To this end, let $k' \le k$ and let
      $(V^l, V^r) \in \simivrel {A_i} {k'}$. We need to show that

      \begin{align*} (
        &(\dncast{(A_i \to_{\sigma_A} A_o)} {(B_i \to_{\sigma_B} B_o)} 
            \upcast{(A_i \to_{\sigma_A} A_o)} {(B_i \to_{\sigma_B} B_o)} V_1)\, V^l, \\
        &V_2\, V^r)\\
        &\quad \quad \in \simierel {\sigma_A} {k'} {\simivrel {A_o} {}}.
      \end{align*}

      The former term steps, so by anti-reduction, it suffices to show that

      \begin{align*} 
        (
          &\dncast{A_o} {B_o} &\dncast {\sigma_A} {\sigma_B}
              ((\upcast{(A_i \to_{\sigma_A} A_o)} {(B_i \to_{\sigma_B} B_o)} V_1)\, \upcast {A_i} {B_i} V^l), 
          &V_2\, V^r
        )\\
        &\quad \quad \in \simierel {\sigma_A} {k'} {\simivrel {A_o} {}}.
      \end{align*}

      Let $V'^l$ be the value to which $\upcast {A_i} {B_i} V^l$ steps.
      By anti-reduction, it suffices to show that

      \begin{align*} 
        (
          &\dncast{A_o} {B_o} \dncast {\sigma_A} {\sigma_B} \\
          &\quad\quad  (\upcast {A_o} {B_o} \upcast {\sigma_A} {\sigma_B} \\
          &\quad\quad\quad\quad (V_1\, \dncast {A_i} {B_i} V'^l)), \\
          &V_2\, V^r
        )\\
        &\quad \quad \in \simierel {\sigma_A} {k'} {\simivrel {A_o} {}}.
      \end{align*}

      We will appeal to transitivity (Lemma \ref{lem:mixed-transitivity-terms}).
      We continue by cases on $\sim$. First assume $\sim$ is $<$.
      Let $V'^r$ be the value to which $\upcast {A_i} {B_i} V^r$ steps. If we show (1)

      \begin{align*} 
        (
          &\dncast{A_o} {B_o} \dncast {\sigma_A} {\sigma_B} \\
          &\quad\quad  (\upcast {A_o} {B_o} \upcast {\sigma_A} {\sigma_B} \\
          &\quad\quad\quad\quad (V_1\, \dncast {A_i} {B_i} V'^l)), \\
          &\dncast{A_o} {B_o} \upcast {A_o} {B_o} \\
          &\quad\quad  (\dncast {\sigma_A} {\sigma_B} \upcast {\sigma_A} {\sigma_B} \\
          &\quad\quad\quad\quad (V_2\, \dncast {A_i} {B_i} V'^r))
        )\\
      &\quad \quad \in \simierel {\sigma_A} {k'} {\simivrel {A_o} {}}.
      \end{align*}

      and (2)

      \begin{align*} 
        (
          &\dncast{A_o} {B_o} \upcast {A_o} {B_o} \\
          &\quad\quad  (\dncast {\sigma_A} {\sigma_B} \upcast {\sigma_A} {\sigma_B} \\
          &\quad\quad\quad\quad (V_2\, \dncast {A_i} {B_i} V'^r)), \\
          &V_2\, V^r
        )\\
      &\quad \quad \in \simierel {\sigma_A} {\omega} {\simivrel {A_o} {}},
      \end{align*}

      then we will be finished by transitivity.
      
      To show (1), first note that by monotonicity of casts (Lemma \ref{lem:cast_monotonicity}),
      it suffices to show that 

      \begin{align*} 
        (
          &\dncast {\sigma_A} {\sigma_B} \\
          &\quad\quad  (\upcast {A_o} {B_o} \upcast {\sigma_A} {\sigma_B} \\
          &\quad\quad\quad\quad (V_1\, \dncast {A_i} {B_i} V'^l)), \\
          &\upcast {A_o} {B_o} \\
          &\quad\quad  (\dncast {\sigma_A} {\sigma_B} \upcast {\sigma_A} {\sigma_B} \\
          &\quad\quad\quad\quad (V_2\, \dncast {A_i} {B_i} V'^r))
        )\\
      &\quad \quad \in \simierel {\sigma_A} {k'} {\simivrel {B_o} {}}.
      \end{align*}

      Then by commutativity of casts (Corollary \ref{cor:cast_commutativity}),
      it suffices to show

      \begin{align*} 
        (
          &\upcast {\sigma_A} {\sigma_B} (V_1\, \dncast {A_i} {B_i} V'^l), \\
          &\upcast {\sigma_A} {\sigma_B} (V_2\, \dncast {A_i} {B_i} V'^r)
        )\\
      &\quad \quad \in \simierel {\sigma_B} {k'} {\simivrel {A_o} {}}.
      \end{align*}

      By monotonicity of casts again, it suffices to show

      \begin{align*} 
        (
          &(V_1\, \dncast {A_i} {B_i} V'^l), \\
          &(V_2\, \dncast {A_i} {B_i} V'^r)
        )\\
      &\quad \quad \in \simierel {\sigma_A} {k'} {\simivrel {A_o} {}}.
      \end{align*}

      By soundness of the precision rule for function application, it suffices to
      show that $(V_1, V_2) \in \simivrel {(A_i \to_{\sigma_A} A_o)} {k}$ and that
      $(\dncast {A_i} {B_i} V'^l, \dncast {A_i} {B_i} V'^r) \in \simierel {\sigma_A} {k} {\simivrel {A_i} {}}$.
      The former holds by assumption, and to show the latter, it suffices by forward reduction to show
      $(\dncast {A_i} {B_i} \upcast {A_i} {B_i} V^l, \dncast {A_i} {B_i} \upcast {A_i} {B_i} V^r) \in 
      \simierel {\sigma_A} {k} {\simivrel {A_i} {}}$. This follows from the inductive hypothesis and
      assumption on $V^l$ and $V^r$.

      To show (2), it suffices by the inductive hypothesis applied twice to show

      \begin{align*} 
        (
          &(V_2\, \dncast {A_i} {B_i} V'^r)),
          V_2\, V^r
        )\\
      &\quad \quad \in \simierel {\sigma_A} {\omega} {\simivrel {A_o} {}},
      \end{align*}

      By forward reduction, it suffices to show

      \begin{align*} 
        (
          &(V_2\, \dncast {A_i} {B_i} V'^r)),
          V_2\, V^r
        )\\
      &\quad \quad \in \simierel {\sigma_A} {\omega} {\simivrel {A_o} {}},
      \end{align*}

      By soundness of function application, it suffices to show
      that $V_2$ is related to itself at $\simivrel {(A_i \to_{\sigma_A} A_o)} {\omega}$
      and that $(\dncast {A_i} {B_i} V'^r, V^r) \in \simierel {\sigma_A} {\omega} {\simivrel {A_i} {}}$.
      The former holds by reflexivity (Corollary \ref{cor:reflexivity}), and to show the
      latter it suffices by forward reduction to show that

      \[ 
        (\dncast {A_i} {B_i} \upcast {A_i} {B_i} V^r, V^r) \in
          \simierel {\sigma_A} {\omega} {\simivrel {A_i} {}},
      \]

      which follows by the inductive hypothesis and reflexivity.

      The case when $\sim$ is $<$ is analogous.

      \item Let $(\gamma_1, \gamma_2) \in \simigrel {A} {j}$.
      We use L\"{o}b induction. We assume that for all $k \le j$ and all related terms
      $(M', N') \in (\later \simierel {\sigma} {} {})_{k} (\simivrel {A} {})$, we have

      \[ (\dncast{\sigma}{\sigma'} \upcast{\sigma}{\sigma'} M', N') 
      \in (\later \simierel {\sigma} {} {})_{k} (\simivrel {A} {}). \]

      Let $(M, N) \in \simierel {\sigma} {j} {\simivrel {A} {}}$. We need to show that

      \[ (\dncast{\sigma}{\sigma'} \upcast{\sigma}{\sigma'} M, N) 
      \in \simierel {\sigma} {j} {\simivrel {A} {}}. \]

      By monadic bind (Lemma \ref{lem:bind_general}), it suffices to consdier 
      the following cases:

      \begin{enumerate}
        \item Let $k \le j$ and $(V_1, V_2) \in \simivrel {A} {k}$. We need to show
        
        \begin{align*}
          &(\dncast{\sigma}{\sigma'} \upcast{\sigma}{\sigma'} V_1, V_2) \in
            \simierel {\sigma} {k} {\simivrel {A} {}}
        \end{align*}

        This follows by anti-reduction and assumption.

        \item Let $k \le j$ and let
        $\effname @ C \leadsto D \in \sigma$.
        Let $C'$ and $D'$ be the types such that
        $\effname @ C' \leadsto D' \in \sigma'$.
        Let $(V^l, V^r) \in (\later \simivrel {C} {})_{k}$
        and let $E^l\apart \effname$ and $E^r \apart \effname$ be such that

        \[
          (x^l.E^l[x^l], x^r.E^r[x^r]) \in
            (\later \simikrel {D})_{k} (\simierel {\sigma} {} {\simivrel {A} {}}).
        \]

        We need to show that

        \begin{align*}
          (
            &\dncast{\sigma}{\sigma'} \upcast{\sigma}{\sigma'} E^l[\raiseOpwithM{\effname}{V^l}], 
            E^r[\raiseOpwithM{\effname}{V^r}]
          ) 
          \\ & \quad \quad \in \simierel {\sigma} {k} {\simivrel {A} {}}.
        \end{align*}

        The first term steps, so by anti-reduction it suffices to show

        \begin{align*}
          (
            &\dncast{\sigma}{\sigma'} (
              \letXbeboundtoYinZ
                {\dncast {D} {D'} \raiseOpwithM{\effname}{\upcast {C} {C'} V^l}}
                {x}
                {\upcast {\sigma} {\sigma'} E^l[x]}
            ), \\
            &E^r[\raiseOpwithM{\effname}{V^r}]
          ) 
          \\ & \quad \quad \in \simierel {\sigma} {k} {\simivrel {A} {}}.
        \end{align*}

        Let $V'^l$ be the value to which $\upcast {C} {C'} V^l$ steps.
        By anti-reduction, it suffices to show

        \begin{align*}
          (
            &\letXbeboundtoYinZ
              {\upcast {D} {D'} \raiseOpwithM{\effname}{\dncast {C} {C'} V'^l}}
              {y}
              {\dncast {\sigma} {\sigma'} 
                \letXbeboundtoYinZ
                  {\dncast {D} {D'} y}
                  {x}
                  {\upcast {\sigma} {\sigma'} E^l[x]}
              }, \\
            &E^r[\raiseOpwithM{\effname}{V^r}]
          ) 
          \\ & \quad \quad \in \simierel {\sigma} {k} {\simivrel {A} {}}.
        \end{align*}

        Let $y'$ be the value to which $\dncast {D} {D'} y$ steps.
        Let $V''^l$ be the value to which $\dncast {C} {C'} V'^l$ steps.

        By anti-reduction, it suffices to show

        \begin{align*}
          (
            &\letXbeboundtoYinZ
              {\upcast {D} {D'} \raiseOpwithM{\effname}{V''^l}}
              {y}
              {\dncast {\sigma} {\sigma'} 
                  {\upcast {\sigma} {\sigma'} E^l[y']}
              }, \\
            &E^r[\raiseOpwithM{\effname}{V^r}]
          ) 
          \\ & \quad \quad \in \simierel {\sigma} {k} {\simivrel {A} {}}.
        \end{align*}

        Neither term steps, so it suffices to show they are related in
        $\simirrel {\sigma} {k} {\simivrel {A} {}}$. To this end, we first
        show that $(V''^l, V^r) \in (\later \simivrel {C} {})_{k}$.
        By forward reduction, it suffices to show that
        $(\dncast {C} {C'} \upcast {C} {C'} V^l, V^r) \in (\later \simivrel {C} {})_{k}$.
        This follows from the inductive hypothesis for value types and our assumption on
        $V^l$ and $V^r$.

        We now show that, given $k' \le k$ and values
        $(V_1, V_2) \in (\later \simivrel {D} {})_{k'}$, we have

        \begin{align*}
          (
            &\letXbeboundtoYinZ
              {\upcast {D} {D'} V_1}
              {y}
              {\dncast {\sigma} {\sigma'} 
                  {\upcast {\sigma} {\sigma'} E^l[y']}
              }, \\
            &E^r[V_2]
          ) 
          \\ & \quad \quad \in (\later \simierel {\sigma} {} {})_{k} (\simivrel {A} {}).
        \end{align*}

        Let $V_1'$ be the value to which $\upcast {D} {D'} V_1$ steps.
        By anti-reduction, it will suffice to show

        \begin{align*}
          (
            &\letXbeboundtoYinZ
              {V_1'}
              {y}
              {\dncast {\sigma} {\sigma'} 
                  {\upcast {\sigma} {\sigma'} E^l[y']}
              }, \\
            &E^r[V_2]
          ) 
          \\ & \quad \quad \in (\later \simierel {\sigma} {} {})_{k} (\simivrel {A} {}).
        \end{align*}

        By forward reduction, it will suffice to show

        \begin{align*}
          (
            &\letXbeboundtoYinZ
              {V_1'}
              {y}
              {\dncast {\sigma} {\sigma'} 
                  {\upcast {\sigma} {\sigma'} E^l[\dncast {D} {D'} y]}
              }, \\
            &E^r[V_2]
          ) 
          \\ & \quad \quad \in (\later \simierel {\sigma} {} {})_{k} (\simivrel {A} {}).
        \end{align*}

        By anti-reduction, it will suffice to show

        \begin{align*}
          (
            &{\dncast {\sigma} {\sigma'} {\upcast {\sigma} {\sigma'} E^l[\dncast {D} {D'} V_1']}},
            E^r[V_2]
          ) 
          \\ & \quad \quad \in (\later \simierel {\sigma} {} {})_{k} (\simivrel {A} {}).
        \end{align*}

        By the L\"{o}b induction hypothesis, it suffices to show that

        \begin{align*}
          (
            &E^l[\dncast {D} {D'} V_1'], E^r[V_2]
          ) 
          \\ & \quad \quad \in (\later \simierel {\sigma} {} {})_{k} (\simivrel {A} {}).
        \end{align*}

        By forward reduction, it suffices to show

        \begin{align*}
          (
            &E^l[\dncast {D} {D'} \upcast {D} {D'} V_1], E^r[V_2]
          ) 
          \\ & \quad \quad \in (\later \simierel {\sigma} {} {})_{k} (\simivrel {A} {}).
        \end{align*}

        By the induction hypothesis for value types, it suffices to show

        \[
          ( E^l[V_1], E^r[V_2] ) \in (\later \simierel {\sigma} {} {})_{k} (\simivrel {A} {}).
        \]

        This follows by our assumption on $E^l$ and $E^r$.

      \end{enumerate}

  \end{enumerate}
\end{proof}


\begin{lemma}[Gradual subtyping]\label{lem:gradual_subty_non_admissible}
    
  Let $c : A \ltdyn B$ and $c' : A' \ltdyn B'$ where $A \subty A'$ and $B \subty B'$.
  Let $d_\sigma : \sigma_1 \ltdyn \sigma_2$ and $d_\sigma' : \sigma_1' \ltdyn \sigma_2'$ where
  $\sigma_1 \subty \sigma_1'$ and $\sigma_2 \subty \sigma_2'$.
  Suppose $M \equiv N$.
  The following hold:

  \begin{enumerate}
      \item 
      \begin{mathpar}
        \inferrule*[]
          {\sg^\ltdyn \vDash_{d_\tau} M \ltdyn N : A}
          {\sg^\ltdyn \vDash_{d_\tau} \upcast A B M \ltdyn \upcast {A'} {B'} N : B'}
      \end{mathpar}

      \item
      \begin{mathpar}
        \inferrule*[]
          {\sg^\ltdyn \vDash_{d_\tau} M \ltdyn N : B}
          {\sg^\ltdyn \vDash_{d_\tau} \dncast {A'} {B'} M \ltdyn \dncast A B N : A'}
      \end{mathpar}

      \item
      \begin{mathpar}
        \inferrule*[]
          {\sg^\ltdyn \vDash_{\sigma_1} M \ltdyn N : d}
          {\sg^\ltdyn \vDash_{\sigma_2'} \upcast {\sigma_1} {\sigma_2} M \ltdyn \upcast {\sigma_1'} {\sigma_2'} N : d}
      \end{mathpar}

      \item
      \begin{mathpar}
        \inferrule*[]
          {\sg^\ltdyn \vDash_{\sigma_2} M \ltdyn N : d}
          {\sg^\ltdyn \vDash_{\sigma_1'} \dncast {\sigma_1'} {\sigma_2'} M \ltdyn \dncast {\sigma_1} {\sigma_2} N : d}
      \end{mathpar}
  \end{enumerate}
\end{lemma}
\begin{proof}
  By simultaneous induction on the derivation $c' : A' \ltdyn B'$ and $d_\sigma' : \sigma_1' \ltdyn \sigma_2'$.
  \begin{enumerate}
      \item We need to show
      
      \[
          ( \upcast {A} {B} M, \upcast {A'} {B'} N ) \in
            \simierel {d_\sigma} {j} {\simivrel {B'} {}}.    
      \]

      By monadic bind (Lemma \ref{lem:bind_general}), with
      $E_1 = \upcast {A} {B} \hole$ and $E_2 = \upcast {A'} {B'} \hole$, it suffices to show the following.

      Let $k \le j$ and let $(V_1, V_2) \in \simivrel {A'} k$.
      We need to show

      \[ (\upcast {A} {B} V_1, \upcast {A'} {B'} V_2) \in 
        \simierel {d_\sigma} {k} {\simivrel {B'} {}}. \]

      We continue by cases on $c'$.

      Case $c' = \boolty$. Then by inversion on the rules for subtyping of precision derivations,
      we have $c = \boolty$.

      We need to show

      \[ (\upcast \boolty \boolty V_1, \upcast \boolty \boolty V_2) \in
        \simierel {d_\sigma} {k} {\simivrel {\boolty} {}}
      \]

      This follows by anti-reduction and our assumption on $V_1$ and $V_2$.

      Case $c' = c_i' \to_{c_\sigma'} c_o' : A_i' \to_{\sigma_A'} A_o' \ltdyn B_i' \to_{\sigma_B'} B_o'$.

      By inversion on the rules for subtyping for precision derivations, we have that
      $c = c_i \to_{c_\sigma} c_o$, where
      $c_i' \subty c_i$, and $c_\sigma \subty c_\sigma'$, and $c_o \subty c_o'$.

      Our assumption then becomes $(V_1, V_2) \in \simivrel {A_i' \to_{\sigma_A'} A_o'} k$.
      We need to show

      \[
          ( \upcast {A_i \to_{\sigma_A} A_o} {B_i \to_{\sigma_B} B_o} V_1, 
            \upcast {A_i' \to_{\sigma_A'} A_o'} {B_i' \to_{\sigma_B'} B_o'} V_2 ) \in 
            \simierel {d_\sigma} {k} {\simivrel {B_i' \to_{\sigma_B'} B_o'} {}}.
      \]

      Since both terms are values, it suffices to show they are related in
      $\simivrel {B_i' \to_{\sigma_B'} B_o'} {k}$.
      Let $k' \le k$ and let $(V^l, V^r) \in \simivrel {B_i'} {k'}$.
      We need to show

      \begin{align*}
        ( &(\upcast {A_i \to_{\sigma_A} A_o} {B_i \to_{\sigma_B} B_o} V_1)\, V^l, \\
          &(\upcast {A_i' \to_{\sigma_A'} A_o'} {B_i' \to_{\sigma_B'} B_o'} V_2)\, V^r ) \\
          & \quad \quad \in \simierel {\sigma_B'} {k'} {\simivrel {B_o'} {}}.
      \end{align*}

      By anti-reduction, it suffices to show

      \begin{align*}
        ( &\upcast {A_o} {B_o} \upcast {\sigma_A} {\sigma_B} (V_1\, \dncast {A_i} {B_i} V^l), \\
          &\upcast {A_o'} {B_o'} \upcast {\sigma_A'} {\sigma_B'} (V_2\, \dncast {A_i'} {B_i'} V^r) ) \\
          & \quad \quad \in \simierel {\sigma_B'} {k'} {\simivrel {B_o'} {}}.
      \end{align*}

      By the induction hypothesis applied twice, it suffices to show

      \begin{align*}
        ( &(V_1\, \dncast {A_i} {B_i} V^l),
          (V_2\, \dncast {A_i'} {B_i'} V^r) ) \\
          & \quad \quad \in \simierel {\sigma_A'} {k'} {\simivrel {A_o'} {}}.
      \end{align*}



      By soundness of the term precision congruence rule for function application (Lemma \ref{lem:cong_app}),
      it suffices to show that $(V_1, V_2) \in \simivrel {A_i' \to_{\sigma_A'} A_o'} {k'}$,
      and that
      
      \[
        (\dncast {A_i} {B_i} V^l, \dncast {A_i'} {B_i'} V^r) \in \simierel {d_\sigma} {k'} {\simivrel {A_i'} {}}.
      \]

      The former holds by assumption.
      To show the latter, it suffices by the admissible direction of gradual subtyping rule ValDnSub (item (2)
      in Lemma \ref{lem:gradual_subty_admissible}), whose proof does not depend on the present lemma,
      to show that $(V^l, V^r) \in \simivrel {B_i'} {k'}$. This is true by assumption.

      \item Similar to the above.
      
      \item We need to show
      
      \[ (\upcast {\sigma_1} {\sigma_2} M, \upcast {\sigma_1'} {\sigma_2'} N) \in
        \simierel {\sigma_2'} {j} {\simivrel {c} {}}. \]

        We use L\"{o}b induction. That is, we assume as our induction hypothesis that
        
        \[ 
          (
            \upcast {\sigma_1} {\sigma_2} M', 
            \upcast {\sigma_1'} {\sigma_2'} N'
          ) \in \later (\simierel {\sigma_2'} {})_{j} {(\simivrel {c} {})},
        \]

        for all $(M', N') \in (\later \simierel {\sigma_1'} {} {})_{j} (\simivrel {c} {})$,
        and we show that under this assumption, we have

        \[ 
          (
            \upcast {\sigma_1} {\sigma_2} M,
            \upcast {\sigma_1'} {\sigma_2'} N
          ) \in \simierel {\sigma_2'} {j} {\simivrel {c} {}}
        \]

        for all $(M, N) \in \simierel {\sigma_1'} {j} {\simivrel {c} {}}$.

        Using Monadic Bind (Lemma \ref{lem:bind_general}), we have the following cases:

        \begin{itemize}
          \item  Let $k \le j$ and $(V_1, V_2) \in \simivrel c k$.
          We need to show
  
          \[ (\upcast {\sigma_1} {\sigma_2} V_1, \upcast {\sigma_1'} {\sigma_2'} V_2) \in
          \simierel {\sigma_2'} {k} {\simivrel {c} {}}. \]
  
          This follows by anti-reduction and our assumption on $V_1$ and $V_2$.

          \item 
          
          Let $\effname @ c_i \leadsto d_i \in \sigma_1$ be an effect caught by
          $\upcast {\sigma_1'} {\sigma_2'} \hole$.
          Let $(V^l, V^r) \in (\later \simivrel {c_i^l} {})_k$, and let 
          $(E^l, E^r) \in (\later \simikrel {d_i^l} {})_k (\simierel {\sigma_1'} {} {\simivrel {c} {}})$.
          We need to show
  
          \begin{align*}
            (
              & \upcast {\sigma_1} {\sigma_2} E^l[\raiseOpwithM{\effname}{V^l}],
                \upcast {\sigma_1'} {\sigma_2'} E^r[\raiseOpwithM{\effname}{V^r}]
            ) \\ & \quad \quad \in \simierel {\sigma_2'} {k} {\simivrel {c} {}}.
          \end{align*}
  
          We continue by cases on subtyping of effect precision derivations.
          We show only the case $d_\sigma'$ is a concrete effect precision set $d_c'$;
          the other cases follow immediately or reduce to this one.
  
          By inversion, we have $d_\sigma$ is also a concrete effect precision set $d_c$ where
          $\dom(d_c) \subseteq \dom(d'_c)$ and
          for all $\effarr \effname c d \in d_c$,
          $\effarr \effname {c'} {d'} \in d'_c$ and  $c \subty c'$ and $d' \subty d$.
          By anti-reduction, it suffices to show

          \begin{align*}
            (
              & \letXbeboundtoYinZ 
                  {\dncast {d_i^l} {d_i^r} \raiseOpwithM{\effname}{\upcast {c_i^l} {c_i^r} V^l}} 
                  {x} {\upcast {\sigma_1} {\sigma_2} E^l[x]}, \\
              & \letXbeboundtoYinZ 
                  {\dncast {d_i'^l} {d_i'^r} \raiseOpwithM{\effname}{\upcast {c_i'^l} {c_i'^r} V^r}}
                  {x} {\upcast {\sigma_1'} {\sigma_2'} E^r[x]}
            ) \\ & \quad \quad \in (\later \simierel {\sigma_2'})_{k} ({\simivrel {c} {}}),
          \end{align*}

          By congruence for Let, it suffices to show (1)

          \begin{align*}
            (
              & {\dncast {d_i^l} {d_i^r} \raiseOpwithM{\effname}{\upcast {c_i^l} {c_i^r} V^l}}, \\
              & {\dncast {d_i'^l} {d_i'^r} \raiseOpwithM{\effname}{\upcast {c_i'^l} {c_i'^r} V^r}}
            ) \\ & \quad \quad \in (\later \simierel {\sigma_2'})_{k} ({\simivrel {c} {}}),
          \end{align*}

          and (2) for $(V_1, V_2) \in \later (\simivrel {d_i} {})_k$ we have

          \begin{align*}
            (
              & {\upcast {\sigma_1} {\sigma_2} E^l[V_1]}, \\
              & {\upcast {\sigma_1'} {\sigma_2'} E^r[V_2]}
            ) \\ & \quad \quad \in (\later \simierel {\sigma_2'})_{k} ({\simivrel {c} {}}),
          \end{align*}

          To show (1), first note that by the induction hypothesis for value types,

          \[
            (
              \raiseOpwithM{\effname}{\upcast {c_i^l} {c_i^r} V^l},
              \raiseOpwithM{\effname}{\upcast {c_i'^l} {c_i'^r} V^r}
            ) \in \simierel {\sigma_2'} {k} {\simivrel {c_i'} {}},
          \]

          and by the induction hypothesis for value types again, (1) follows.
          To show (2), note that $E^l[x^l]$ and $E^r[x^r]$ are related by
          assumption on $E^l$ and $E^r$. So we may apply the L\"{o}b induction hypothesis
          to reach the desired conclusion.
          
        \end{itemize}

      \item 
      
      We again use L\"{o}b induction and monadic bind. In the related raises case of the bind lemma,
      we let
      $\effname @ c_i \leadsto d_i \in \sigma_2$ be an effect caught by
      $\upcast {\sigma_1'} {\sigma_2'} \hole$. We let
      $(V^l, V^r) \in (\later \simivrel {c_i^l} {})_k$, and let 
      $(E^l, E^r) \in (\later \simikrel {d_i^l} {})_k (\simierel {\sigma_1'} {} {\simivrel {c} {}})$.

      We need to show

      \begin{align*}
        (
          & \upcast {\sigma_1} {\sigma_2} E^l[\raiseOpwithM{\effname}{V^l}],
            \upcast {\sigma_1'} {\sigma_2'} E^r[\raiseOpwithM{\effname}{V^r}]
        ) \\ & \quad \quad \in \simierel {\sigma_2'} {k} {\simivrel {c} {}}.
      \end{align*}

      If $\effname \notin \sigma_1$, then both sides step to $\err$. Since $\err$ is
      related to itself by ErrBot (Lemma \ref{lem:error_bot}), we are finished by anti-reduction.
      
      Otherwise, the proof proceeds analogously to that of the previous case, with
      upcasts and downcasts interchanged.

  \end{enumerate}
\end{proof}

\begin{lemma}[effect casts commute with pure function values]\label{lem:effect_casts_commute_fns}
    Let $E$ be an evaluation context such that 
    (1) for all $\sigma$, $\hastyGDRhoMT{\Gamma}{\holeRhoT{\sigma}{A}}{\sigma}{E}{B}$,
    and such that (2) $E \apart \effname$ for all $\effname \in \Sigma$.
    Furthermore, suppose that (3) for all values $V$, there exists a value $V'$ such that $E[V] \stepsin * V'$.

    
    Let $\sg^\ltdyn \vDash_{\sigma_2} M \equiv N : A$.

    Then $\sg^\ltdyn \vDash_{\sigma_1} E[\dncast {\sigma_1}{\sigma_2} M] \equiv 
                                     \dncast {\sigma_1} {\sigma_2} E[N] : B$,
    and likewise for upcasts. 
\end{lemma}
\begin{proof}
    We show the statement for downcasts only; the proof for upcasts is similar.
    Additionally, we show only one of the directions of the equivalence; the other is symmetric.
    
    We need to show

    \[ ( E[\dncast {\sigma_1}{\sigma_2} M], \dncast {\sigma_1} {\sigma_2} E[N] ) \in
    \simierel {\sigma_1} {j} {\simivrel {B} {}}. \]

    We apply monadic bind (Lemma \ref{lem:bind_general}) with $E_1 = E[\dncast {\sigma_1}{\sigma_2} \hole]$ and
    $E_2 = \dncast {\sigma_1} {\sigma_2} E$. By assumption on $M$ and $N$, will suffice to consider the following
    cases.
    \begin{itemize}
      \item Let $k \le j$ and let $(V_1, V_2) \in \simivrel {A} {k}$.
      We need to show
  
      \[ ( E[\dncast {\sigma_1}{\sigma_2} V_1], \dncast {\sigma_1} {\sigma_2} E[V_2] ) \in
      \simierel {\sigma_1} {k} {\simivrel {B} {}}. \]
  
      By the operational semantics, we have $E[\dncast {\sigma_1}{\sigma_2} V_1] \stepsin 1 E[V_1]$.
  
      By anti-reduction, it suffices to show
  
      \[ ( E[V_1], \dncast {\sigma_1} {\sigma_2} E[V_2] ) \in
      \simierel {\sigma_1} {k} {\simivrel {B} {}}. \]
  
      Furthermore, there exist $i_1$ and $i_2$ and values $V_1'$ and $V_2'$ such that
      $E[V_1] \stepsin {i_1} V_1'$ and
      $E[V_2] \stepsin {i_2} V_2'$.
  
      We also have $\dncast {\sigma_1} {\sigma_2} V_2' \stepsin 1 V_2'$.
  
      Putting the above facts together, by anti-reduction, it suffices to show
  
      \[ ( V_1', V_2' ) \in
      \simierel {\sigma_1} {k} {\simivrel {B} {}}. \]
  
      But by forward reduction, it suffices to show that 
      $(E[V_1], E[V_2]) \in \simierel {\sigma_1} {k} {\simivrel {B} {}}$.
  
      For this, it suffices (by the congruence lemmas) that $V_1$ and $V_2$ are related,
      which is true by assumption.

      \item Let $k \le j$ and let
      $\effname @ c^r \leadsto d^r \in \sigma_2$ be an effect caught by
      $\dncast {\sigma_1} {\sigma_2} \hole$.  
      Let $V^l, V^r, E^l\apart \effname, E^r \apart \effname$ be as in the statement of
      Lemma \ref{lem:bind_general}.
      We need to show
  
      \[ ( E[\dncast {\sigma_1}{\sigma_2} E^l[\raiseOpwithM{\effname}{V^l}]], 
           \dncast {\sigma_1} {\sigma_2} E[E^r[\raiseOpwithM{\effname}{V^r}]] ) \in
          \simierel {\sigma_1} {k} {\simivrel {B} {}}. 
      \]
  
      If $\effname \notin \sigma_1$, then, by the operational semantics,
      both terms will step to $\err$.
      By anti-reduction, it suffices to show that
      $(\err, \err) \in \simierel {\sigma_1} {k} {\simivrel {B} {}}$.
      This follows by ErrBot (Lemma \ref{lem:error_bot}).
  
  
      Now suppose $\effname @ c^l \leadsto d^l \in \sigma_1$.
      According to the operational semantics, we have
  
      \[ 
          E[\dncast {\sigma_1}{\sigma_2}
            E^l[\raiseOpwithM{\effname}{V^l}]] \stepsin 1
          E[E^l[\upcast {d^l} {d^r} 
            \raiseOpwithM
              \effname
              {\dncast {c^l} {c^r} V^l}]],
      \]
  
      and
  
      \[ \dncast {\sigma_1} {\sigma_2} 
           E[E^r[\raiseOpwithM\effname{V^r}]] \stepsin 1
         E[E^r[\upcast {d^l} {d^r}
           \raiseOpwithM
             \effname
             {\dncast {c^l} {c^r} V^r}]].
      \]
  
      Thus, by anti-reduction, it suffices to show
  
      \begin{align*}
          (
              &E[E^l[\upcast {d^l} {d^r} 
                \raiseOpwithM
                  \effname
                  {\dncast {c^l} {c^r} V^l}]], \\
              &E[E^r[\upcast {d^l} {d^r}
                \raiseOpwithM
                  \effname
                  {\dncast {c^l} {c^r} V^r}]]
          ) \\
          & \quad \quad \in \simierel {\sigma_1} {k} {\simivrel {B} {}}. 
      \end{align*}
  
      Let $V'^l$ be the value to which $\dncast {c^l} {c^r} V^l$ steps, and similarly
      let $V'^r$ be the value to which $\dncast {c^l} {c^r} V^r$ steps.
      By anti-reduction, it suffices to show
  
      \begin{align*}
          (
              &E[E^l[\upcast {d^l} {d^r} 
                \raiseOpwithM
                  \effname
                  {V'^l}]], \\
              &E[E^r[\upcast {d^l} {d^r}
                \raiseOpwithM
                  \effname
                  {V'^r}]]
          ) \\
          & \quad \quad \in \simierel {\sigma_1} {k} {\simivrel {B} {}}. 
      \end{align*}
  
      As neither term steps, it is sufficient to show that they are related in
      $\simirrel {\sigma_1} {k} {\simivrel {B} {}}$.
      We assert the second disjunct in the definition of $\simirrel{\cdot}{}{}$,
      taking $E^l = E[E^l[\upcast {d^l} {d^r} \hole]]$ and $E^r = E[E^r[\upcast {d^l} {d^r} \hole]]$.
      
      We first need to show that $(V'^l, V'^r) \in (\later \simivrel {c} {})_k$.
      By forward reduction, it suffices to show that
  
      \[ (\dncast {c^l} {c^r} V^l, \dncast {c^l} {c^r} V^r)
          \in (\later \simierel {\sigma_1} {} {})_{k} (\simivrel {c^r} {}). 
      \]
  
      By monotonicity of casts (lemma \ref{lem:cast_monotonicity}), it suffices to show
      $(V^l, V^r) \in (\later \simierel {\sigma_1} {} {})_{k} (\simivrel {c^r} {})$.
      This follows from our assumption about $V^l$ and $V^r$.
  
      We now need to show that
  
      \[ ( x^l.E[E^l[\upcast {d^l} {d^r} x^l]], x^r.E[E^r[\upcast {d^l} {d^r} x^r]] )
          \in (\later \simikrel {d} {} {})_{k} (\simierel {\sigma_1} {} {\simivrel {B} {}}).
      \]
  
      To this end, let $k' \le k$ and let $(V_1, V_2) \in (\later \simivrel {d^l})_{k'}$.
      We need to show
  
      \[
          ( E[E^l[\upcast {d^l} {d^r} V_1]], E[E^r[\upcast {d^l} {d^r} V_2]] ) \in
          (\later \simierel {\sigma_1} {})_{k'} (\simivrel {B} {}).
      \]
  
      It will suffice by the soundness of the congruence rules to show that
  
      \[ 
          ( E^l[\upcast {d^l} {d^r} V_1], E^r[\upcast {d^l} {d^r} V_2] ) \in
          (\later \simierel {\sigma_1} {})_{k'} (\simivrel {B} {}).
      \]
  
      Let $V_1'$ and $V_2'$ be the values to which $\upcast {d^l} {d^r} V_1$ and
      $\upcast {d^l} {d^r} V_2$ step, respectively. By anti-reduction, it suffices to show
  
      \[ 
          ( E^l[V_1'], E^r[V_2'] ) \in
          (\later \simierel {\sigma_1} {})_{k'} (\simivrel {B} {}).
      \]
  
      By assumption on $E^l$ and $E^r$, it suffices to show that
      $(V_1', V_2') \in (\later \simivrel {d^r} {})_{k'}$.
      By forward reduction, it suffices to show
  
      \[ 
          ( \upcast {d^l} {d^r} V_1, \upcast {d^l} {d^r} V_2 ) \in
          (\later \simierel {\sigma_1} {})_{k'} (\simivrel {B} {}).
      \]
  
      By monotonicity of casts (lemma \ref{lem:cast_monotonicity}), it suffices to show
  
      \[ 
          ( V_1, V_2 ) \in
          (\later \simierel {\sigma_1} {})_{k'} (\simivrel {B} {}).
      \]
  
      This follows from our assumption on $V_1$ and $V_2$.
  
    \end{itemize}

\end{proof}

\begin{corollary}[commutativity of casts]\label{cor:cast_commutativity}
Value casts commute with effect casts.
\end{corollary}
\begin{proof}
    This follows from \ref{lem:effect_casts_commute_fns}, because
    $\upcast{A}{B} \hole$ and $\dncast{A}{B} \hole$
    satisfy the requirements in the lemma.
\end{proof}

\begin{lemma}[functoriality of casts]\label{lem:cast_functoriality}
  Let $M$ be a term such that $\hastyDRhoMT{\cdot}{\sigma}{M}{A}$.
  Let $c : A \ltdyn B$ and $e : B \ltdyn C$.
  Let $d_\sigma : \sigma \ltdyn \sigma'$ and let $d_\sigma' : \sigma' \ltdyn \sigma''$
  
  Suppose $\sg^\ltdyn \vDash_{\sigma} M \equiv N : A$.
  Then the following hold:

  \textbf{Identity properties:}
  Suppose $\sg^\ltdyn \vDash_{\sigma} M \equidyn N : A$. We have
  \begin{enumerate}
      \item $\sg \vDash_{\sigma} \upcast A A M \equiv N : A$
      \item $\sg \vDash_{\sigma} \dncast A A M \equiv N : A$
      \item $\sg \vDash_{\sigma} \upcast {\sigma} {\sigma} M \equiv N : A$
      \item $\sg \vDash_{\sigma} \dncast {\sigma} {\sigma} M \equiv N : A$
  \end{enumerate}

  \textbf{Composition properties:}
  Let $c : A \ltdyn B$ and $e : B \ltdyn C$.
  Let $d_\sigma : \sigma \ltdyn \sigma'$ and $d_\sigma' : \sigma' \ltdyn \sigma''$.
  Suppose $M \equidyn N$. Then
  \begin{enumerate}
      \item $\sg \vDash_{\sigma} \upcast A C M \equidyn \upcast B C \upcast A B N : C$
      \item $\sg \vDash_{\sigma} \dncast A C M \equidyn \dncast A B \dncast B C N : A$
      \item $\sg \vDash_{\sigma''} \upcast {\sigma} {\sigma''} M \equidyn \upcast {\sigma'} {\sigma''} \upcast {\sigma} {\sigma'} N : A$
      \item $\sg \vDash_{\sigma} \dncast {\sigma} {\sigma''} M \equidyn \dncast {\sigma} {\sigma'} \dncast {\sigma'} {\sigma''} N : A$
  \end{enumerate}
\end{lemma}
\begin{proof}
  We prove more general, ``pointwise" versions of the above statements.
  For instance, we show that if $(M, N) \in \simierel {\sigma} {j} {\simivrel {A} {}}$,
  then $(\upcast A A M, N) \in \simierel {\sigma} {j} {\simivrel {A} {}}$.

  Additionally, we only prove one direction of each of the equivalences (i.e., $\ltdyn$);
  the proof of the other direction is symmetric.

  The statements are proven simultaneously by induction on $A$ and $\sigma$.
  \begin{itemize}
      \item Identity properties:
      \begin{enumerate}
          \item We need to show $(\upcast A A M, N) \in \simierel {\sigma} j {\simivrel A {}}$.
          By monadic bind (Lemma \ref{lem:bind_general}), with $E_1 = \upcast A A \hole$ and $E_2 = \hole$,
          it will suffice to show the following: Let $k \le j$ and $(V_1, V_2) \in \simivrel A k$. We will show

          \[ (\upcast A A V_1, V_2) \in \simierel {\sigma} k {\simivrel A {}}. \]

          We continue by induction on $A$. If $A = \boolty$, then we need to show

          \[ (\upcast \boolty \boolty V_1, V_2) \in \simierel {d_\sigma} k {\simivrel \boolty {}}. \]

          By anti-reduction, it suffices to show $(V_1, V_2) \in \simierel {\sigma} k {\simivrel {\boolty} {}}$,
          which follows from our assumption on $(V_1, V_2)$.

          If $A = A_i \to_{\sigma_A} A_o$, we need to show

          \[ (\upcast {(A_i \to_{\sigma_A} A_o)} {(A_i \to_{\sigma_A} A_o)} V_1, V_2) \in
            \simierel {\sigma} k {\simivrel {A_i \to_{\sigma_A} A_o} {}}. \]

          As both terms are values, it suffices to show they are related in
          $\simivrel {A_i \to_{\sigma_A} A_o} k$. So, let $k' \le k$ and let
           $(V^l, V^r) \in \simivrel {A_i} {k'}$. We need to show

           \[ ( (\upcast {(A_i \to_{\sigma_A} A_o)} {(A_i \to_{\sigma_A} A_o)} V_1)\, V^l,
                 V_2\, V^r ) \in
              \simierel {\sigma_A} {k'} {\simivrel {A_o} {}}. \]

          By anti-reduction, it suffices to show

          \[
              ( \upcast {A_o} {A_o} \upcast {\sigma_A} {\sigma_A} (V_1\, \dncast {A_i} {A_i} V^l),
              V_2\, V^r ) \in
                  \simierel {\sigma_A} {k'} {\simivrel {A_o} {}}.
          \]

          By the induction hypothesis (applied twice), it suffices to show

          \[
              ( (V_1\, \dncast {A_i} {A_i} V^l),
              V_2\, V^r ) \in
                  \simierel {\sigma_A} {k'} {\simivrel {A_o} {}}.
          \]

          By the soudness of function application, it suffices to show that
          $(V_1, V_2) \in \simivrel {A_i \to_{\sigma_A} A_o} {k'}$ and
          $(\dncast {A_i} {A_i} V^l, V^r) \in \simierel {\sigma_A} {k'} {\simivrel {A_i} {}}$.
          The former is true by assumption and downward closure ($k' \le k$).
          The latter is true by inductive hypothesis, since $V^l$ and $V^r$ are related.







          \item This is dual to the above.
      
          \item We prove this statement by L\"{o}b induction (Lemma \ref{lem:lob-induction}).
          That is, assume for all $(M', N') \in (\later \simierel {\sigma} {} {})_j ({\simivrel {A} {}})$, we have
          $(\upcast \sigma \sigma M', N') \in (\later \simierel {\sigma} {} {})_j ({\simivrel A {}})$.
          Let $(M, N) \in \simierel {\sigma} {j} {\simivrel {A} {}}$.
          We need to show $(\upcast \sigma \sigma M, N) \in \simierel {\sigma} j {\simivrel A {}}$.
          By monadic bind (Lemma \ref{lem:bind_general}), with $E_1 = \upcast \sigma \sigma \hole$ and $E_2 = \hole$,
          it will suffice to consider the following cases.

          \begin{itemize}
          \item Let $k \le j$ and let $(V_1, V_2) \in \simivrel A k$. We need to show

          \[ (\upcast \sigma \sigma V_1, V_2) \in \simierel {\sigma} k {\simivrel c {}}. \]

          Per the operational semantics, we have $\upcast \sigma \sigma V_1 \stepsin 1 V_1$, so by anti-reduction it
          suffices to show $(V_1, V_2) \in \simierel {\sigma} k {\simivrel A {}}$, which follows by the
          assumption that $(V_1, V_2) \in \simivrel A k$.

          \item Let $k \le j$ and let
          $\effname @ c_\effname \leadsto d_\effname$ be an effect caught by
          $\upcast \sigma \sigma \hole$ -- i.e., $\effname @ c_\effname \leadsto d_\effname \in \sigma$.
          Note that, as $\sigma$ is a reflexivity derivation, $c_\effname$ and $d_\effname$ are
          also reflexivity derivations, i.e., $c_\effname^l = c_\effname^r$
          and likewise for $d_\effname$. For simplicity, let $C = c_\effname^l$
          and $D = d_\effname^l$.

          Let $V^l, V^r, E^l\apart \effname, E^r \apart \effname$ be as in the statement of
          Lemma \ref{lem:bind_general}.
          We need to show

          \begin{align*} 
              ( & \upcast \sigma \sigma E^l[\raiseOpwithM\effname{V^l}],
               E^r[\raiseOpwithM\effname{V^r}] ) 
               \\ &\quad\quad \in \simierel {\sigma} k {\simivrel A {}}.
          \end{align*}

          According to the operational semantics, we have

          \begin{align*}
              &\upcast \sigma \sigma 
                E^l[\raiseOpwithM\effname{V^l}] \stepsin 1 \\
              &\quad\quad 
                  \letXbeboundtoYinZ 
                  {\dncast {D} {D} \raiseOpwithM \effname {\upcast {C} {C} V^l}}
                  {x}
                  {\upcast \sigma \sigma E^l[x]}
          \end{align*}

          So, by anti-reduction it suffices to show that

          \begin{align*}
              (
                  &\letXbeboundtoYinZ 
                  {\dncast {D} {D} \raiseOpwithM \effname {\upcast {C} {C} V^l}}
                  {x}
                  {\upcast \sigma \sigma E^l[x]}, \\
                  &E^r[\raiseOpwithM\effname{V^r}]
              ) \\
              & \quad \quad \in \simierel {\sigma} {k} {\simivrel A {}}.
          \end{align*}

          Let $V'^l$ be the term to which $\upcast {C} {C} V^l$ steps.
          By anti-reduction, it suffices to show that

          \begin{align*}
            (
                &\letXbeboundtoYinZ 
                {\dncast {D} {D} \raiseOpwithM \effname {V'^l}}
                {x}
                {\upcast \sigma \sigma E^l[x]}, \\
                &E^r[\raiseOpwithM \effname {V^r}]
            ) \\
            & \quad \quad \in \simierel {\sigma} {k} {\simivrel A {}}.
          \end{align*}

          The above terms do not step, so it suffices to show that they are related in
          $\simirrel {\sigma} {k} {\simivrel A {}}$.
          To this end, we will first show that
          $(V'^l, V^r) \in (\later \simivrel {C} {})_k$.
          By forward reduction, it suffices to show that 
          $(\upcast {C} {C} V^l, V^r) \in (\later \simierel {\simivrel {C} {}} {})_{k}$.
          By the induction hypothesis, it suffices to show that $(V^l, V^r) \in (\later \simivrel {C} {})_(k)$.

          Now we will show that

          \begin{align*}
              ( &x^l.
                (\letXbeboundtoYinZ 
                  {\dncast {D} {D} x^l}
                  {x}
                  {\upcast \sigma \sigma E^l[x]}), \, 
                x^r.E^r[x^r] )
              \\ &\quad\quad \in (\later \simikrel {D} {} {})_k (\simierel {\sigma} {} {\simivrel A {}}).
          \end{align*}

          Let $k' \le k$ and let $(V_1, V_2) \in (\later \simivrel {A} {})_{k'}$. We need to show

          \begin{align*}
              ( &(\letXbeboundtoYinZ 
                  {\dncast {D} {D} V_1}
                  {x}
                  {\upcast \sigma \sigma E^l[x]}), \,
               E^r[V_2]) 
              \\ &\quad\quad \in (\later \simierel {\sigma} {} {})_{k'} ({\simivrel {A} {}}).
          \end{align*}

          Let $V_1'$ be the value to which $\dncast {D} {D} V_1$ steps.
          By anti-reduction, it suffices to show

          \begin{align*}
              ( &(\letXbeboundtoYinZ 
                  {V_1'}
                  {x}
                  {\upcast \sigma \sigma E^l[x]}), \,
              E^r[V_2]) 
              \\ &\quad\quad \in (\later \simierel {\sigma} {} {})_{k'} ({\simivrel {A} {}}),
          \end{align*}

          and then since the let term steps, it suffices by anti-reduction again to show

          \begin{align*}
            ( & {\upcast \sigma \sigma E^l[V_1']}, \,
            E^r[V_2] ) 
            \in (\later \simierel {\sigma} {} {})_{k'} ({\simivrel {A} {}}),
          \end{align*}

          By the L\"{o}b induction hypothesis, it suffices to show that

          \[
            (E^l[V_1'], E^r[V_2]) \in (\later \simierel {\sigma} {} {})_{k'} ({\simivrel {A} {}})
          \]

          By our assumption on $E^l$ and $E^r$, it suffices to show
          
          \[ 
              (V_1', V_2) \in
               (\later \simivrel {A} {})_{k'}.
          \]

          By forward reduction, it suffices to show

          \[ 
              (\dncast {D} {D} V_1, V_2) \in
               (\later \simivrel {A} {})_{k'}. 
          \]

          By the induction hypothesis for value types, it suffices to show

          \[ 
              (V_1, V_2) \in
               (\later \simierel {A} {})_{k'}.
          \]

          This follows by assumption.    
            
          \end{itemize}
          
          \item We again use L\"{o}b induction and monadic bind.
              
          That is, assume for all $(M', N') \in (\later \simierel {\sigma} {} {})_j ({\simivrel {A} {}})$, we have
          $(\dncast \sigma \sigma M', N') \in (\later \simierel {\sigma} {} {})_j ({\simivrel A {}})$. We need to show

          \[ (\dncast \sigma \sigma M, N) \in \simierel {\sigma} j {\simivrel A {}} \]

          where $(M, N) \in \simierel {\sigma} {j} {\simivrel {A} {}}$.
          We again use monadic bind, and as in the previous proof, the case of related values
          follows trivially since effect casts are the identity on values.
          Thus, it will suffice to show the related raises case.
          That is, let $k \le j$ and let
          $\effname @ c_\effname \leadsto d_\effname$ be an effect caught by
          $\dncast \sigma \sigma \hole$ -- i.e., $\effname @ c_\effname \leadsto d_\effname \in \sigma$.
          As in the previous proof, since $\sigma$ is a reflexivity derivation, $c_\effname$ and $d_\effname$ are
          also reflexivity derivations, so for simplicity, let $C = c_\effname^l = c_\effname^r$
          and $D = d_\effname^l = d_\effname^r$.

          Let $V^l, V^r, E^l\apart \effname, E^r \apart \effname$ be as in the statement of
          the monadic bind lemma.
          We need to show

          \begin{align*} 
              ( & \dncast \sigma \sigma E^l[\raiseOpwithM\effname{V^l}],
               E^r[\raiseOpwithM \effname {V^r}] ) 
               \\ &\quad\quad \in \simierel {\sigma} k {\simivrel A {}}.
          \end{align*}

          Note that, since $\effname \in \sigma$, the downcast cannot fail.
          
          The remainder of the proof proceeds exactly like the previous proof,
          with upcasts and downcasts interchanged.

      \end{enumerate}

      \item Composition properties:
      \begin{enumerate}
          \item We need to show 
          $(\upcast A C M, \upcast B C \upcast A B N) \in \simierel {\sigma} {j} {\simivrel {C} {}}$.
          
          By monadic bind (Lemma \ref{lem:bind_general}) with $E_1 = \upcast A C \hole$ and $E_2 = \upcast B C \upcast A B \hole$,
          it will suffice to show the following: Let $k \le j$ and let $(V_1, V_2) \in \simivrel {A} {k}$.
          We will show

          \[ (\upcast A C V_1, \upcast B C \upcast A B V_2) \in \simierel {\sigma} {k} {\simivrel {C} {}}. \]
          
          If $c \circ e = \boolty$, then $c = e = \boolty$, and we need to show

          \[ (\upcast \boolty \boolty V_1, \upcast \boolty \boolty \upcast \boolty \boolty V_2) \in
            \simierel {\sigma} {k} {\simivrel {\boolty} {}}. \]

          By anti-reduction, it suffices to show $(V_1, V_2) \in \simivrel {\boolty} {j}$,
          which follows from our assumption.

          Now suppose $c \circ e = (c_i \circ e_i) \to_{(c_\sigma \circ e_\sigma)} (c_o \circ e_o)$.
          We need to show

          \begin{align*} ( 
               &\upcast {(A_i \to_{\sigma_A} A_o)} {(C_i \to_{\sigma_C} C_o)} V_1, \\ 
               &\upcast {(B_i \to_{\sigma_B} B_o)} {(C_i \to_{\sigma_C} C_o)} 
                   \upcast {(A_i \to_{\sigma_A} A_o)} {(B_i \to_{\sigma_B} B_o)} V_2 
            ) \\
            &\quad\quad\quad\quad \in \simierel {\sigma} {k} {\simivrel {C_i \to_{\sigma_C} C_o} {}}.
          \end{align*}

          Both terms are values, so it suffices to show that they are related in
          $\simivrel {C_i \to_{\sigma_C} C_o} {k}$.
          Let $k' \le k$ and let $(V^l, V^r) \in \simivrel {C_i} {k'}$.
          We need to show that

          \begin{align*} (
              &(\upcast {(A_i \to_{\sigma_A} A_o)} {(C_i \to_{\sigma_C} C_o)} V_1)\, V^l, \\ 
              &(\upcast {(B_i \to_{\sigma_B} B_o)} {(C_i \to_{\sigma_C} C_o)} 
                  \upcast {(A_i \to_{\sigma_A} A_o)} {(B_i \to_{\sigma_B} B_o)} V_2)\, V^r
           ) \\
           &\quad\quad\quad\quad \in \simierel {\sigma_C} {k'} {\simivrel {C_o} {}}.
          \end{align*}

          By anti-reduction, it suffices to show

          \begin{align*} (
              &\upcast {A_o} {C_o} \upcast {\sigma_A} {\sigma_C} (V_1\, \dncast {A_i} {C_i} V^l), \\ 
              &\upcast {B_o} {C_o} \upcast {\sigma_B} {\sigma_C} \\
              &\quad\quad ((\upcast {(A_i \to_{\sigma_A} A_o)} {(B_i \to_{\sigma_B} B_o)} V_2)\, \dncast {B_i} {C_i} V^r)
          ) \\
          &\quad\quad\quad\quad \in \simierel {\sigma_C} {k'} {\simivrel {C_o} {}}.
          \end{align*}

          Let $V'^r$ be the value to which $\dncast {B_i} {C_i} V^r$ steps.
          By anti-reduction, it suffices to show

          \begin{align*} (
              &\upcast {A_o} {C_o} \upcast {\sigma_A} {\sigma_C} (V_1\, \dncast {A_i} {C_i} V^l), \\ 
              &\upcast {B_o} {C_o} \upcast {\sigma_B} {\sigma_C} \\
              &\quad\quad ((\upcast {(A_i \to_{\sigma_A} A_o)} {(B_i \to_{\sigma_B} B_o)} V_2)\, V'^r)
          ) \\
          &\quad\quad\quad\quad \in \simierel {\sigma_C} {k'} {\simivrel {C_o} {}}.
          \end{align*}

          By anti-reduction again, it suffices to show

          \begin{align*} (
              &\upcast {A_o} {C_o} \upcast {\sigma_A} {\sigma_C} (V_1\, \dncast {A_i} {C_i} V^l), \\ 
              &\upcast {B_o} {C_o} \upcast {\sigma_B} {\sigma_C} \\
              &\quad\quad (\upcast {A_o} {B_o} \upcast {\sigma_A} {\sigma_B} (V_2\, \dncast {A_i} {B_i} V'^r))
          ) \\
          &\quad\quad\quad\quad \in \simierel {\sigma_C} {k'} {\simivrel {C_o} {}}.
          \end{align*}

          We will appeal to transitivity (Lemma \ref{lem:mixed-transitivity-terms}).
          We continue by cases on $\sim$.
          
          \begin{itemize}
          \item First suppose $\sim \, =\, <$.  We first claim that

          \begin{align*} (
            &\upcast {A_o} {C_o} \upcast {\sigma_A} {\sigma_C} (V_1\, \dncast {A_i} {C_i} V^l), \\ 
            &\upcast {B_o} {C_o} \upcast {A_o} {B_o} \\
            &\quad\quad (\upcast {\sigma_B} {\sigma_C} \upcast {\sigma_A} {\sigma_B} (V_2\, \dncast {A_i} {B_i} V'^r))
          ) \\
          &\quad\quad\quad\quad \in \simierel {\sigma_C} {k'} {\simivrel {C_o} {}}.
          \end{align*}

          By the induction hypothesis applied twice, it suffices to show

          \begin{align*} (
            &(V_1\, \dncast {A_i} {C_i} V^l), (V_2\, \dncast {A_i} {B_i} V'^r)
          ) \\
          &\quad\quad\quad\quad \in \simierel {\sigma_A} {k'} {\simivrel {A_o} {}}.
          \end{align*}

          By soundness of function application, it suffices to show that 
          $(V_1, V_2) \in \simivrel {A_i \to_{\sigma_A} A_o} {k'}$ and that 
          
          \[ (\dncast {A_i} {C_i} V^l, \dncast {A_i} {B_i} V'^r) 
          \in \simierel {\sigma} {k'} {\simivrel {A_i} {}}. \] 

          The former holds by assumption and downward closure. To show the latter, it suffices by 
          forward reduction to show that

          \[ (\dncast {A_i} {C_i} V^l, \dncast {A_i} {B_i} \dncast {B_i} {C_i} V^r) 
          \in \simierel {\sigma} {k'} {\simivrel {A_i} {}}. \] 
          
          Now, by the induction hypothesis, it suffices to show that

          \[ (V^l, V^r)
          \in \simierel {\sigma} {k'} {\simivrel {C_i} {}}, \]

          which follows from our assumption.

          Now by transitivity, it will suffice to show

          \begin{align*} (
              &\upcast {B_o} {C_o} \upcast {A_o} {B_o} \\
              &\quad\quad (\upcast {\sigma_B} {\sigma_C} \upcast {\sigma_A} {\sigma_B} (V_2\, \dncast {A_i} {B_i} V'^r)), \\
              &\upcast {B_o} {C_o} \upcast {\sigma_B} {\sigma_C} \\
              &\quad\quad (\upcast {A_o} {B_o} \upcast {\sigma_A} {\sigma_B} (V_2\, \dncast {A_i} {B_i} V'^r))
          ) \\
          &\quad\quad\quad\quad \in \gtierel {\sigma_C} {\omega} {\simivrel {C_o} {}}.
          \end{align*}

          By reflexivity (Corollary \ref{cor:reflexivity}), we have
          that $\upcast {A_o} {B_o} \upcast {\sigma_B} {\sigma_C} \upcast {\sigma_A} {\sigma_B} (V_2\, \dncast {A_i} {B_i} V'^r)$
          is related to itself.
          Then by commutativity of casts (Corollary \ref{cor:cast_commutativity}), we can interchange the order of
          $\upcast {A_o} {B_o}$ and $\upcast {\sigma_B} {\sigma_C}$, and the resulting terms are related.
          Finally by monotonicity of casts (Lemma \ref{lem:cast_monotonicity}), we can apply $\upcast {B_o} {C_o}$,
          and the resulting terms are still related. Moreover, all of these relations hold ``at $\omega"$.

          \vspace{4ex}

          \item Now suppose $\sim \, = \, >$. By similar reasoning as in the previous case, we have

          \begin{align*} (
            &\upcast {A_o} {C_o} \upcast {\sigma_A} {\sigma_C} (V_1\, \dncast {A_i} {C_i} V^l), \\ 
            &\upcast {B_o} {C_o} \upcast {A_o} {B_o} \\
            &\quad\quad (\upcast {\sigma_B} {\sigma_C} \upcast {\sigma_A} {\sigma_B} (V_2\, \dncast {A_i} {B_i} V^l))
          ) \\
          &\quad\quad\quad\quad \in \simierel {\sigma_C} {\omega} {\simivrel {C_o} {}}.
          \end{align*}

          Thus, by transitivity it will suffice to show

          \begin{align*} (
            &\upcast {B_o} {C_o} \upcast {A_o} {B_o} \\
            &\quad\quad (\upcast {\sigma_B} {\sigma_C} \upcast {\sigma_A} {\sigma_B} (V_2\, \dncast {A_i} {B_i} V^l)), \\
            &\upcast {B_o} {C_o} \upcast {\sigma_B} {\sigma_C} \\
            &\quad\quad (\upcast {A_o} {B_o} \upcast {\sigma_A} {\sigma_B} (V_2\, \dncast {A_i} {B_i} V'^r))
          ) \\
        &\quad\quad\quad\quad \in \gtierel {\sigma_C} {k'} {\simivrel {C_o} {}}.
        \end{align*}

        The reasoning is analogous to that of the previous case.

      \end{itemize}

      \item This is dual to the above.
      
      \item We prove this statement by L\"{o}b induction (Lemma \ref{lem:lob-induction}).
      That is, assume for all $(M', N') \in (\later \simierel {\sigma} {} {})_j ({\simivrel {A} {}})$, we have

      \[ (\upcast {\sigma} {\sigma''} M, \upcast {\sigma'} {\sigma''} \upcast {\sigma} {\sigma'} N)
        \in (\later \simierel {\sigma''} {} {})_{j} ({\simivrel {A} {}}).
      \]

      Let $(M, N) \in \simierel {\sigma} {j} {\simivrel {A} {}}$.
      We need to show
      
      \[ (\upcast {\sigma} {\sigma''} M, \upcast {\sigma'} {\sigma''} \upcast {\sigma} {\sigma'} N)
          \in \simierel {\sigma''} {j} {\simivrel {A} {}}.
      \]

      By monadic bind (Lemma \ref{lem:bind_general}), with 
      $E_1 = \upcast {\sigma} {\sigma''} \hole$
      and $E_2 = \upcast {\sigma'} {\sigma''} \upcast {\sigma} {\sigma'} \hole$,
      it suffices to consider the follwing cases:

      \begin{itemize}

      \item Let $k \le j$ and let $(V_1, V_2) \in \simivrel {A} {k}$. We need to show that
      
        \[ (\upcast {\sigma} {\sigma''} V_1, \upcast {\sigma'} {\sigma''} \upcast {\sigma} {\sigma'} V_2)
          \in \simierel {\sigma''} {j} {\simivrel {A} {}}.
        \]

      Since the effect cast is the identity on values, the above follows immediately
      by anti-reduction.

      \item Let $k \le j$ and let
      $\effname @ c_\effname \leadsto d_\effname \in \sigma$ be an effect caught by either $E_1$ or $E_2$.
      Note that, as $\sigma$ is a reflexivity derivation, $c_\effname$ and $d_\effname$ are
      also reflexivity derivations, i.e., $c_\effname^l = c_\effname^r$
      and likewise for $d_\effname$. For simplicity, let $C^L = c_\effname^l$
      and $D^L = d_\effname^l$.
      

      Let $V^l, V^r, E^l\apart \effname, E^r \apart \effname$ be as in the statement of
      the monadic bind lemma.
      We need to show

      \begin{align*} (
          &\upcast {\sigma} {\sigma''}
             E^l [\raiseOpwithM \effname {V^l}], \\
          &\upcast {\sigma'} {\sigma''} \upcast {\sigma} {\sigma'} 
             E^r[\raiseOpwithM \effname {V^r}]
          ) \\
          &\quad \quad \in \simierel {\sigma''} {k} {\simivrel {A} {}}.
      \end{align*}

      Let $C^M$ and $D^M$ be the types such that $\effname @ C^M \leadsto D^M \in \sigma'$
      Let $C^R$ and $D^R$ be the types such that $\effname @ C^R \leadsto D^R \in \sigma''$.
      By anti-reduction, it suffices to show


      \begin{align*}
        (
          &\letXbeboundtoYinZ
            {\dncast {D^L} {D^R} \raiseOpwithM\effname{\upcast {C^L} {C^R} V^l}}
            {x}
            {\upcast {\sigma} {\sigma''} E^l[x]}, \\
          &\upcast {\sigma'} {\sigma''} 
            (\letXbeboundtoYinZ
              {\dncast {D^L} {D^M} \raiseOpwithM\effname{\upcast {C^L} {C^M} V^r}}
              {x}
              {\upcast {\sigma} {\sigma'} E^r[x]})
        ) \\
        &\quad\quad \in \simierel {\sigma''} {k} {\simivrel {A} {}}.
      \end{align*}

      Let $V'^l$ be the value to which $\upcast {C^L} {C^R} V^l$ steps, say in $i$ steps.
      Let $V'^r$ be the value to which $\upcast {C^L} {C^M} V^r$ steps, say in $j$ steps.

      By anti-reduction, it suffices to show

      \begin{align*}
        (
          &\letXbeboundtoYinZ
            {\dncast {D^L} {D^R} \raiseOpwithM\effname{V'^l}}
            {x}
            {\upcast {\sigma} {\sigma''} E^l[x]}, \\
          &\upcast {\sigma'} {\sigma''} 
            (\letXbeboundtoYinZ
              {\dncast {D^L} {D^M} \raiseOpwithM\effname{V'^r}}
              {x}
              {\upcast {\sigma} {\sigma'} E^r[x]})
        ) \\
        &\quad\quad \in \simierel {\sigma''} {k} {\simivrel {A} {}}.
      \end{align*}

      Now (taking $E' = \letXbeboundtoYinZ{\dncast {D^L} {D^M} \hole}{x}{\upcast {\sigma} {\sigma'} E^r[x]}$
      in the \textsc{EffUpCast} rule), it will suffice by anti-reduction to show

      \begin{align*}
        (
          &\letXbeboundtoYinZ
            {\dncast {D^L} {D^R} \raiseOpwithM\effname{V'^l}}
            {x}
            {\upcast {\sigma} {\sigma''} E^l[x]}, \\
          &\letXbeboundtoYinZ
            {\dncast {D^M} {D^R} 
              \raiseOpwithM\effname{\upcast {C^M} {C^R} V'^r}}
            {y}
            {\\ & \quad \upcast {\sigma'} {\sigma''} 
              (\letXbeboundtoYinZ
                {\dncast {D^L} {D^M} y}
                {x}
                {\upcast {\sigma} {\sigma'} E^r[x]})
            }
        ) \\
        &\quad\quad \in \simierel {\sigma''} {k} {\simivrel {A} {}}.
      \end{align*}

      Let $V''^r$ be the value to which $\upcast {C^M} {C^R} V'^r$ steps.
      By anti-reduction, it suffices to show

      \begin{align*}
        (
          &\letXbeboundtoYinZ
            {\dncast {D^L} {D^R} \raiseOpwithM \effname {V'^l}}
            {x}
            {\upcast {\sigma} {\sigma''} E^l[x]}, \\
          &\letXbeboundtoYinZ
            {\dncast {D^M} {D^R} 
              \raiseOpwithM \effname {V''^r}}
            {y}
            {\\ & \quad \upcast {\sigma'} {\sigma''} 
              (\letXbeboundtoYinZ
                {\dncast {D^L} {D^M} y}
                {x}
                {\upcast {\sigma} {\sigma'} E^r[x]})
            }
        ) \\
        &\quad\quad \in \simierel {\sigma''} {k} {\simivrel {A} {}}.
      \end{align*}


      As neither term steps, we will show that they belong to $\simirrel {\sigma''} {k} {\simivrel {A} {}}$.
      We first need to show that

      \[
        (V'^l, V''^r) \in (\later \simivrel {A} {})_k.
      \]

      By forward-reduction, it suffices to show that 
      
      \[ 
        (\upcast {C^L} {C^R} V^l, \upcast {C^M} {C^R} \upcast {C^L} {C^M} V^r) \in 
          (\later \simivrel {A} {})_k.
      \]

      By the induction hypothesis for value types, it suffices to show that 
      $(V^l, V^r) \in (\later \simivrel {A} {})_k$, which is true by assumption.

      Now we need to show that, for all $k' \le k$ and related values
      $(V_1, V_2) \in (\later \simivrel {A} {})_{k'}$, we have

      \begin{align*}
        (
          &\letXbeboundtoYinZ
            {\dncast {D^L} {D^R} V_1}
            {x}
            {\upcast {\sigma} {\sigma''} E^l[x]}, \\
          &\letXbeboundtoYinZ
            {\dncast {D^M} {D^R} V_2}
            {y}
            {\\ & \quad \upcast {\sigma'} {\sigma''} 
              (\letXbeboundtoYinZ
                {\dncast {D^L} {D^M} y}
                {x}
                {\upcast {\sigma} {\sigma'} E^r[x]})
            }
        ) \\
        &\quad\quad \in (\later \simierel {\sigma''} {} {})_{k'} ({\simivrel {A} {}}).
      \end{align*}

      Let $V_1'$ and $V_2'$ be the values to which $\dncast {D^L} {D^R} V_1$
      and $\dncast {D^M} {D^R} V_2$ step, respectively. By anti-reduction, it
      will suffice to show

      \begin{align*}
        (
          &{\upcast {\sigma} {\sigma''} E^l[V_1']}, \\
          &{\upcast {\sigma'} {\sigma''} 
              (\letXbeboundtoYinZ
                {\dncast {D^L} {D^M} V_2'}
                {x}
                {\upcast {\sigma} {\sigma'} E^r[x]})
            }
        ) \\
        &\quad\quad \in (\later \simierel {\sigma''} {} {})_{k'} ({\simivrel {A} {}}).
      \end{align*}

      Let $V_2''$ be the value to which $\dncast {D^L} {D^M} V_2'$ steps. By anti-reduction,
      it will suffice to show

      \begin{align*}
        (
          &{\upcast {\sigma} {\sigma''} E^l[V_1']}, \\
          &{\upcast {\sigma'} {\sigma''} 
              ({\upcast {\sigma} {\sigma'} E^r[V_2'']})
            }
        ) \\
        &\quad\quad \in (\later \simierel {\sigma''} {} {})_{k'} ({\simivrel {A} {}}).
      \end{align*}

      Now by the L\"{o}b induction hypothesis, it suffices to show

      \begin{align*}
        (
          &{E^l[V_1']}, {E^r[V_2'']}
        )
        \in (\later \simierel {\sigma''} {} {})_{k'} ({\simivrel {A} {}}).
      \end{align*}

      By assumption on $E^l$ and $E^r$, it suffices to show

      \begin{align*}
        (
          &V_1', V_2''
        )
        \in (\later \simivrel {A} {})_{k'}.
      \end{align*}

      Now by forward reduction it suffices to show

      \begin{align*}
        (
          & \dncast {D^L} {D^R} V_1,
            \dncast {D^L} {D^M} \dncast {D^M} {D^R} V_2
        )
        \in (\later \simierel {\sigma''} {} {})_{k'} (\simivrel {A} {}).
      \end{align*}

      This follows by the inductive hypothesis for value types and our assumption on
      $V_1$ and $V_2$.

    \end{itemize}

      \item This is dual to the above: we use L\"{o}b induction and monadic bind,
      and we reach a point where we need to show
      
      \begin{align*} (
        &\dncast {\sigma} {\sigma''}
           E^l [\raiseOpwithM \effname {V^l}], \\
        &\dncast {\sigma'} {\sigma''} \dncast {\sigma} {\sigma'} 
           E^r[\raiseOpwithM \effname {V^r}]
        ) \\
        &\quad \quad \in \simierel {\sigma} {k} {\simivrel {A} {}}.
      \end{align*}

      where $\effname @ C^R \leadsto D^R \in \sigma''$.

      If $\effname \notin \sigma$, then the left-hand side steps to $\err$, as does the right-hand side.
      By ErrBot (Lemma \ref{lem:error_bot}), $\err$ is related to itself, so by anti-reduction, we are finished.
      If $\effname \notin \sigma'$, then in fact, $\effname \notin \sigma$ (since $\sigma \ltdyn \sigma'$),
      and so again, both sides step to $\err$.

      Otherwise, we proceed as in the proof of the previous case, with
      the upcasts and downcasts interchanged.

      \end{enumerate}
  \end{itemize}
\end{proof}

\begin{lemma}[monotonicity of casts]\label{lem:cast_monotonicity}
  Let $c : A \ltdyn B$, and $d_\sigma : \sigma \ltdyn \sigma'$, and let $M$ and $N$
  be terms such that $\sg^\ltdyn \vDash_{\sigma} M \ltdyn N : A$.
  The following hold:

  \begin{enumerate}
      \item $\sg^\ltdyn \vDash_{\sigma} \upcast {A} {B} M \ltdyn \upcast {A} {B} N : B$
      \item $\sg^\ltdyn \vDash_{\sigma} \dncast {A} {B} M \ltdyn \dncast {A} {B} N : A$
      \item $\sg^\ltdyn \vDash_{\sigma'} \upcast {\sigma} {\sigma'} M \ltdyn \upcast {\sigma} {\sigma'} N : A$
      \item $\sg^\ltdyn \vDash_{\sigma} \dncast {\sigma} {\sigma'} M \ltdyn \dncast {\sigma} {\sigma'} N : A$
  \end{enumerate}
\end{lemma}
\begin{proof}
  As in the proof of the functoriality properties of casts, we prove
  stronger, ``pointwise'' versions of the above statements, i.e., we assume
  $(M, N) \in \simierel {\sigma} {j} {\simivrel A {}}$, and show, for example, that
  $(\upcast A B M, \upcast A B N) \in \simierel {\sigma} {j} {\simivrel {B} {}}$.

  The proof is by induction on $c$ and $d_\sigma$.

  \begin{enumerate}
      \item We need to show
      
      \[ (\upcast A B M, \upcast A B N) \in \simierel {\sigma} {j} {\simivrel B {}}. \]

      By monadic bind (Lemma \ref{lem:bind_general}), with $E_1 = E_2 = \upcast A B \hole$,
      it will suffice to show that

      \[ (\upcast A B {V_1}, \upcast A B {V_2}) \in \simierel {\sigma} {k} {\simivrel B {}}, \]

      where $k \le j$ and let $(V_1, V_2) \in \simivrel A k$.
      
      If $c = \boolty$, then we need to show

      \[ (\upcast \boolty \boolty {V_1}, \upcast \boolty \boolty {V_2}) \in
        \simierel {\sigma} {k} {\simivrel \boolty {}}. \]

      By anti-reduction, it suffices to show that 
      $(V_1, V_2) \in \simierel {\sigma} {k} {\simivrel \boolty {}}$,
      which follows from our assumption.

      If $c = c_i \to_{c_\sigma} c_o$, then we need to show

      \begin{align*} (
          &\upcast {(A_i \to_{\sigma_A} A_o)} {(B_i \to_{\sigma_B} B_o)} {V_1}, \,
           \upcast {(A_i \to_{\sigma_A} A_o)} {(B_i \to_{\sigma_B} B_o)} {V_2}) 
          \\ & \quad \quad \in \simierel {\sigma} {k} {\simivrel {B_i \to_{\sigma_B} B_o} {}}.
      \end{align*}

      As both terms are values, it suffices to show that they are related in
      $\simivrel {B_i \to_{\sigma_B} B_o} k$.
      Let $k' \le k$ and let $(V^l, V^r) \in \simikrel {B_i} {k'}$.
      We need to show
      
      \begin{align*} 
        ( &(\upcast {(A_i \to_{\sigma_A} A_o)} {(B_i \to_{\sigma_B} B_o)} {V_1})\, V^l, \\
          &(\upcast {(A_i \to_{\sigma_A} A_o)} {(B_i \to_{\sigma_B} B_o)} {V_2})\, V^r ) 
          \\ & \quad \quad \in \simierel {\sigma_B} {k'} {\simivrel {B_o} {}}.
      \end{align*}

      By anti-reduction, it suffices to show

      \begin{align*}
        ( &\upcast {A_o} {B_o} \upcast {\sigma_A} {\sigma_B} (V_1\, \dncast {A_i} {B_i} V^l) , \\
          &\upcast {A_o} {B_o} \upcast {\sigma_A} {\sigma_B} (V_2\, \dncast {A_i} {B_i} V^r)
        ) \\ & \quad \quad \in \simierel {\sigma_B} {k'} {\simivrel {B_o} {}}.
      \end{align*}

      By the inductive hypothesis applied twice, it suffices to show

      \[
          ( (V_1\, \dncast {A_i} {B_i} V^l) ,
            (V_2\, \dncast {A_i} {B_i} V^r)
          ) \in \simierel {\sigma_A} {k'} {\simivrel {A_o} {}}.
      \]

      By soundness of function application, it suffices to show that $(V_1, V_2) \in \simivrel {A_i \to_{\sigma_A} A_o} {k'}$
      and that $(\dncast {A_i} {B_i} V^l, \dncast {A_i} {B_i} V^r) \in \simierel {\sigma} {k'} {\simivrel {A_i} {}}$.
      The former is true by our assumption about $V_1$ and $V_2$.
      To show the latter, it suffices by the inductive hypothesis to show that
      $(V^l, V^r) \in \simierel {\sigma} {k'} {\simivrel {B_i} {}}$, which follows by our assumption.

      \item This is dual to the above.
      \item This is dual to the below, and in fact easier since these are upcasts.
      
      \item We prove this statement by L\"{o}b induction (Lemma \ref{lem:lob-induction}).
      That is, assume for all 
      $(M', N') \in (\later \simierel {\sigma'} {} {})_j ({\simivrel {A} {}})$, we have

      \[
        (\dncast {\sigma} {\sigma'} M, \dncast {\sigma} {\sigma'} N) \in
          (\later \simierel {\sigma} {} {})_j ({\simivrel A {}}).
      \]

      Let $(M, N) \in \simierel {\sigma'} {j} {\simivrel {A} {}}$.
      We need to show
      
      \[ (\dncast {\sigma} {\sigma'} M, \dncast {\sigma} {\sigma'} N) \in
          \simierel {\sigma} {j} {\simivrel {A} {}}.
      \] 

      By monadic bind (Lemma \ref{lem:bind_general}), it will suffice to consider
      the following two cases:
      
     \begin{itemize}

      \item Let $k \le j$ and let $(V_1, V_2) \in \simivrel {A} {k}$. We need to show that
      
      \[ 
        (\dncast {\sigma} {\sigma'} V_1, \dncast {\sigma} {\sigma'} V_2) \in
        \simierel {\sigma} {j} {\simivrel {A} {}}.
      \]

      Since the effect cast is the identity on values, the above follows immediately
      by anti-reduction.

      \item Let $k \le j$ and let
      $\effname @ c_\effname \leadsto d_\effname \in \sigma'$ be an effect caught by
      $\dncast {\sigma} {\sigma'} \hole$. Recalling that $\sigma'$ is shorthand for the
      reflexivity derivation $\sigma' \ltdyn \sigma'$, we have that $c_\effname$ and $d_\effname$
      are themselves reflexivity (type precision) derivations; for brevity, we
      refer to the types as $C$ and $D$.


      Let $(V^l, V^r) \in (\later \simivrel {C} {})_{k}$ and
      and let $E^l\apart \effname, E^r \apart \effname$ be such that

      \[ 
        (x^l.E^l[x^l], x^r.E^r[x^r]) \in
          (\later \simikrel {D} {})_{k} (\simierel {\sigma} {} {\simivrel {A} {}}).
      \]

      We need to show

      \begin{align*} 
        (&\dncast {\sigma} {\sigma'} E^l[\raiseOpwithM \effname {V^l}], \\
         &\dncast {\sigma} {\sigma'} E^r[\raiseOpwithM \effname {V^r}])
         \\ &\quad \quad \in \simierel {\sigma} {k} {\simivrel {A} {}}.
      \end{align*}

      First, if $\effname \notin \sigma$, then both sides step to $\err$, and we are finished
      by anti-reduction since $\err$ is related to itself by ErrBot (Lemma \ref{lem:error_bot}).

      Otherwise, by anti-reduction, it suffices to show

      \begin{align*}
          (
            &\letXbeboundtoYinZ
              {\upcast {D} {D} \raiseOpwithM \effname {\dncast {C} {C} V^l}}
              {x}
              {\dncast {\sigma} {\sigma'} E^l[x]}, \\
            &\letXbeboundtoYinZ
              {\upcast {D} {D} \raiseOpwithM \effname {\dncast {C} {C} V^r}}
              {x}
              {\dncast {\sigma} {\sigma'} E^r[x]}
          ) \\
          & \quad \quad \quad \quad \in (\later \simierel {\sigma} {} {})_{k} {(\simivrel {A} {})}.
      \end{align*}

      By the soundness of the term precision congruence rule for let, it suffices to show that (1)

      \begin{align*}
        (
            &{\upcast {D} {D} \raiseOpwithM\effname{\dncast {C} {C} V^l}}, \\
            &{\upcast {D} {D} \raiseOpwithM \effname {\dncast {C} {C} V^r}}
        ) \\
        & \quad \quad \in (\later \simierel {\sigma} {} {})_{k} ({\simivrel {A} {}}).
    \end{align*}

    and (2) for all related $(V_1, V_2) \in (\later \simivrel {A} {})$, we have

    \begin{align*}
      (
        &\dncast {\sigma} {\sigma'} E^l[V_1],
         \dncast {\sigma} {\sigma'} E^r[V_2]
      ) \\
      & \quad \quad \in \simierel {\sigma} {k} {\simivrel {A} {}}.
  \end{align*}

    \end{itemize}

  \end{enumerate}
\end{proof}

\subsubsection{Transitivity}\label{sec:transitivity}






We introduce the following notation. We define $(M_1, M_2) \in R_\omega$ to mean that
$(M_1, M_2) \in R_k$ for all natural numbers $k$.

We now state and prove a ``mixed transitivity'' lemma, in which we allow
one of the two relations in the assumption to occur at a ``proper" precision derivation,
while the other is constrained to occur at a reflexivity derivation.

\begin{lemma}[mixed transitivity, terms]\label{lem:mixed-transitivity-terms}
    If (1) $(M_1, M_2) \in \gtierel {\sigma} {\omega} {\gtivrel A {}}$ and
    (2) $(M_2, M_3) \in \gtierel {d_\sigma} {j} {\gtivrel c {}}$,
    then $(M_1, M_3) \in \gtierel {d_\sigma} j {\gtivrel c {}}$.

    Similarly, if $(M_1, M_2) \in \ltierel {d_\sigma} {j} {\ltivrel c {}}$ and
    $(M_2, M_3) \in \ltierel {\sigma} {\omega} {\ltivrel A {}}$,
    then $(M_1, M_3) \in \ltierel {d_\sigma} j {\ltivrel c {}}$.
\end{lemma}
\begin{proof}
    This is proved simultaneuously with the following two lemmas on transitivity
    for results and values.
    We prove the lemma for $\sim = >$; the other case is similar.

    The proof is by L\"{o}b-induction (Lemma \ref{lem:lob-induction}).
    That is, assume that for all $M_1', M_2'$, and $M_3'$,
    if $(M_1', M_2') \in (\later \gtierel {\sigma} {})_{\omega} (\gtivrel A {})$ and
    $(M_2', M_3') \in (\later \gtierel {d_\sigma} {})_{j} (\gtivrel c {})$,
    then $(M_1', M_3') \in (\later \gtierel {d_\sigma} {})_{j} (\gtivrel c {})$.

    We proceed by considering cases on the assumption that
    $(M_2, M_3) \in \gtierel {d_\sigma} {j} {\gtivrel c {}}$.

    In the first case, $M_3 \stepsin {j+1}$. Then we immediately have that
    $(M_1, M_3) \in \gtierel {d_\sigma} j {\gtivrel c {}}$, via the first disjunct.

    In the second case, there is $k \le j$ such that $M_3 \stepsin {j - k} \err$
    and $M_2 \stepsin s \err$, for some number of steps $s$.
    By assumption (1), we have that $(M_1, M_2) \in \gtierel {\sigma} {s} {\gtivrel A {}}$.
    By inversion, we see that the second disjunct must have been
    true (with $k = 0$). 
    This means in particular that $M_1 \stepsin * \err$.
    Thus, we may conclude using the second disjunct that 
    $(M_1, M_3) \in \gtierel {d_\sigma} j {\gtivrel c {}}$.

    In the third case, there is $k \le j$ and $N_3$ such that $M_3 \stepsin {j - k} N_3$,
    and $M_2 \stepsin s \err$, for some number of steps $s$.
    By similar reasoning to the previous case, we may conclude using the third
    disjunct that
    $(M_1, M_3) \in \gtierel {d_\sigma} j {\gtivrel c {}}$.

    Finally, in the fourth case, there exist $k \le j$ and 
    $(N_2, N_3) \in \gtirrel {d_\sigma} {k} {\gtivrel c {}}$
    such that $M_2 \stepsin s N_2$ for some $s$, and $M_3 \stepsin {j - k} N_3$.
    By assumption (1), we have that $(M_1, M_2) \in \gtierel {\sigma} {s + i} {\gtivrel A {}}$
    for all $i \in \mathbb{N}$.
    By inversion, we see that either the third or the fourth disjunct was true,
    with $k = i$ in both cases (notice that $(s + i) - i = s$,
    which is precisely the number of steps that $M_2$ takes to $N_2$).

    In the former case, we have $M_1 \stepsin * \err$ and we can then finish by
    asserting the third disjunct. In the latter case, there exists $N_1$ such that
    $M_1 \stepsin * N_1$ and $(N_1, N_2) \in \gtirrel {\sigma} {i} {\gtivrel {A} {}}$.
    Since $i$ is arbitrary, this tells us that
    $(N_1, N_2) \in \gtirrel {\sigma} {\omega} {\gtivrel {A} {}}$.
    To recap, we have
    $(N_1, N_2) \in \gtirrel {\sigma} {\omega} {\gtivrel {A} {}}$,
    and $(N_2, N_3) \in \gtirrel {d_\sigma} {k} {\gtivrel c {}}$, for some $k \le j$.
    We want to show that $(N_1, N_3) \in \gtirrel {d_\sigma} {k} {\gtivrel c {}}$.

    This follows from Lemma \ref{lem:mixed-transitivity-results}.

\end{proof}

\begin{lemma}[mixed transitivity, values]\label{lem:mixed-transitivity-vals}
    If $(V_1, V_2) \in {\gtivrel A {\omega}}$ and
    $(V_2, V_3) \in {\gtivrel c {j}}$,
    then $(V_1, V_3) {\gtivrel c {j}}$.
    
    Similarly, if $(V_1, V_2) \in {\ltivrel c {j}}$ and
    $(V_2, V_3) \in {\ltivrel A {\omega}}$,
    then $(V_1, V_3) {\ltivrel c {j}}$.
\end{lemma}

\begin{proof}
    Proved simultaneously with the homogeneous transitivity for terms 
    (Lemma \ref{lem:mixed-transitivity-terms}) and for results (Lemma 
    \ref{lem:mixed-transitivity-results}).
    The proof is by induction on the type precision derivation $c$.
    We prove the first statement only; the other is proved similarly.

    \begin{itemize}
        \item Case $c = \boolty$. Then we have $V_1 = V_2 = V_3$ and either all are
        $\tru$, or all are $\fls$. In either case, $V_1$ is related to $V_3$.
        
        \item Case $c = c_i \to_{c_\sigma} c_o$. Then $A = A_i \to_{\sigma_A} A_o$
        and $B = B_i \to_{\sigma_B} B_o$.
    
        We have
        $(V_1, V_2) \in \simivrel {A_i \to_{\sigma_A} A_o} {\omega}$ and
        $(V_2, V_3) \in \simivrel {c_i \to_{c_\sigma} c_o} {k}$.
        
        We need to show 
        
        \[ (V_1, V_3) \in {\gtivrel {c_i \to_{c_\sigma} c_o} {j}}. \]

        Let $k \le j$ and let $(V^l, V^r) \in \gtivrel {c_i} k$.
        We need to show that

        \[ (V_1\, V^l, V_3\, V^r) \in \gtierel {c_\sigma} {k} {\gtivrel {c_o} {}}. \]

        By reflexivity (\ref{cor:reflexivity}), we know that
        $(V^l, V^l) \in \gtivrel {A_i} {\omega}$.

        From our assumption about $(V_1, V_2)$, it follows that

        \[ (V_1\, V^l, V_2\, V^l) \in \gtierel {\sigma_A} {\omega} {\gtivrel {A_o} {}}. \]

        From our assumption about $(V_2, V_3)$, we have
        
        \[ (V_2\, V^l, V_3\, V^r) \in \gtierel {c_\sigma} {k} {\gtivrel {c_o} {}}. \]
        
        Now we apply the induction hypothesis (Lemma \ref{lem:mixed-transitivity-terms})
        to conclude that

        \[ (V_1\, V^l, V_3\, V^r) \in \gtierel {c_\sigma} {k} {\gtivrel {c_o} {}}, \]

        as needed.

    \end{itemize}

\end{proof}

\begin{lemma}[mixed transitivity, results]\label{lem:mixed-transitivity-results}
    If (1) $(N_1, N_2) \in \gtirrel {\sigma} {\omega} {\gtivrel A {}}$ and
    (2) $(N_2, N_3) \in \gtirrel {d_\sigma} {j} {\gtivrel c {}}$,
    then $(N_1, N_3) \in \gtirrel {d_\sigma} j {\gtivrel c {}}$.

    Similarly, if $(N_1, N_2) \in \ltirrel {d_\sigma} {j} {\ltivrel c {}}$ and
    $(N_2, N_3) \in \ltirrel {\sigma} {\omega} {\ltivrel A {}}$,
    then $(N_1, N_3) \in \ltirrel {d_\sigma} j {\ltivrel c {}}$.
\end{lemma}
\begin{proof}
    We prove only the first statement; the second is analogous.

    Let $j$ be fixed. We consider cases on assumption (1).
    There are two subcases to consider. First, $N_1$ and $N_2$ are values
    and $(N_1, N_2) \in \gtivrel A {\omega}$.
    Then $N_3$ is also a value, and $(N_2, N_3) \in \gtivrel c {j}$.
    By \ref{lem:mixed-transitivity-vals}, we have that $(N_1, N_3) \in \gtivrel A {j}$.

    Otherwise, there exist $\effname @ C \leadsto D \in \sigma$, $E_1 \apart epsilon$ and 
    $E_2 \apart \effname$, and $V_1$ and $V_2$ such that
    $(V_1, V_2) \in (\later \gtivrel {C} {})_{\omega}$, and
    $(x_1.E_1[x_1], x_2.E_2[x_2]) \in 
      (\later \gtikrel {D} {} {})_{\omega} (\gtierel {\sigma} {} {\gtivrel {A} {}})$, and

    \[ 
        N_1 = E_1[\raiseOpwithM{\effname}{V_1}],
    \]

    and 

    \[ 
        N_2 = E_2[\raiseOpwithM{\effname}{V_2}].    
    \]

    Similarly, since $N_2$ and $N_3$ are related in $\gtirrel {d_\sigma} {j} {\gtivrel c {}}$,
    it follows that $\effname @ c_\effname \leadsto d_\effname \in d_\sigma$, where
    $c_\effname : C \ltdyn C'$ and $d_\effname : D \ltdyn D'$.
    We also know that there exist $E_3 \apart \effname$ and $V_3$ such that
    $(V_2, V_3) \in (\later \gtivrel {c_\effname} {})_{j}$, and
    $(x_2.E_2[x_2], x_3.E_3[x_3]) \in 
      (\later \gtikrel {d_\effname} {} {})_{j} (\gtierel {d_\sigma} {} {\gtivrel {c} {}})$, and 

    \[
        N_3 = E_3[\raiseOpwithM{\effname}{V_3}].
    \]

    Recall that we need to show

    \[ 
        ( E_1[\raiseOpwithM{\effname}{V_1}], E_3[\raiseOpwithM{\effname}{V_3}] ) \in
            \gtirrel {d_\sigma} {j} {\gtivrel c {}}.
    \]

    We assert the second disjunct in the definition of $\gtirrel{\cdot}{}{}$.
    
    We first claim that $(V_1, V_3) \in (\later \gtivrel {c_\effname} {})_{j}$.
    By transitivity for values (Lemma \ref{lem:mixed-transitivity-vals}), it suffices to show
    that $(V_1, V_2) \in (\later \gtivrel {c_\effname} {})_{\omega}$ and 
    $(V_2, V_3) \in (\later \gtivrel {c_\effname} {})_{j}$. These follow by assumption.

    Now we claim that 
    
    \[ 
        (x_1.E_1[x_1], x_3.E_3[x_3]) \in 
          (\later \gtikrel {d_\effname} {} {})_{j} (\gtierel {d_\sigma} {} {\gtivrel {c} {}}).
    \]

    Let $k \le j$ and let $(V^l, V^r) \in (\later \gtivrel {d_\effname} {})_{k}$. We need to show

    \[
        (E_1[V^l], E_3[V^r]) \in (\later \gtierel {d_\sigma} {} {})_{k} (\gtivrel {c} {}).
    \]

    By the induction hypothesis (recall we are proving this simultaneously with transitivity
    for terms, which is being proven by L\"{o}b induction), it suffices to find a term $M$ such that
    $(E_1[V^l], M) \in (\later \gtierel {\sigma} {} {})_{\omega} (\gtivrel {A} {})$, and 
    $(M, E_3[V_r]) \in (\later \gtierel {d_\sigma} {} {})_{k} (\gtivrel {c} {})$.

    By reflexivity (Corollary \ref{cor:reflexivity}), we have
    $(V^l, V^l) \in (\later \simivrel {} {})_{\omega}$.

    Then by our assumption on $(E_1, E_2)$, we have

    \[ 
        (E_1[V^l], E_2[V^l]) \in (\later \gtierel {\sigma} {} {})_{\omega} (\gtivrel {A} {})
    \]

    By our assumption on $(E_2, E_3)$ we have

    \[
        (E_2[V^l], E_3[V^r]) \in (\later \gtierel {d_\sigma} {} {})_{k} (\gtivrel {c} {}),
    \]

    which finishes the proof.
\end{proof}

\begin{lemma}[heterogeneous transitivity]\label{lem:heterogeneous-transitivity}
    Let $c : A_1 \ltdyn A_2$ and $e : A_2 \ltdyn A_3$.
    Let $d_\sigma : \sigma \ltdyn \sigma'$ and let $d_\sigma' : \sigma' \ltdyn \sigma''$.

    If (1) $(M_1, M_2) \in \gtierel {d_\sigma} {\omega} {\gtivrel c {}}$ and
    (2) $(M_2, M_3) \in \gtierel {d_\sigma'} {j} {\gtivrel e {}}$,
    then $(M_1, M_3) \in \gtierel {d_\sigma \circ d_\sigma'} j {\gtivrel {c \circ e} {}}$.

    Similarly, if $(M_1, M_2) \in \ltierel {d_\sigma} {j} {\ltivrel c {}}$ and
    $(M_2, M_3) \in \ltierel {d_\sigma'} {\omega} {\ltivrel e {}}$,
    then $(M_1, M_3) \in \ltierel {d_\sigma \circ d_\sigma'} j {\ltivrel {c \circ e} {}}$.
\end{lemma}
\begin{proof}
  Follows from mixed transitivity (Lemma {\ref{lem:mixed-transitivity-terms}}) and
  the generalized cast lemmas (Lemmas 
  \ref{lem:ValUpR_general}, \ref{lem:ValUpL_general},
  \ref{lem:ValDnL_general}, \ref{lem:ValDnR_general}, 
  \ref{lem:EffUpR_general}, \ref{lem:EffUpL_general},
  \ref{lem:EffDnL_general}, and \ref{lem:EffDnR_general}
  ).
  
  For example, by EffDnR and ValDnR, we have

  \[
    (M_1,\, \dncast {\sigma} {\sigma'} \dncast {A_1} {A_2} M_2) 
      \in \gtierel {\sigma} {\omega} {\gtivrel {A_1} {}},
  \]

  and by EffDnL and ValDnL, we have 

  \[
    (\dncast {\sigma} {\sigma'} \dncast {A_1} {A_2} M_2, \, M_3) 
      \in \gtierel {d_\sigma \circ d_\sigma'} {j} {\gtivrel {c \circ e} {}}.
  \]

  Then applying mixed transitivity, we have

  \[
    (M_1, M_3) \in \gtierel {d_\sigma \circ d_\sigma'} j {\gtivrel {c \circ e} {}},
  \]

  as desired.

\end{proof}

\end{document}